\documentclass[11pt]{article}

\usepackage[margin=1in]{geometry}
\usepackage[T1]{fontenc}
\usepackage[utf8]{inputenc}
\usepackage{microtype}
\usepackage{parskip}           

\usepackage{amsmath,amssymb,amsfonts,amsthm}
\usepackage{bm}
\usepackage{mathtools}

\usepackage{float}             
\usepackage{graphicx}
\usepackage{booktabs}
\usepackage{multirow}
\usepackage{threeparttable}
\usepackage{subcaption}
\usepackage{caption}
\usepackage{xcolor}
\usepackage{tikz}
\usepackage{pgfplots}
\usepgfplotslibrary{groupplots}
\pgfplotsset{compat=1.18}
\usepackage{pgfplotstable}
\usepackage{standalone}

\usepackage{algorithm}
\usepackage{algpseudocode}

\usepackage{listings}
\usepackage{courier}

\usepackage{pifont}
\usepackage{enumitem}
\usepackage{hyperref}
\hypersetup{
    colorlinks=true,
    linkcolor=blue!60!black,
    citecolor=blue!60!black,
    urlcolor=blue!60!black,
}

\usepackage[authoryear,round]{natbib}
\theoremstyle{plain}

\newtheorem{proposition}{Proposition}

\theoremstyle{definition}

\theoremstyle{remark}

\usetikzlibrary{shapes.geometric, arrows.meta, positioning, fit, backgrounds,
                calc, decorations.pathreplacing, patterns, shadows, chains,
                matrix, shapes.symbols, shapes.arrows, fadings}

\definecolor{susceptible}{RGB}{31,119,180}
\definecolor{exposed}{RGB}{255,127,14}
\definecolor{infectious}{RGB}{44,160,44}
\definecolor{recovered}{RGB}{214,39,40}
\definecolor{kernelblue}{RGB}{70,130,180}
\definecolor{kernelgreen}{RGB}{60,179,113}
\definecolor{kernelpurple}{RGB}{138,43,226}
\definecolor{markovblue}{RGB}{31,119,180}
\definecolor{renewgreen}{RGB}{44,160,44}
\definecolor{stateorange}{RGB}{230,136,28}
\definecolor{decisionyellow}{RGB}{240,210,60}
\definecolor{graphnode}{RGB}{31,119,180}
\definecolor{csrcell}{RGB}{230,136,28}
\definecolor{highlightcell}{RGB}{214,39,40}
\definecolor{accentblue}{RGB}{70,130,180}
\definecolor{threadblue}{RGB}{31,119,180}
\definecolor{regblue}{RGB}{70,130,180}
\definecolor{globalorange}{RGB}{230,136,28}
\definecolor{l2purple}{RGB}{138,43,226}
\definecolor{accentred}{RGB}{214,39,40}
\definecolor{hbmorange}{RGB}{230,136,28}
\definecolor{smgreen}{RGB}{44,160,44}
\definecolor{choice}{RGB}{240,210,60}
\definecolor{outcome}{RGB}{70,130,180}
\definecolor{markov}{RGB}{31,119,180}
\definecolor{renew}{RGB}{44,160,44}
\definecolor{codegreen}{rgb}{0,0.6,0}
\definecolor{codegray}{rgb}{0.5,0.5,0.5}
\definecolor{codepurple}{rgb}{0.58,0,0.82}
\definecolor{backcolour}{rgb}{0.95,0.95,0.92}

\lstdefinestyle{pythonstyle}{
    backgroundcolor=\color{backcolour},
    commentstyle=\color{codegreen},
    keywordstyle=\color{blue},
    numberstyle=\tiny\color{codegray},
    stringstyle=\color{codepurple},
    basicstyle=\ttfamily\scriptsize,
    breakatwhitespace=false,
    breaklines=true,
    captionpos=b,
    keepspaces=true,
    numbers=left,
    numbersep=5pt,
    showspaces=false,
    showstringspaces=false,
    showtabs=false,
    tabsize=2,
    frame=single,
    language=Python
}
\newcommand{\flashspread}{\textsc{FlashSpread}}
\newcommand{\bigO}{\mathcal{O}}

\title{%
  FlashSpread: IO-Aware GPU Simulation of Non-Markovian Epidemic
  Dynamics via Kernel Fusion
}

\author{%
  Heman Shakeri\thanks{Corresponding author:
  \href{mailto:hs9hd@virginia.edu}{hs9hd@virginia.edu}; ORCID
  0000-0002-9891-5748.}\,$^{,1}$
  \and Behnaz Moradi-Jamei$^{2}$
  \and Aram Vajdi$^{1}$
  \and Ehsan Ardjmand$^{3}$ \\[6pt]
  \footnotesize $^{1}$School of Data Science, University of Virginia,
  Charlottesville, VA 22904, USA \\
  \footnotesize $^{2}$Department of Mathematics \& Statistics,
  James Madison University, Harrisonburg, VA 22807, USA \\
  \footnotesize $^{3}$Ohio University, Athens, OH 45701, USA
}

\date{\today}

\begin{document}
\maketitle

\begin{abstract}
\noindent
Non-Markovian (renewal) epidemic simulation on multi-million-node
contact networks is essential for realistic forecasting under general
age-dependent holding-time distributions (log-normal, Weibull, Erlang,
and similar), but the age-dependent hazard forces dense per-step
updates that render the sparse event-queue strategies of standard CPU
methods ineffective. We present \flashspread{}, a GPU framework that
consolidates the per-step renewal pipeline (CSR traversal, numerically
stable $\mathrm{erfcx}$-based hazard evaluation, Bernoulli tau-leaping,
state transition, and next-step infectivity write-back) into a single
fused Triton kernel whose intermediates never leave streaming-
multiprocessor registers, with block-scalar skips that preserve CUDA
Graph capture and a degree-aware CSR dispatch (thread / warp /
edge-merge) that keeps the peak throughput on scale-free graphs. On an
NVIDIA A100 the fused CUDA-Graph engine reaches 8.09~Giga-NUPS at
$N = 10^6$ on a uniform-degree graph, a $217\times$ strict hardware
speedup over optimised CPU tau-leaping at the same $N$; on a
Barab\'{a}si--Albert graph of the same size the merge-based dispatch
recovers $4.5\times$ ($0.45 \to 2.0$ Giga-NUPS) over the default
kernel, and the framework scales to $N = 10^8$ on a single A100
(40~GB), with a mixed-precision storage path that extends the
L2-reachable scale by roughly $3\times$ and delivers a $1.15\times$
throughput lift at the far bandwidth-bound end. Validation against an
exact non-Markovian Gillespie reference shows a structural-bias floor
of $\sim$6\% on peak infection and $\sim$7\% on final attack rate that
does not detectably decrease as $\varepsilon \to 0$ across two decades
of tolerance, comfortably within typical epidemiological parameter
uncertainty. Code: \url{https://github.com/Shakeri-Lab/FlashSpread}.
\end{abstract}

\noindent\textbf{Keywords:} Stochastic simulation \textperiodcentered\
Non-Markovian renewal processes \textperiodcentered\ GPU computing
\textperiodcentered\ Epidemic modeling \textperiodcentered\ Tau-leaping
\textperiodcentered\ Complex networks \textperiodcentered\ Kernel
fusion \textperiodcentered\ CUDA Graphs

\bigskip\hrule\bigskip

\section{Introduction}\label{sec:introduction}

Stochastic spreading processes on networks (epidemics, information cascades, product adoption, malware propagation) pose fundamental modeling and control challenges \citep{pastor2015epidemic, kiss2017mathematics, barrat2008dynamical}. The underlying system is a high-dimensional continuous-time Markov chain whose state space grows exponentially with network size, rendering exact analytical solutions intractable for realistic systems \citep{anderson1991infectious, keeling2008modeling}.

Two distinct computational regimes arise depending on the temporal physics. In \emph{Markovian processes} (e.g., traditional SIS, SIR), transition rates depend only on current state and external control inputs; rates remain constant between events, enabling sparse incremental updates. In \emph{stochastic renewal processes} (e.g., SEIR with age-dependent incubation), transition rates depend on continuous holding times (ages) and change continuously even in the absence of events, requiring dense synchronous updates.

Markovian dynamics are the special case of renewal dynamics in which every holding-time distribution is exponential; throughout this paper we use \emph{Markovian} for that special case and \emph{non-Markovian (renewal)} for the general age-dependent case that motivates the framework. Computationally the two regimes require \emph{fundamentally opposite optimisation strategies} --- sparse event-driven updates for memoryless dynamics versus dense time-stepping for age-dependent hazards --- which motivates our dual-engine architecture.

The reliance on Markovian processes in epidemiological modeling, while mathematically convenient, imposes a ``memoryless'' assumption that contradicts the biological reality of pathogen transmission \citep{vajdi2023non}. Standard Markovian models assume dwell times follow exponential distributions, whose monotonically decreasing density implies a constant hazard rate: an infected individual is equally likely to recover at any moment, regardless of how long they have been infected. In contrast, empirical data from outbreaks such as COVID-19 demonstrate that critical epidemiological parameters (incubation period, serial interval, infectious duration) follow peaked, unimodal distributions (e.g., Weibull, log-normal, Erlang) \citep{sanche2020high, backer2020incubation, lauer2020incubation}, for which transition probability is negligible immediately after state entry and peaks after a characteristic delay. Neglecting this temporal structure leads to significant biases in predicting epidemic peaks and basic reproduction numbers, since distributions with identical means but differing variances yield substantially different trajectories \citep{vajdi2023non, sherborne2018mean, feng2019equivalence}. \citet{van2013non} showed that non-Markovian infection dynamics can dramatically alter epidemic thresholds compared to Markovian approximations, and for practical decision-making regarding contact tracing and quarantine the precise timing of infectivity relative to symptom onset is paramount \citep{vajdi2023non, he2020temporal}.

Despite these limitations in biological realism, Markovian models remain the most widely studied and analytically tractable class, with a mature literature on mean-field approximations \citep{keeling2005networks, kiss2017mathematics} and on control-theoretic formulations \citep{nowzari2016analysis, watkins2019robust, preciado2014optimal} that exploit memoryless dynamics. Nevertheless, a practical simulation framework must cover \emph{both} regimes efficiently, not just one.

We present \flashspread{}, a unified GPU framework that addresses this computational duality. Its contribution is primarily computational: a fused Triton renewal kernel that consolidates the entire per-step pipeline into a single launch while respecting CUDA Graph capture, and a degree-aware CSR dispatch that adapts the thread mapping to the graph's degree heterogeneity. To the best of our knowledge it is also the first open-source, end-to-end GPU framework for non-Markovian (renewal) spreading on networks with kernel-fused dense stepping. Recent exact event-driven CPU methods such as NEXT-Net \citep{cure2025fast} already reach $\sim$$10^6$ nodes per simulation; \flashspread{} is complementary, extending synchronous tau-leaping to $10^8$ on a single A100 with the per-step compute pipeline fused into a single GPU kernel. Our specific contributions are:

\begin{enumerate}
    \item \textbf{GPU-accelerated renewal epidemics at scale.} A first open-source, end-to-end GPU framework for non-Markovian network epidemic simulation with kernel-fused dense stepping, scaling to $N = 10^8$ nodes on a single A100 and achieving 8.09~Giga-NUPS on a uniform-degree graph at $N = 10^6$, a $217\times$ strict hardware speedup over the CPU tau-leaping baseline measured on the same network.

    \item \textbf{Fused renewal engine plus a companion Markovian engine.} The scientific core of this paper is the dense $\bigO(N/P)$ fused renewal engine; the Markovian engine is a secondary, companion contribution that the framework additionally provides. We identify that the two regimes demand opposite GPU strategies --- sparse $\bigO(K \cdot D_{\mathrm{avg}}/P)$ incremental updates for memoryless dynamics versus the renewal engine's dense synchronous stepping --- and validate both against exact references within their relevant regime (generalised non-Markovian Gillespie for SEIR in Section~\ref{sec:validation}; Doob--Gillespie for SIS and SIR in Section~\ref{sec:markov_validation} and Appendix~\ref{app:sis_sir}).

    \item \textbf{IO-aware kernel fusion for dense renewal dynamics.} Building on the IO-aware design principle of \citet{dao2022flashattention}, we consolidate CSR traversal, numerically stable $\mathrm{erfcx}$-based hazard evaluation, adaptive Bernoulli tau-leaping, and the next-step infectivity write-back into a single fused Triton kernel. All intermediates remain in streaming-multiprocessor registers; block-scalar sparsity skips the hazard for inert thread blocks without breaking CUDA Graph capture.

    \item \textbf{Degree-aware CSR dispatch.} To cope with highly heterogeneous degree distributions, we add a warp-per-node kernel and an edge-partitioned merge-based kernel \citep{merrill2016merge} and auto-select among the three from the graph's $D_{\max}/D_{\mathrm{avg}}$ ratio. The merge strategy delivers $4.5\times$ throughput on a $10^6$-node Barab\'{a}si--Albert graph while the auto-dispatch rule preserves peak throughput on regular graphs (Section~\ref{sec:csr_dispatch}, Appendix~\ref{app:degree_dispatch}).

    \item \textbf{Tau-leaping fidelity against an exact reference.} We quantify the structural-bias floor of synchronous Bernoulli updates on contact networks against an exact non-Markovian Gillespie reference: the per-run peak-$I$ error is $\sim$6\% and does not detectably decrease below that floor as $\varepsilon \to 0$ across two decades of tolerance on the validation benchmarks (Appendix~\ref{app:fidelity}).

    \item \textbf{Active-node compaction for temporal sparsity.} We exploit the inter-step temporal sparsity that synchronous tau-leaping typically overlooks by shrinking the fused-kernel grid from $\lceil N/B \rceil$ blocks to $\lceil |\mathbf{X} {\ne} R| / B \rceil$ via a Fixed-Grid Early-Exit pattern compatible with CUDA Graph capture. On saturating scale-free epidemics this delivers a $\mathbf{1.53\times}$ whole-run speedup on BA $m{=}4$ at $N{=}10^6$ with bit-identical compartment counts (Section~\ref{sec:compaction}).

    \item \textbf{Mixed-precision storage with an fp32 accumulator contract.} We downcast state (int8), age (fp16), infectivity (bf16), and weights (bf16) while holding the pressure accumulator and all kernel math in fp32. The storage-only compression delivers a $1.14\times$ throughput lift at $N = 10^6$ and extends the L2-reachable scale of the fused kernel by roughly $3\times$ (reaching $N = 10^8$ on a single A100), within a $\sim\!0.1\%$ deviation in final attack rate, well below the structural-bias floor (Section~\ref{sec:mixed_precision}).
\end{enumerate}

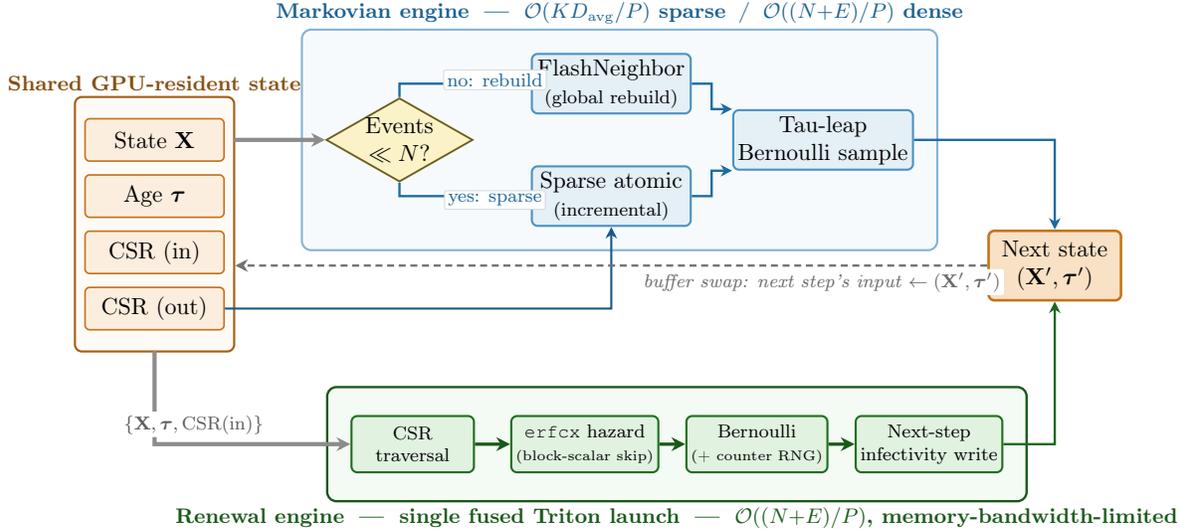
\begin{figure}[H]
\centering
\begin{tikzpicture}[
    state/.style={rectangle, draw=stateorange!85!black, line width=0.7pt,
                  rounded corners=2pt, fill=stateorange!14,
                  minimum width=2.1cm, minimum height=0.64cm,
                  align=center, font=\small, inner sep=2pt},
    mkernel/.style={rectangle, draw=markovblue!80, line width=0.9pt,
                    rounded corners=2.5pt, fill=markovblue!12,
                    minimum width=2.4cm, minimum height=0.9cm,
                    align=center, font=\small, inner sep=2pt},
    rkernel/.style={rectangle, draw=renewgreen!70!black, line width=0.9pt,
                    rounded corners=2.5pt, fill=renewgreen!14,
                    minimum width=1.85cm, minimum height=0.85cm,
                    align=center, font=\scriptsize, inner sep=2pt},
    decision/.style={diamond, draw=decisionyellow!45!black, line width=0.8pt,
                     aspect=1.75, inner sep=0pt, align=center,
                     font=\small, fill=decisionyellow!28},
    outstate/.style={rectangle, draw=stateorange!85!black, line width=1.1pt,
                     rounded corners=2.5pt, fill=stateorange!25,
                     minimum width=2cm, minimum height=1.05cm,
                     align=center, font=\small, inner sep=2pt},
    mflow/.style={-{Stealth[length=5pt,width=5pt]}, line width=1pt,
                  color=markovblue!90!black},
    rflow/.style={-{Stealth[length=5pt,width=5pt]}, line width=1pt,
                  color=renewgreen!65!black},
    pipe/.style={-{Stealth[length=6pt,width=6pt]}, line width=1.4pt,
                 color=renewgreen!55!black},
    bundle/.style={-{Stealth[length=6pt,width=6pt]}, line width=1.8pt,
                   color=black!45, line cap=round},
    feedback/.style={-{Stealth[length=5.5pt,width=5.5pt]}, line width=1pt,
                     color=black!55, densely dashed},
    elabel/.style={font=\footnotesize\bfseries, inner sep=2pt},
    branchlabel/.style={font=\scriptsize, inner sep=1.5pt,
                        fill=white, rounded corners=1pt,
                        draw=black!15, line width=0.3pt},
    buslabel/.style={font=\scriptsize\itshape, fill=white, inner sep=1.5pt,
                     rounded corners=1pt, text=black!65}
]

\node[state] (Xs)     at (0, 2.25) {State $\mathbf{X}$};
\node[state] (tau)    at (0, 1.40) {Age $\bm{\tau}$};
\node[state] (csrin)  at (0, 0.55) {CSR (in)};
\node[state] (csrout) at (0,-0.30) {CSR (out)};

\begin{scope}[on background layer]
    \node[fit=(Xs)(tau)(csrin)(csrout),
          draw=stateorange!70!black, line width=0.9pt,
          fill=stateorange!4, rounded corners=3pt,
          inner xsep=4pt, inner ysep=9pt,
          label={[elabel, text=stateorange!58!black]
                 above:Shared GPU-resident state}] (gpustate) {};
\end{scope}

\node[decision] (mode)   at (3.7,  2.25) {Events\\$\ll N$?};
\node[mkernel]  (flashN) at (6.9,  3.10) {FlashNeighbor\\\scriptsize{(global rebuild)}};
\node[mkernel]  (sparse) at (6.9,  1.40) {Sparse atomic\\\scriptsize{(incremental)}};
\node[mkernel]  (taulp)  at (10.1, 2.25) {Tau-leap\\Bernoulli sample};

\draw[bundle] (gpustate.east|-mode) -- (mode.west);

\draw[mflow] (mode.north) |- (flashN.west);
\node[branchlabel, text=markovblue!85!black]
      at ($(mode.north)!0.5!(flashN.west) + (0.45, 0.16)$) {no: rebuild};

\draw[mflow] (mode.south) |- (sparse.west);
\node[branchlabel, text=markovblue!85!black]
      at ($(mode.south)!0.5!(sparse.west) + (0.45,-0.16)$) {yes: sparse};

\draw[mflow] (flashN.east) -- ++(0.3,0) |- (taulp.north west);
\draw[mflow] (sparse.east) -- ++(0.3,0) |- (taulp.south west);

\begin{scope}[on background layer]
    \node[fit=(mode)(flashN)(sparse)(taulp),
          draw=markovblue!55, line width=0.9pt,
          fill=markovblue!4, rounded corners=4pt,
          inner xsep=10pt, inner ysep=10pt,
          label={[elabel, text=markovblue!80!black]
                 above:Markovian engine
                       \;\textemdash\; $\mathcal{O}(K D_{\mathrm{avg}}/P)$
                       sparse \,/\, $\mathcal{O}((N{+}E)/P)$ dense}]
          (markovbox) {};
\end{scope}

\node[rkernel] (r1) at (3.9,  -2.35) {CSR\\traversal};
\node[rkernel] (r2) at (6.5,  -2.35) {\texttt{erfcx} hazard\\\tiny(block-scalar skip)};
\node[rkernel] (r3) at (9.1,  -2.35) {Bernoulli\\\tiny(+ counter RNG)};
\node[rkernel] (r4) at (11.7, -2.35) {Next-step\\infectivity write};

\draw[pipe] (r1.east) -- (r2.west);
\draw[pipe] (r2.east) -- (r3.west);
\draw[pipe] (r3.east) -- (r4.west);

\draw[bundle]
     (gpustate.south) |- (r1.west) coordinate[pos=0.6] (bundleMid);
\node[buslabel] at ($(bundleMid)+(0,0.30)$)
     {$\{\mathbf{X},\bm{\tau},\mathrm{CSR}(\mathrm{in})\}$};

\draw[mflow, line width=0.9pt, color=markovblue!75!black]
     (csrout.east) -| (sparse.south);

\begin{scope}[on background layer]
    \node[fit=(r1)(r2)(r3)(r4),
          draw=renewgreen!55!black, line width=1pt,
          fill=renewgreen!5, rounded corners=4pt,
          inner xsep=10pt, inner ysep=12pt,
          label={[elabel, text=renewgreen!55!black]
                 below:Renewal engine \;\textemdash\;
                        single fused Triton launch \;\textemdash\;
                        $\mathcal{O}((N{+}E)/P)$, memory-bandwidth-limited}]
          (renewbox) {};
\end{scope}

\node[outstate] (out) at (13.6, 0.35)
      {Next state\\$(\mathbf{X}',\bm{\tau}')$};

\draw[mflow] (taulp.east) -| (out.north);
\draw[rflow] (r4.east) -| (out.south);

\coordinate (fbEnd) at (gpustate.east |- out.west);

\draw[feedback] (out.west) -- (fbEnd);

\node[font=\scriptsize\itshape, color=black!60, inner sep=1.5pt,
      fill=white, rounded corners=1pt]
     at ($(out.west)!0.22!(fbEnd)+(0,-0.25)$)
     {buffer swap: next step's input $\leftarrow (\mathbf{X}',\bm{\tau}')$};

\end{tikzpicture}
\caption{System overview. Every kernel reads from a shared GPU-resident state block (left). The Markovian engine (top) branches at runtime between a global rebuild (Control Mode, FlashNeighbor) and a sparse atomic update (Inertial Mode) depending on whether the number of per-step events is $\ll N$; both modes converge on a common tau-leap Bernoulli sampler. The Renewal engine (bottom) is a single fused Triton kernel that streams CSR traversal, $\mathrm{erfcx}$ hazard evaluation, Bernoulli sampling + RNG, and the next-step infectivity write-back in one launch (Section~\ref{sec:fused_kernel}). The dashed arc is the per-step buffer swap; next-step state becomes the next call's input.}
\label{fig:architecture}
\end{figure}

\section{Related work}\label{sec:related}

The Gillespie algorithm \citep{gillespie1977exact, gillespie1976general} provides exact stochastic simulation of continuous-time Markov processes. The Generalized Epidemic Modeling Framework (GEMF) of \citet{sahneh2013generalized} provides the mean-field formulation that underpins a family of compartmental simulators on multilayer networks; two high-performance simulator implementations of that family serve as our CPU baselines, c-GEMF \citep{sahneh2017gemfsim} (C, Next Reaction Method) and FastGEMF \citep{samaei2025fastgemf} (Python, vectorized updates). For non-Markovian systems, \citet{boguna2014simulating} introduced a generalized Gillespie algorithm that tracks the age $\tau_i$ of each process. The EoN package \citep{miller2020eon} provides efficient non-Markovian simulation without GPU acceleration. More recently, NEXT-Net \citep{cure2025fast} reformulated the next-reaction method so that one serial CPU simulation of an exact non-Markovian process scales to networks on the order of $10^6$ nodes, closing most of the scale gap between exact event-driven CPU simulation and approximate synchronous methods. The \flashspread{} contribution is therefore not a scale claim against all CPU methods: it is a scale claim under \emph{synchronous tau-leaping} (the regime that exposes dense $\bigO(N{+}E)$ parallelism) together with a first open-source, end-to-end GPU framework for renewal-process epidemics on networks with kernel-fused dense stepping, reaching $N = 10^8$ on a single A100. Recent work has characterized when non-Markovian and Markovian dynamics yield equivalent behavior \citep{feng2019equivalence, starnini2017equivalence}, finding that equivalence generally fails for transient dynamics relevant to epidemic forecasting.

Gillespie's tau-leaping \citep{gillespie2001approximate} accelerates stochastic simulation by approximating event counts as Poisson random variables. \citet{cao2006efficient} developed adaptive step selection, and \citet{chatterjee2005binomial} introduced binomial tau-leaping. Our \emph{Bernoulli tau-leaping} differs fundamentally: each node undergoes at most one transition per step, naturally modeled by Bernoulli rather than Poisson statistics. This is not an approximation against the binomial formulation: on the compartmental SEIR state space each node can change compartment at most once in any infinitesimal sub-interval, and the tau-leap is selected to bound the per-step transition probability by $\varepsilon \ll 1$; the binomial distribution therefore reduces to Bernoulli with higher-order correction $\bigO(\varepsilon^2)$ per node, which remains below the structural-bias floor of synchronous updates (Appendix~\ref{app:fidelity}).

GPU acceleration of epidemic simulation has primarily focused on ensemble parallelization \citep{perumalla2006discrete, zou2013epidemic} and agent-based models \citep{bisset2009parallel}. \citet{komarov2012accelerating} developed GPU-accelerated exact Gillespie for chemical networks, and \citet{nobile2014cutauleaping} implemented GPU tau-leaping for biochemical systems. To our knowledge, no prior open-source work has published a kernel-fused, end-to-end GPU pipeline for non-Markovian (renewal) epidemic simulation on networks; the challenge is fundamental, as the sparse update strategies that make Markovian GPU simulation efficient fail when all node hazards change continuously with time. \citet{dao2022flashattention} demonstrated that IO-aware kernel fusion (keeping intermediates in on-chip SRAM rather than writing to off-chip DRAM) can dramatically accelerate memory-bound workloads; we apply this principle to graph aggregation using a custom Triton kernel.\footnote{While frameworks such as Julia/CUDA.jl or \texttt{DifferentialEquations.jl} offer automatic kernel fusion via broadcast graphs, they cannot accommodate the data-dependent control flow of our block-scalar skip (Section~\ref{sec:fused_kernel}). Combined with the static operator requirements of CUDA Graph capture, achieving the optimised fusion path quantified in Figure~\ref{fig:memory_flow} required a custom Triton kernel.}

Mature GPU graph frameworks such as cuGraph \citep{nvidia2019cugraph} and Gunrock \citep{wang2016gunrock} do provide efficient CSR-structured SpMV and traversal primitives, and an obvious question is why not assemble the renewal step from them. We do not, for two reasons. First, neither framework exposes a single-kernel fusion point at which stochastic sampling (random draws per node) and age-dependent $\mathrm{erfcx}$-based hazard evaluation can co-reside with CSR traversal in streaming-multiprocessor registers; wrapping them would reintroduce the $\bigO(N)$ intermediate-tensor traffic that our fused kernel removes (Appendix~\ref{app:memory_flow}, Figure~\ref{fig:memory_flow}). Second, our degree-aware auto-dispatch (Section~\ref{sec:csr_dispatch}) is tied to the \emph{per-simulation-step} compute profile --- state, age, infectivity, and RNG all change every step --- rather than to a once-off graph query, which is the regime these libraries optimise for. We additionally rely on the CUDA Graphs API \citep{nvidia2024cudagraphs} to capture the $b$ repetitions of the fused step as a single replayable unit, eliminating per-step launch overhead (Section~\ref{sec:fused_kernel}).

Table~\ref{tab:software_comparison} compares \flashspread{} with existing packages. \flashspread{} demonstrates $N = 10^8$ on a single A100 (40~GB); at $N = 10^7$ the fp32 CSR graph structure ($\sim$640~MB for degree $d = 8$) already exceeds the GPU's 40~MB L2 cache, making every indirect neighbour lookup a full DRAM transaction, and the mixed-precision storage path of Section~\ref{sec:mixed_precision} (bf16 weights, int8 state, fp16 age, bf16 infectivity) compresses the per-step working set enough to reach $N = 10^8$ within the 40~GB HBM budget. Scaling beyond $N \sim 10^8$ exceeds single-GPU global memory and requires multi-GPU domain decomposition.

\begin{table}[H]
\centering
\caption{Comparison with existing network epidemic simulation software.}
\label{tab:software_comparison}
\begin{threeparttable}
\scriptsize
\setlength{\tabcolsep}{2pt}
\begin{tabular*}{\linewidth}{@{\extracolsep{\fill}}lccccccl@{}}
\toprule
\textbf{Package} & \textbf{GPU} & \textbf{Non-Markovian} & \textbf{Edge-w/} & \textbf{Kernel Fusion} & \textbf{Max scale} & \textbf{Language} & \textbf{Ref.} \\
 & \textbf{support} & \textbf{dynamics} & \textbf{ML} & \textbf{(IO-aware)} & \textbf{(nodes $N$)} & & \\
\midrule
GEMFsim     & \ding{55} & \ding{55} & \ding{51} / \ding{51} & \ding{55} & $10^5$ & C & \citet{sahneh2017gemfsim} \\
FastGEMF    & \ding{55} & \ding{55} & \ding{51} / \ding{51} & \ding{55} & $10^5$ & Python & \citet{samaei2025fastgemf} \\
EpiModel    & \ding{55} & \ding{55} & \ding{51} / \ding{55} & \ding{55} & $10^4$ & R & \citet{jenness2018epimodel} \\
NDlib       & \ding{55} & \ding{55} & \ding{55} / \ding{55} & \ding{55} & $10^5$ & Python & \citet{rossetti2018ndlib} \\
epydemic    & \ding{55} & \ding{55} & \ding{55} / \ding{55} & \ding{55} & $10^5$ & Python & \citet{dobson2017epydemic} \\
EoN         & \ding{55} & \ding{51} & \ding{51} / \ding{55} & \ding{55} & $10^5$ & Python & \citet{miller2020eon} \\
NEXT-Net    & \ding{55} & \ding{51} & \ding{51} / \ding{55} & \ding{55} & $10^6$ & C++/Python & \citet{cure2025fast} \\
\midrule
\textbf{\flashspread{}} & \ding{51} & \ding{51} & \ding{51} / \ding{51} & \ding{51} & $\mathbf{10^8}$ & Python & this work \\
\bottomrule
\end{tabular*}
\begin{tablenotes}
\footnotesize
\item \textbf{Non-Markovian dynamics}: supports age-dependent transition rates (renewal / log-normal / Weibull holding times), not only exponential (Markovian) waiting times. \textbf{Edge-weighted / multilayer}: first check indicates support for weighted edges; second check indicates support for multiple interaction layers. \textbf{Kernel Fusion}: consolidates CSR graph traversal, hazard evaluation, and stochastic sampling into a single GPU kernel launch with intermediates kept in on-chip registers (IO-aware design). \textbf{Max scale}: largest network size demonstrated in the corresponding reference; NEXT-Net \citep{cure2025fast} uses an exact next-reaction reformulation on CPU and reports runs on networks with $\sim$$10^6$ nodes, the current state of the art for exact event-driven CPU simulation. The \flashspread{} $\mathbf{10^8}$ entry is the approximate synchronous tau-leaping ceiling under mixed-precision storage on a single A100 (40~GB) and is therefore not directly comparable to NEXT-Net's exact-simulation claim.
\end{tablenotes}
\end{threeparttable}
\end{table}

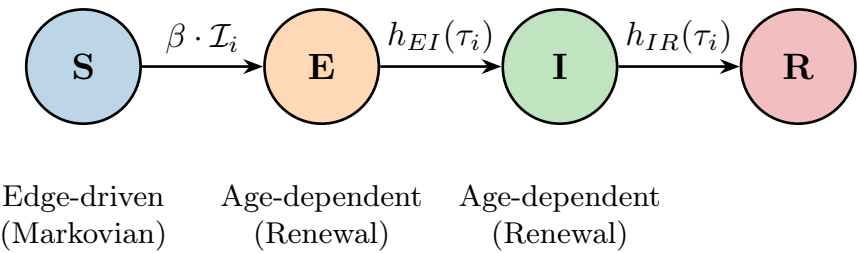
\begin{figure}[H]
\centering
\begin{tikzpicture}[
    state/.style={circle, draw, thick, minimum size=1.2cm, font=\bfseries},
    arrow/.style={-{Stealth[scale=0.9]}, thick},
    label/.style={font=\small, midway}
]

\node[state, fill=susceptible!30] (S) at (0,0) {S};
\node[state, fill=exposed!30] (E) at (2.5,0) {E};
\node[state, fill=infectious!30] (I) at (5,0) {I};
\node[state, fill=recovered!30] (R) at (7.5,0) {R};

\draw[arrow] (S) -- node[above, label] {$\beta \cdot \mathcal{I}_i$} (E);
\draw[arrow] (E) -- node[above, label] {$h_{EI}(\tau_i)$} (I);
\draw[arrow] (I) -- node[above, label] {$h_{IR}(\tau_i)$} (R);

\node[below=0.5cm of S, font=\footnotesize, align=center] {Edge-driven\\(Markovian)};
\node[below=0.5cm of E, font=\footnotesize, align=center] {Age-dependent\\(Renewal)};
\node[below=0.5cm of I, font=\footnotesize, align=center] {Age-dependent\\(Renewal)};

\end{tikzpicture}
\caption{SEIR model with mixed dynamics. The S$\to$E transition is edge-mediated and Markovian; the E$\to$I and I$\to$R transitions are nodal and non-Markovian, with age-dependent hazard functions $h(\tau_i)$.}
\label{fig:seir}
\end{figure}

\section{Problem formulation}\label{sec:formulation}

Consider a contact network $\mathcal{G} = (\mathcal{V}, \mathcal{E}, \mathbf{W})$ with $N = |\mathcal{V}|$ nodes and $L$ interaction layers. Each node $i$ occupies one of $M$ compartments. The system state is:
\begin{equation}
    (\mathbf{X}(t), \bm{\tau}(t)) \in \{1,\ldots,M\}^N \times \mathbb{R}_{\geq 0}^N
\end{equation}
where $\mathbf{X}(t)$ is the compartmental state vector and $\bm{\tau}(t)$ is the \emph{age tensor}, the time each node has spent in its current state. The age resets to zero upon any state transition (the \emph{renewal property}).

The instantaneous transition rate for node $i$ from compartment $m$ to $n$ decomposes into nodal and edge-mediated components:
\begin{equation}
    \lambda_{i}^{m \to n}(t) = \underbrace{\Psi_{\text{node}}^{m \to n}(\tau_i, \mathbf{u})}_{\text{spontaneous}} + \underbrace{\sum_{\ell=1}^{L} \Psi_{\text{edge}}^{m \to n, \ell}\left(\mathcal{I}_i^{(\ell)}, \mathbf{u}\right)}_{\text{contact-driven}}
    \label{eq:general_rate}
\end{equation}

The \emph{inducer influence} $\mathcal{I}_i^{(\ell)}$ aggregates contributions from neighbors weighted by an \emph{infectivity function} $\rho^{(\ell)}$:
\begin{equation}
    \mathcal{I}_i^{(\ell)}(\mathbf{X}, \bm{\tau}) = \sum_{j \in \mathcal{N}_{\text{in}}^{(\ell)}(i)} w_{ji}^{(\ell)} \cdot \rho^{(\ell)}(X_j, \tau_j)
    \label{eq:influence}
\end{equation}
For standard Markovian models, $\rho^{(\ell)}(X_j, \tau_j) = \mathbf{1}\{X_j = q^{(\ell)}\}$ is a binary indicator. For non-Markovian models with age-dependent transmission, $\rho^{(\ell)}$ can be a continuous function of the source node's infection age $\tau_j$ (Section~\ref{sec:source_node_approx}), enabling biologically realistic viral shedding curves without incurring $\bigO(E)$ per-edge memory costs.

The \emph{computational duality} central to \flashspread{} arises from how the nodal rate depends on age. In the Markovian case, the rate is age-independent, enabling $\bigO(K \cdot D_{\text{avg}})$ sparse updates. In the renewal case, the rate depends on holding time via the hazard function $h(\tau) = f(\tau)/S(\tau)$, requiring dense $\bigO(N)$ updates at every step.

\section{Computational engines}\label{sec:engines}

\flashspread{} provides two engines (Markovian, Renewal) and three CSR traversal strategies (thread, warp, merge) that are typically abstracted from the user: a compact decision tree, reproduced in Figure~\ref{fig:decision_tree}, drives the runtime dispatch. The remainder of this section outlines the engine-level mechanics referenced in Figure~\ref{fig:decision_tree}; Sections~\ref{sec:renewal_engine} and~\ref{sec:csr_dispatch} cover the renewal side.
\begin{figure}[H]
\centering
\begin{tikzpicture}[
    question/.style={rectangle, rounded corners=3pt,
                     draw=choice!45!black, line width=0.9pt,
                     fill=choice!20, align=center, font=\small,
                     minimum width=4.7cm, minimum height=0.95cm,
                     inner sep=3pt},
    outcomeA/.style={rectangle, rounded corners=3pt,
                     draw=markov!70!black, line width=1pt,
                     fill=markov!12, align=center, font=\small,
                     minimum width=3.8cm, minimum height=0.85cm,
                     inner sep=3pt},
    outcomeB/.style={rectangle, rounded corners=3pt,
                     draw=renew!60!black, line width=1pt,
                     fill=renew!12, align=center, font=\small,
                     minimum width=3.8cm, minimum height=0.85cm,
                     inner sep=3pt},
    thresh/.style={rectangle, fill=white, inner sep=2pt,
                   rounded corners=1pt, draw=black!18,
                   font=\scriptsize, line width=0.3pt},
    flow/.style={-{Stealth[length=5pt,width=5pt]}, line width=0.8pt,
                 color=black!65},
]

\node[question] (q1) at (0, 4.2)
    {Is any per-node transition rate\\ age-dependent
     (e.g.\ log-normal hazard)?};

\node[outcomeA] (markov) at (-4.5, 2.0)
    {\textbf{Markovian engine}\\
     \scriptsize(Section~\ref{sec:markov_engine}, Algorithm~\ref{alg:markov})};
\draw[flow] (q1.west) -| (markov.north);
\node[thresh] at ($(q1.west)!0.35!(markov.north)+(0,0.25)$)
    {\textbf{no} $\to$ all rates piecewise constant};

\node[question] (q2) at (4.5, 2.0)
    {How heavy-tailed is the degree\\ distribution?
     \\[1pt]\scriptsize($\rho := D_{\max}/D_{\mathrm{avg}}$, auto-measured once)};
\draw[flow] (q1.east) -| (q2.north);
\node[thresh] at ($(q1.east)!0.35!(q2.north)+(0,0.25)$)
    {\textbf{yes} $\to$ Renewal engine};

\node[outcomeB] (thread) at (-0.6, -0.45)
    {\textbf{thread} (1 per node)\\
     \scriptsize regular / ER / lattice};
\node[outcomeB] (warp)   at (4.5, -0.45)
    {\textbf{warp} (32 per node)\\
     \scriptsize moderately skewed};
\node[outcomeB] (merge)  at (9.6, -0.45)
    {\textbf{merge} (edge-partitioned)\\
     \scriptsize strongly scale-free};

\draw[flow] (q2.south) -| (thread.north);
\draw[flow] (q2.south) -- (warp.north);
\draw[flow] (q2.south) -| (merge.north);

\node[thresh] at ($(q2.south)!0.35!(thread.north)+(-0.2,0.25)$)
    {$\rho < 4$};
\node[thresh] at ($(q2.south)!0.50!(warp.north)+(0,0.50)$)
    {$4 \leq \rho < 50$};
\node[thresh] at ($(q2.south)!0.35!(merge.north)+(0.2,0.25)$)
    {$\rho \geq 50$};

\node[font=\scriptsize\itshape, color=black!55, align=center]
    at (2.5, -2.3)
    {\textbf{Auto-dispatch} (default): the engine reads
     \texttt{row\_ptr} once at construction,\\
     computes $\rho$, and selects the strategy above.
     The user can override via \texttt{csr\_strategy=...}};

\end{tikzpicture}
\caption{Practitioner decision tree. A single top-level question (``is the hazard age-dependent?'') selects the engine; a single scalar ($\rho = D_{\max}/D_{\mathrm{avg}}$, inspected once at construction) selects the CSR traversal strategy within the Renewal engine. Thresholds $(\rho_w, \rho_m) = (4, 50)$ are calibrated in Appendix~\ref{app:degree_dispatch}.}
\label{fig:decision_tree}
\end{figure}
Both engines share the FlashNeighbor kernel\label{sec:flash_neighbor} for computing inducer influence (Eq.~\ref{eq:influence}). A naive approach requires materializing dense indicator vectors, incurring $2NL$ memory transactions for discarded intermediates. We fuse operations into a single kernel that streams CSR edges, looks up neighbor states via indirect memory access, evaluates predicates in registers, and accumulates weighted sums, writing only the final influence tensor to global memory. The graph is stored in CSR format indexed by \emph{incoming} edges, enabling \emph{gather-based parallelism}: each GPU thread owns a unique target node and accumulates influence in private registers without atomic operations. The concrete data layout, per-thread register usage, and the corresponding traffic through HBM / L2 / SM registers are documented in Appendix~\ref{app:data_flow} (Figures~\ref{fig:csr_layout}--\ref{fig:memory_flow}); these details drive the $\sim$64$N \to \sim$20$N$ bytes/step fusion gain reported in Section~\ref{sec:fused_kernel}.

Listing~\ref{lst:python_api} shows the user-facing Python API, which is deliberately kept flat so that the Markovian / renewal engine choice, the CSR traversal strategy, and the tau-leaping tolerance are all explicit in the engine constructor.
\begin{figure}[H]
\begin{lstlisting}[language=Python, caption={Declarative simulation setup with \flashspread{}. The engine is constructed once with a graph, a compartment model, and a tolerance; subsequent \texttt{step()} calls advance one batched CUDA Graph replay. Both the CSR strategy (\texttt{"auto" / "thread" / "warp" / "merge"}) and the model's \texttt{transmission\_mode} (\texttt{"constant"} Markovian edges vs.\ \texttt{"age\_dependent"} source-node shedding) are runtime choices, not recompiles.}, label={lst:python_api}]
import torch
from flashspread import SEIRModel
from flashspread.core.network import FixedDegreeGraph
from flashspread.engines.renewal_fused import RenewalEngineFusedCUDAGraph

# 1. Graph and model are declarative:
graph = FixedDegreeGraph(num_nodes=1_000_000, degree=8, device="cuda")
model = SEIRModel(beta=0.25,
                  mean_ei=5.0, median_ei=4.0,
                  mean_ir=7.5, median_ir=5.0)
model.transmission_mode = "age_dependent"   # alternative: "constant"

# 2. Engine picks the best CSR strategy from D_max / D_avg:
engine = RenewalEngineFusedCUDAGraph(
    graph, model, device="cuda",
    epsilon=0.03, tau_max=0.1,      # tau-leaping knobs
    csr_strategy="auto",            # thread / warp / merge / auto
    steps_per_launch=50,            # CUDA Graph batch size b
    seed=12345,
)
engine.seed_infection(num_infected=100, state=model.exposed)

# 3. Run:
while engine.current_time < 50.0:
    _, state = engine.step()

print(engine.count_by_state())   # [S, E, I, R] on device
# Expected ensemble-mean populations at t=50 on this benchmark:
#   roughly [~0.04, ~0.02, ~0.05, ~0.89] * num_nodes
# (useful as a one-line post-install sanity check; variance across
#  seeds is a few percent, dominated by the synchronous-update floor
#  discussed in Appendix C.)
\end{lstlisting}
\end{figure}

For Markovian dynamics\label{sec:markov_engine}, transition rates are piecewise constant, enabling two operational modes: \emph{Control Mode} performs dense FlashNeighbor passes when control inputs change ($\bigO((N + E)/P)$), while \emph{Inertial Mode} exploits event-driven sparsity with $\bigO(|\mathcal{T}| \cdot D_{\text{avg}} / P)$ updates where $|\mathcal{T}| \ll N$. To bound accumulated rounding from the repeated sparse atomic updates, Inertial Mode periodically triggers a full Control-Mode recomputation (every 200 events in our implementation); this cadence is a tunable accuracy knob rather than a correctness requirement. Algorithm~\ref{alg:markov} summarizes the step logic. This sparse engine achieves $>2 \times 10^7$ events/sec at $N = 10^6$, but its optimization strategy (incremental updates between rare events) is inapplicable to renewal processes where all hazards change continuously.

\begin{algorithm}[H]
\caption{Markovian Engine: One Tau-Leaping Step}
\label{alg:markov}
\begin{algorithmic}[1]
\Require State $\mathbf{X}$, rates $\bm{\lambda}$, influence $\mathbf{I}$, parameters $\theta, p_{\max}, \tau_{\max}$
\Ensure Updated state $\mathbf{X}'$, elapsed time $\tau$
\State $\Lambda \leftarrow \sum_i \lambda_i$ \Comment{Total rate}
\State $\tau \leftarrow \min\!\left(\frac{\theta \cdot N}{\Lambda},\; \frac{p_{\max}}{\max_i \lambda_i},\; \tau_{\max}\right)$
\State $p_i \leftarrow 1 - e^{-\lambda_i \tau}$ for all $i$ \Comment{Transition probabilities}
\State $\mathbf{m} \leftarrow \text{Bernoulli}(\mathbf{p})$ \Comment{Sample event mask}
\State $\mathbf{X}' \leftarrow \textsc{ApplyTransitions}(\mathbf{X}, \mathbf{m})$
\If{few transitions (Inertial Mode)}
    \State Update $\mathbf{I}$ via sparse atomics on outgoing CSR
\Else
    \State $\mathbf{I} \leftarrow \textsc{FlashNeighbor}(\mathbf{X}')$ \Comment{Control Mode}
\EndIf
\State $\bm{\lambda} \leftarrow \textsc{ComputeRates}(\mathbf{X}', \mathbf{I})$
\end{algorithmic}
\end{algorithm}

\section{The Renewal engine}\label{sec:renewal_engine}

For renewal processes, all $N$ hazards change at every time step, precluding the sparse updates that make Markovian simulation efficient. This section presents our core contribution: a fused GPU kernel that achieves memory-bandwidth-limited execution for dense renewal dynamics.

\subsection{Numerically stable log-normal hazards}

Direct computation of $h(\tau) = f(\tau)/S(\tau)$ suffers from catastrophic cancellation for large $\tau$. We derive a stable formulation using the scaled complementary error function:

\begin{proposition}[Stable Log-Normal Hazard]
For a log-normal distribution with parameters $(\mu, \sigma)$:
\begin{equation}\label{eq:stable_hazard}
h_{\mathrm{LN}}(\tau; \mu, \sigma) = \sqrt{\frac{2}{\pi}} \,\frac{1}{\tau \sigma \,\mathrm{erfcx}(z)}, \quad z = \frac{\ln \tau - \mu}{\sigma\sqrt{2}}
\end{equation}
\end{proposition}

\begin{proof}
The log-normal density and survival function are $f(\tau) = \frac{1}{\tau \sigma \sqrt{2\pi}} e^{-z^2}$ and $S(\tau) = \frac{1}{2} \mathrm{erfc}(z)$. The hazard is:
\begin{align}
    h(\tau) = \frac{f(\tau)}{S(\tau)} = \sqrt{\frac{2}{\pi}} \cdot \frac{1}{\tau \sigma} \cdot \frac{1}{e^{z^2} \mathrm{erfc}(z)} = \sqrt{\frac{2}{\pi}} \cdot \frac{1}{\tau \sigma \cdot \mathrm{erfcx}(z)}
\end{align}
where $\mathrm{erfcx}(z) = e^{z^2} \mathrm{erfc}(z)$ remains well-conditioned for all $\tau > 0$ via \texttt{torch.special.erfcx}.
\end{proof}

\subsection{Bernoulli tau-leaping}

In compartmental epidemic models, each node can undergo \emph{at most one} transition per step. The appropriate model is Bernoulli:
\begin{equation}
    \text{transition}_i \sim \text{Bernoulli}(p_i), \quad p_i = 1 - e^{-\lambda_i \Delta t}
\end{equation}
with adaptive step selection:
\begin{equation}
    \Delta t = \min\left(\Delta t_{\max},\; \frac{\varepsilon}{\max_i \lambda_i + \delta}\right)
\end{equation}
where $\varepsilon \in [0.01, 0.05]$ bounds the maximum transition probability per step. As $\varepsilon \to 0$, the approximation converges to the exact continuous-time dynamics; our validation (Section~\ref{sec:validation}) demonstrates that $\varepsilon = 0.03$ produces trajectories in close agreement with exact Gillespie methods.

Achieving GPU acceleration necessitates a fundamental algorithmic paradigm shift. In existing sequential non-Markovian simulators, the continuous re-evaluation of hazard rates relies on a vast system size ($N \gg 1$) to naturally produce infinitesimal time steps \citep{boguna2014simulating}. By employing a first-order expansion, this approach inherently imposes an approximation that mathematically freezes the hazard rates during each interval. While highly accurate in the thermodynamic limit, this finite-$N$ approximation degrades if intrinsic system rates drop and the natural time step becomes too large. To guarantee tight bounds and prevent this degradation, sequential models often inject an independent auxiliary Markovian process, effectively acting as a phantom process \citep{vajdi2020stochastic}. This approach floods the simulation with $\mathcal{O}(\lambda_0)$ dummy events per unit time. Because the auxiliary rate $\lambda_0$ must be kept aggressively large to restrict integration error, the strict sequential event queue becomes an insurmountable bottleneck, choked by zero-state-change events. Consequently, tau-leaping is not merely a performance heuristic; it is the required algorithmic paradigm to break this sequential dependency and expose the dense $\mathcal{O}(N)$ data parallelism that the GPU subsequently exploits.

\subsection{Source-node approximation for age-dependent edges}
\label{sec:source_node_approx}

Exact non-Markovian simulation typically tracks per-edge ages at $\mathcal{O}(E \times M)$ memory cost to enable arbitrary edge dynamics \citep{boguna2014simulating}. However, we can bypass this massive memory footprint by adopting a host-centric approach—theoretically equivalent to ``Rule 2'' in \citet{boguna2014simulating}. By observing that viral shedding is fundamentally a property of the \emph{host} rather than the \emph{contact}, we define the instantaneous infectivity as:
\begin{equation}
    \rho^{(\ell)}(X_j, \tau_j) = \beta \cdot s(\tau_j) \cdot \mathbf{1}\{X_j = q^{(\ell)}\}
    \label{eq:infectivity}
\end{equation}
where $s(\tau_j)$ is a user-defined viral shedding profile reflecting the time course of infectiousness (e.g., a log-normal density calibrated to empirical viral load data \citep{he2020temporal}). This formulation decouples a node's external infectiousness, governed by $s(\tau)$, from its internal recovery dynamics governed by the hazard $h(\tau)$. Setting $s(\tau_j) = 1$ recovers standard Markovian edge semantics.\footnote{The approximation is exact when the per-contact transmission rate depends only on the source node's infection age, which is the standard epidemiological assumption for viral shedding profiles. It breaks down in three families of scenarios: (i) \emph{dose-dependent transmission}, where the per-edge contribution depends on the cumulative contact history between a specific (source, target) pair rather than only on the source's shedding trajectory; (ii) \emph{temporal networks with independently forming edges}, where each edge itself has a lifetime that must be tracked; and (iii) \emph{heterogeneous per-edge susceptibility} drawn from a non-trivial distribution, when the goal is to replicate that exact distributional shape rather than an ensemble-mean effect. In all three cases true per-edge age tracking at $\bigO(E)$ memory is required; we view this as an extension of the renewal engine rather than a correction to the present paper.}

\subsection{Fused kernel with block-level sparsity}
\label{sec:fused_kernel}

The unfused renewal pipeline executes $\sim$12 separate kernel launches per step, materializing $\sim$7 intermediate $\bigO(N)$ buffers to global memory (pressure, rates, event probabilities, random samples, event masks, next state, age updates), totaling $\sim$64$N$ bytes of traffic. We consolidate everything into a single fused kernel, reducing traffic to $\sim$20$N$ bytes per step.
A fundamental tension arises between compute sparsity and CUDA Graph compatibility. During a typical SEIR epidemic, $\sim$70\% of nodes occupy states S or R where the expensive $\mathrm{erfcx}$-based hazard ($\sim$55 FLOPs) need not be evaluated. In PyTorch, exploiting this requires data-dependent guards that force CPU--GPU synchronization and break CUDA Graph capture. Our fused kernel resolves this by performing a block-scalar reduction over a per-block E/I-existence flag that produces a true hardware branch rather than per-lane predication: thread blocks containing no active nodes skip the entire hazard evaluation, while blocks with at least one E or I node evaluate the hazard and mask-select the result across lanes. This preserves CUDA Graph capture while still recovering most of the compute sparsity. 

The GPU kernel framework lacks a native $\mathrm{erfcx}$ instruction, necessitating a piecewise identity-based (for $|z| \leq 3.5$) and asymptotic (for $|z| > 3.5$) approximation inside the kernel. The measured maximum relative error of this approach is $\sim$$4 \times 10^{-2}$ at the branch-switch $z \approx 3.5$ and $\sim$$6 \times 10^{-3}$ elsewhere on an fp32-safe grid (Appendix~\ref{app:erfcx}). Both errors are well within the tolerance of stochastic tau-leaping, as the induced bias in the Bernoulli probability scales as $\lambda_i \Delta t \lesssim \varepsilon$ and stays far below the $\sim$6--7\% structural-bias floor of synchronous updates (Section~\ref{sec:validation}; $\sim$6\% on peak-$I$, $\sim$7\% on final-$R$).

\begin{algorithm}[H]
\caption{Fused Renewal Engine with CUDA Graph: Capture}
\label{alg:fused_capture}
\begin{algorithmic}[1]
\Require CSR graph $(R, C, W)$, SEIR model parameters $(\beta, \mu_{EI}, \sigma_{EI}, \mu_{IR}, \sigma_{IR})$, batch size $b$, RNG seed $s$, tolerance $\varepsilon$, cap $\tau_{\max}$
\State \textbf{Persistent state}: tensors $\mathbf{X}, \bm{\tau}$, scalars $\tau_{\mathrm{prev}} \in \mathbb{R}_{\geq 0}$, \texttt{step\_id} $\in \mathbb{N}$
\State Initialise $\tau_{\mathrm{prev}} \gets \tau_{\max}$ and \texttt{step\_id} $\gets 0$
\State Snapshot all tensor state
\State Warmup: execute 3 steps for JIT compilation
\State Capture $b$ repetitions of \textsc{FusedStep} (Alg.~\ref{alg:fused}) as CUDA Graph $\mathcal{G}$
\State Restore tensor state from snapshot
\State \Return $\mathcal{G}$, initialised persistent state
\end{algorithmic}
\end{algorithm}

\begin{algorithm}[H]
\caption{Fused Renewal Engine with CUDA Graph: per-call replay of $b$ fused steps (\textsc{FusedStep}). Note: $\tau_{\text{prev}}$ is initialised to $\tau_{\max}$ in Algorithm~\ref{alg:fused_capture}, so the \emph{first} step of each replay uses $\tau_{\max}$ as a conservative upper bound on the step size; subsequent steps use the adaptive $\tau$ written by the previous iteration's tau-update line. This introduces at most one over-conservative step per $b{=}50$-step replay window and is the direct cost of keeping the adaptive step inside a capturable graph --- the alternative, recomputing $\tau$ outside the graph every step, would eliminate the gain of CUDA Graph batching.}
\label{alg:fused}
\begin{algorithmic}[1]
\Require State $\mathbf{X}$, age $\bm{\tau}$, CSR graph $(R, C, W)$, model parameters $(\beta, \mu, \sigma)$, batch size $b$, seed $s$; captured graph $\mathcal{G}$
\Ensure Updated $(\mathbf{X}', \bm{\tau}')$, elapsed time $T$
\State \textbf{Persistent state (carried across calls)}: $\tau_{\mathrm{prev}}$, \texttt{step\_id}; initialised by Alg.~\ref{alg:fused_capture}
\State $T \leftarrow 0$
\State $\mathcal{G}.\text{replay}()$ \Comment{Executes $b$ fused steps:}
\For{$k = 1$ \textbf{to} $b$} \Comment{(inside captured graph)}
    \State $T \leftarrow T + \tau_{\text{prev}}$
    \State \Comment{Step 1: Infectivity pre-pass ($\bigO(N)$, elementwise)}
    \State $\text{infectivity}[j] \leftarrow \beta \cdot s(\tau_j) \cdot \mathbf{1}\{X_j = I\}$ for all $j$
    \State step\_id $\leftarrow$ step\_id $+ 1$
    \State
    \State \Comment{Step 2: Fused GPU kernel (single launch, per thread $i$):}
    \State $p_i \leftarrow 0$ \Comment{CSR traversal $\to$ pressure in registers}
    \For{$j \in \mathcal{N}_{\text{in}}(i)$ via CSR}
        \State $p_i \leftarrow p_i + \text{infectivity}[j] \cdot W[j]$
    \EndFor
    \State Load $X_i, \tau_i$ from global memory
    \State $\lambda_i \leftarrow p_i$ if $X_i = S$; $\lambda_i \leftarrow 0$ if $X_i = R$
    \If{any $E$-node in block}
        \State $\lambda_i \leftarrow h_{\text{LN}}(\tau_i; \mu_{EI}, \sigma_{EI})$ for $E$-nodes \Comment{block-level skip}
    \EndIf
    \If{any $I$-node in block}
        \State $\lambda_i \leftarrow h_{\text{LN}}(\tau_i; \mu_{IR}, \sigma_{IR})$ for $I$-nodes
    \EndIf
    \State $q_i \leftarrow 1 - e^{-\lambda_i \cdot \tau_{\text{prev}}}$; \; $u \leftarrow \textsc{Rand}(s + \text{step\_id},\; i)$
    \State $X_i' \leftarrow \textsc{Transition}(X_i,\; u < q_i)$
    \State $\tau_i' \leftarrow 0$ if $X_i' \neq X_i$, else $\tau_i + \tau_{\text{prev}}$ \Comment{renewal reset}
    \State Store $X_i', \tau_i', \lambda_i$ \Comment{single global write}
    \State
    \State \Comment{Step 3: Tau update ($\bigO(N)$ reduction)}
    \State $\tau_{\text{prev}} \leftarrow \min\!\left(\tau_{\max},\; \varepsilon / (\max_i \lambda_i + \delta)\right)$
\EndFor
\end{algorithmic}
\end{algorithm}

Figure~\ref{fig:fused_step} diagrams the per-thread control flow inside Algorithm~\ref{alg:fused}. Each thread opens with five coalesced HBM reads (\texttt{state}, \texttt{age}, \texttt{pressure}, \texttt{infectivity}, RNG state), then participates in two block-scalar reductions (\texttt{any\_E}, \texttt{any\_I}) that produce a \emph{block-level} branch into either a skip path (no $\mathrm{erfcx}$ FLOPs) or an evaluate path (all lanes compute the log-normal hazard and mask-select the result). The Bernoulli sample, state transition, renewal age reset, and next-step infectivity write all happen in registers, and the thread closes with a single coalesced HBM write that commits the updated \texttt{state}, \texttt{age}, and next-step \texttt{infectivity}. No thread writes intermediate tensors to HBM --- the essential property that makes the fused design compatible with CUDA Graph capture of many repetitions and that keeps the arithmetic intensity on the memory-bound side of the roofline (Table~\ref{tab:roofline}).

\begin{figure}[H]
\centering
\begin{tikzpicture}[
    node distance=7mm and 5mm,
    box/.style={rectangle, rounded corners=3pt, draw, align=center,
                font=\footnotesize, inner sep=5pt, line width=0.8pt},
    hbmbox/.style={box, draw=hbmorange!80!black, fill=hbmorange!12,
                   minimum width=7.2cm},
    scalarbox/.style={box, draw=smgreen!60!black, fill=smgreen!10,
                      minimum width=7.2cm},
    regbox/.style={box, draw=regblue!65!black, fill=regblue!10,
                   minimum width=9.0cm},
    decision/.style={diamond, draw=choice!55!black, fill=choice!25,
                     aspect=2.4, inner sep=2pt, font=\footnotesize,
                     align=center, line width=0.8pt},
    skipbox/.style={box, draw=accentred!55!black, fill=accentred!8,
                    minimum width=3.3cm, font=\scriptsize},
    evalbox/.style={box, draw=smgreen!55!black, fill=smgreen!12,
                    minimum width=3.3cm, font=\scriptsize},
    flow/.style={-{Stealth[length=5pt,width=5pt]}, line width=0.8pt,
                 color=black!65},
    lbl/.style={font=\scriptsize\itshape, text=black!60},
]

\node[hbmbox] (reads)
     {\textbf{5 HBM reads per thread}\\
      \texttt{\scriptsize state, age, pressure, infectivity, rng}};

\node[scalarbox, below=of reads] (scalar)
     {\textbf{Block-scalar reductions}\\
      \texttt{\scriptsize any\_E = OR(state=E);\ any\_I = OR(state=I)}};

\node[decision, below=of scalar] (dec) {any\_E or any\_I?};

\node[skipbox, below left=13mm and 6mm of dec] (skip)
     {\textbf{Skip}\\set $\lambda\!=\!0$;\\
      skip erfcx};

\node[evalbox, below right=13mm and 6mm of dec] (eval)
     {\textbf{Evaluate}\\$\lambda = p \cdot \mathrm{erfcx}(z)$};

\node[regbox, below=24mm of dec] (bern)
     {\textbf{Register-only: Bernoulli sample + transition + next-step infectivity}\\
      \scriptsize sample $r$;\ fire if $r < 1 - \exp(-\lambda \tau)$;\ update state; reset age on event};

\node[hbmbox, below=of bern] (write)
     {\textbf{1 HBM write per thread}\\
      \texttt{\scriptsize state, age, infectivity (next-step)}};

\draw[flow] (reads) -- (scalar);
\draw[flow] (scalar) -- (dec);
\draw[flow] (dec.west) -| (skip.north);
\draw[flow] (dec.east) -| (eval.north);
\draw[flow] (skip.west) to[out=180, in=180, looseness=1.2] (bern.west);
\draw[flow] (eval.east) to[out=0, in=0, looseness=1.2] (bern.east);
\draw[flow] (bern) -- (write);

\node[lbl] at ([yshift=2mm, xshift=-5mm] dec.west) {no E / no I};
\node[lbl] at ([yshift=2mm, xshift=5mm] dec.east) {at least one};

\node[lbl, right=2mm of reads] {HBM};
\node[lbl, right=2mm of bern] {registers};
\node[lbl, right=2mm of write] {HBM};

\end{tikzpicture}
\caption{Per-thread control flow inside Algorithm~\ref{alg:fused}: five coalesced HBM reads on entry, a two-way branch selected by block-scalar reductions over the current compartments in the thread block, register-only work (Bernoulli sample, transition, age reset, next-step infectivity), and one coalesced HBM write on exit. The branch is \emph{block-level} (a scalar decision for the whole warp/block), not per-lane: blocks with no E or I node skip the $\sim$55-FLOP $\mathrm{erfcx}$ call entirely without breaking CUDA Graph capture, while blocks with at least one E or I node evaluate the hazard for all lanes and mask-select the result. Companion to the data-layout diagram of Figure~\ref{fig:csr_layout} and the memory-hierarchy diagram of Figure~\ref{fig:memory_flow}.}
\label{fig:fused_step}
\end{figure}
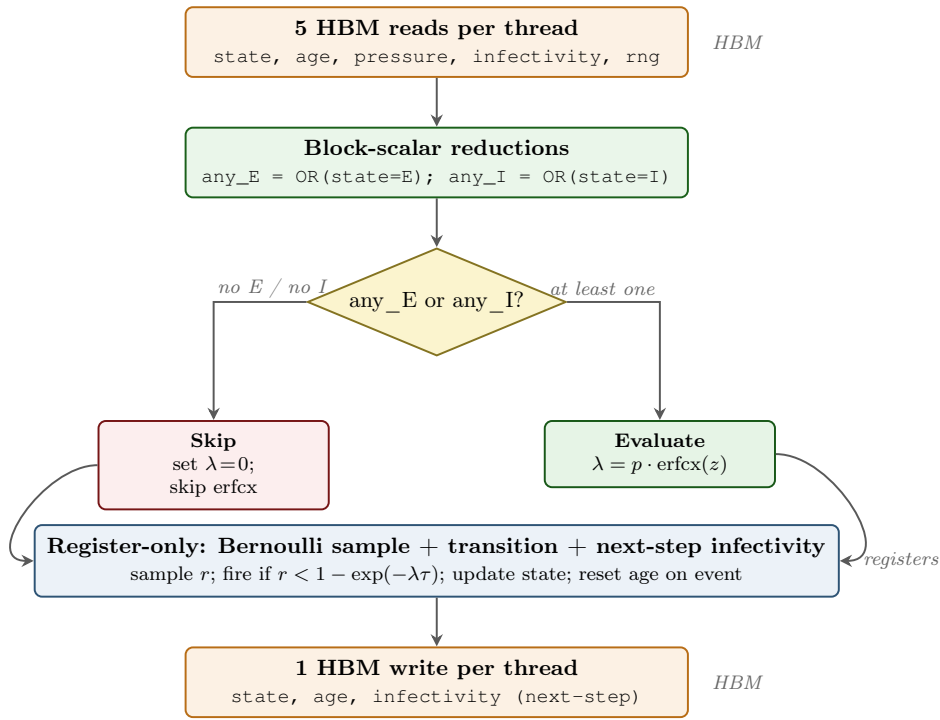

\subsection{CSR traversal strategies and auto-dispatch}
\label{sec:csr_dispatch}

The fused kernel described so far assigns one thread per node and iterates that node's incoming neighbors in a scalar loop. This is optimal when the degree distribution is narrow (regular, Erd\H{o}s--R\'{e}nyi, or lightly perturbed lattices), because the loop count across a 128-thread block is uniform and memory accesses coalesce. On heavy-tailed topologies (scale-free networks, social graphs), a single hub can dominate the block's runtime while the other 127 threads wait: measured throughput on a Barab\'{a}si--Albert graph drops by more than an order of magnitude despite equivalent $N$ and average degree (Section~\ref{sec:experiments}).
We therefore provide two additional traversal kernels and let the engine select among them at construction time:

\begin{itemize}
\item \emph{Warp-per-node}. The block is laid out as $[\texttt{NODES\_PER\_BLOCK}, \texttt{LANES\_PER\_NODE}]$ with $\texttt{LANES\_PER\_NODE}=32$. Each warp cooperates on one node's neighbor list, iterating $\lceil D_i / 32\rceil$ chunks in parallel, then a warp-local reduction yields the node's pressure. This reduces hub traversal from $\bigO(D_{\max})$ to $\bigO(D_{\max}/32)$ but still leaves inter-warp idle time inside a block that mixes one hub with several low-degree nodes.
\item \emph{Edge-partitioned merge-based}. Following the classic merge-path load-balancing idea of \citet{merrill2016merge}, we launch one program per fixed chunk of $\texttt{EDGES\_PER\_BLOCK}$ contiguous edges. Each lane loads one edge's $(\text{col\_ind}, \text{weight}, \text{infectivity}[\text{neighbor}])$, recovers the source node id via a fully unrolled binary search of depth $\lceil \log_2(N+2)\rceil$ on \texttt{row\_ptr}, and atomic-adds the weighted contribution to a shared \texttt{pressure} buffer. A cheap second kernel then runs the per-node tail (hazard, Bernoulli, transition, next-infectivity write) exactly as in Algorithm~\ref{alg:fused}. Load balance across blocks is perfect regardless of degree skew, but the scheme costs one atomic per edge, one scratch buffer read/write, and one extra kernel launch.
\end{itemize}

The three strategies are bit-exact equivalent (to within floating-point reduction order) because they all use the same counter-based RNG seeded by global node id and step counter, and the same per-node tail logic. We verified this on $N=10^4$ regular and BA graphs: thread and warp are bit-identical every step; merge uses non-deterministic atomic-add ordering and therefore matches only at population-count granularity (zero-delta in our tests).

\paragraph{Auto-dispatch.} At engine construction the framework inspects the graph's \texttt{row\_ptr} once, computes $\rho = D_{\max}/D_{\mathrm{avg}}$, and picks:
\begin{equation*}
\text{strategy} = \begin{cases}
\texttt{thread} & \rho < \rho_w,\\
\texttt{warp}   & \rho_w \leq \rho < \rho_m,\\
\texttt{merge}  & \rho \geq \rho_m,
\end{cases}
\end{equation*}
with $(\rho_w, \rho_m) = (4, 50)$ calibrated against the benchmarks in Table~\ref{tab:csr_dispatch}. The user can override the choice via an explicit strategy flag at engine construction. The dispatch cost (one pass over \texttt{row\_ptr}) is amortized over every subsequent simulation step.

\begin{table}[H]
\caption{CSR traversal strategy throughput (Fused CG, $b = 50$, $N = 10^6$, mean degree 8). The regular graph has $\rho = D_{\max}/D_{\mathrm{avg}} = 2$ and auto-dispatch picks \texttt{thread}; the BA graph has $\rho = 484$ and auto-dispatch picks \texttt{merge}. Each strategy produces statistically equivalent trajectories; only wall-clock differs.}
\label{tab:csr_dispatch}
\centering
\begin{tabular*}{\linewidth}{@{}lccc@{}}
\toprule
\textbf{Strategy} & \textbf{Regular ($d=8$)} & \textbf{BA ($m=4$)} & \textbf{BA / Regular} \\
 & G-NUPS & G-NUPS & \% \\
\midrule
\texttt{thread} (1 thread / node)        & 7.88 & 0.45 & 5.7\%  \\
\texttt{warp} (32 threads / node)        & 1.70 & 1.30 & 76.5\% \\
\texttt{merge} (edge-partitioned)        & 3.92 & 2.00 & 51.0\% \\
\bottomrule
\end{tabular*}
\end{table}

Relative to the default 1-thread-per-node kernel, the merge-based kernel is $4.5\times$ faster on BA but $2\times$ slower on a regular graph, which is exactly why the auto-dispatch heuristic is needed: the same kernel cannot be optimal for both regimes, and the user should not have to know which one their graph falls into. For each strategy, \texttt{NODES\_PER\_BLOCK} (warp) and \texttt{EDGES\_PER\_BLOCK} (merge) are small tunables with modest impact (Section~\ref{sec:experiments}).

\subsection{Epidemic-aware compute sparsity: active-node compaction}
\label{sec:compaction}

The fused kernel of Section~\ref{sec:fused_kernel} already skips the $\mathrm{erfcx}$ hazard for thread blocks that contain no $E$ or $I$ node; the block-scalar reduction harvests \emph{intra-step} compute sparsity along the spatial dimension (which blocks are inert this step). A natural complementary optimisation harvests \emph{inter-step temporal} sparsity: once a node reaches the absorbing state $R$, it takes no further action, and every replay step it still participates in (a block it belongs to, an $\bigO(N)$ copy-back of \texttt{state}/\texttt{age}/\texttt{infectivity}) is pure overhead. Dropping recovered nodes from the per-step workload shrinks the kernel's block count from $\lceil N/B \rceil$ toward $\lceil |\mathbf{X}{\ne}R| / B \rceil$, which approaches $\bigO(1)$ as the tail of a saturating epidemic completes.

We implement this as a CUDA-Graph-compatible \emph{active-node compaction} path. The predicate is $\mathbf{X} \ne R$ --- \emph{not} $\mathbf{X} \in \{E, I\}$, because the pull-based CSR gather of Section~\ref{sec:fused_kernel} also requires $S$ nodes to evaluate their incoming pressure in order to transition. Dropping $S$ would freeze the epidemic. Because $R$ is absorbing, the set of non-$R$ nodes shrinks monotonically, so refreshing the active list at the CUDA Graph replay boundary (once per $b=50$-step window) stays correct even when nodes transition to $R$ mid-replay: they remain on the list doing harmless $\mathrm{rate}=0$ evaluations until the next refresh.

\paragraph{Fixed-grid, early-exit pattern.}
To keep the refresh outside the captured CUDA Graph while the kernel launch inside stays immutable, we pre-allocate two static buffers: an int32 \texttt{active\_nodes[$N{+}B$]} holding the sorted non-$R$ node ids (zero-padded past $N$ to make any tail-block masked load well-defined) and a scalar int32 \texttt{num\_active}. The captured kernel grid is fixed at $\lceil N/B \rceil$. On entry, the kernel loads \texttt{num\_active}, sets \texttt{mask = offsets < num\_active}, and remaps \texttt{idx = active\_nodes[offsets]} through a masked gather; inactive tail lanes produce \texttt{idx = 0} but contribute nothing because every subsequent \texttt{tl.load}/\texttt{tl.store} is mask-gated. The refresh itself is a single \texttt{torch.nonzero} + in-place \texttt{copy\_} between replays; the underlying pointers never move, so no graph recapture is triggered. The only correctness subtlety is that inactive positions of the \texttt{rates} buffer retain stale, historically-nonzero values from when the corresponding nodes were last E or I; we zero \texttt{rates} once before each replay, matching the baseline invariant that $\mathrm{rate}[R] = 0$ every step. With that zeroing, the compaction path produces bit-identical compartment counts against the baseline at every step checkpoint on an $N=10^4$ sanity sweep (five seeded checkpoints, three independent seeds).

\paragraph{Topology-dependent gain.}
Figure~\ref{fig:compaction_nups_t} overlays NUPS$(t)$ for both topologies on the renewal benchmark, bucketed into ten temporal windows of the $\mathrm{TF}=50$ run. Table~\ref{tab:compaction_summary} gives the whole-run means. On ER $d{=}8$ the epidemic saturates slowly under the benchmark parameters (final attack rate $\approx 15\%$), so most of the run has $\mathrm{num\_active} \approx N$ and the refresh overhead exceeds the marginal block-skip gain: compaction is $\sim$1\% slower on the whole-run mean, with the per-bucket ratio crossing back above 1 only in the last temporal window. On BA $m{=}4$ at the same $N$ and $\beta$, explosive hub-driven takeoff reaches $\approx 97\%$ attack rate by $\mathrm{TF}=50$, so the shrinking active set dominates almost the entire run and the per-bucket speedup climbs from $1.0\times$ (takeoff) to $4.8\times$ (last window, num\_active $\approx 3\% N$) for a whole-run mean of $\mathbf{1.53\times}$ (Table~\ref{tab:compaction_summary}). This makes compaction a high-saturation / long-horizon optimisation: it is most valuable on topologies and parameter regimes that drive the epidemic through a long tail phase (scale-free contact networks, policy evaluations that simulate past the peak, RL rollouts that deliberately run to extinction).

\begin{table}[H]
\caption{Active-node compaction summary at $N{=}10^6$, $\mathrm{TF}{=}50$, 5 replicate trials. ``Final $R/N$'' is the mean attack rate at $t{=}50$. NUPS is the whole-run mean. Std across replicates is below $1\%$ in every cell. Bit-identity of compaction and baseline compartment counts was verified at five checkpoints ($t \in \{1, 5, 15, 30, 50\}$) across three seeds on a paired $N{=}10^4$ run.}
\label{tab:compaction_summary}
\centering
\begin{tabular*}{\linewidth}{@{\extracolsep{\fill}}lcccc@{}}
\toprule
\textbf{Graph} & \textbf{Final $R/N$} & \textbf{Baseline (G-NUPS)} & \textbf{Compaction (G-NUPS)} & \textbf{Speedup} \\
\midrule
ER $d{=}8$  & $0.156$ & $7.60$ & $7.52$ & $0.99\times$ \\
BA $m{=}4$  & $0.978$ & $0.44$ & $0.67$ & $\mathbf{1.53\times}$ \\
\bottomrule
\end{tabular*}
\end{table}

\begin{figure}[H]
\centering
\includegraphics[width=0.98\columnwidth]{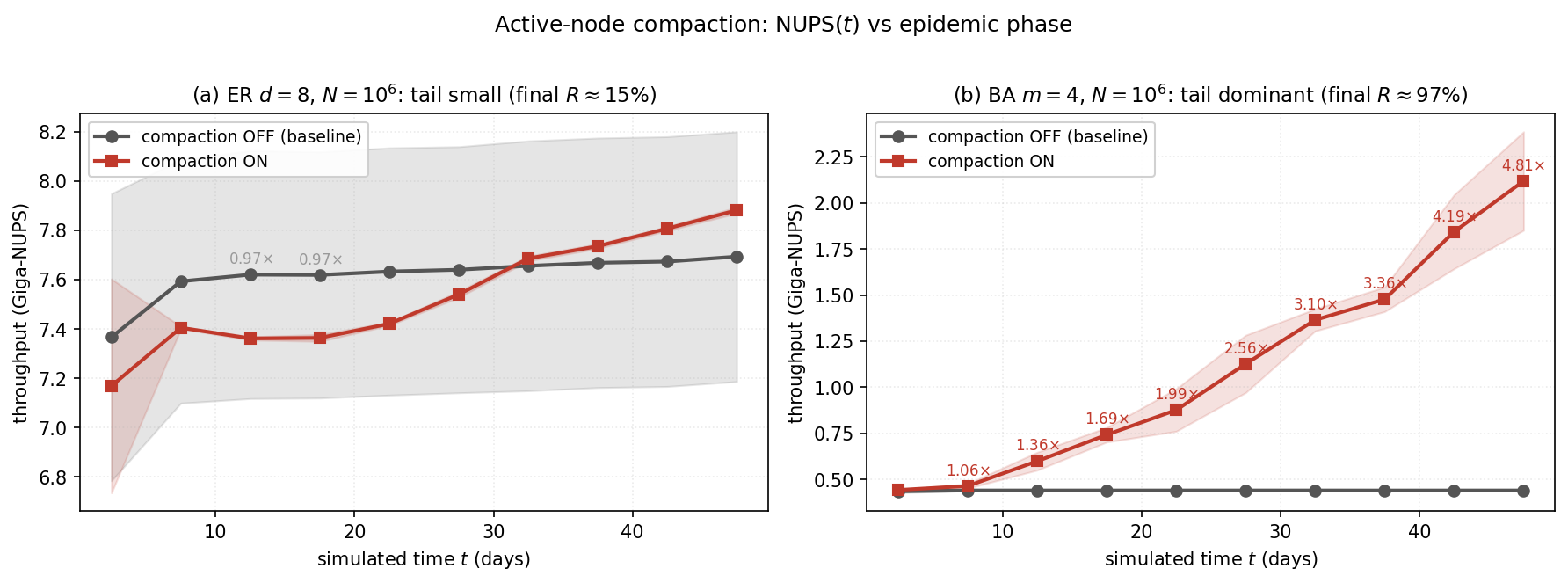}
\caption{Active-node compaction: NUPS$(t)$ stratified into ten temporal windows of the $\mathrm{TF}=50$ renewal benchmark, 5 replicate trials (mean $\pm$ 1 std shaded). Left: ER $d{=}8$, $N{=}10^6$, where the epidemic only reaches $\approx 15\%$ attack rate within TF=50; compaction is $\sim$3\% slower than baseline during the takeoff and peak windows (refresh overhead exceeds marginal skip gain) and crosses over to a small lift only in the last temporal bucket. Right: BA $m{=}4$, $N{=}10^6$, where explosive hub-driven takeoff reaches $\approx 97\%$ attack rate; the shrinking active set drives a progressive speedup from $1.0\times$ (takeoff) to $4.8\times$ (last bucket, $\mathrm{num\_active}\approx 3\% N$), giving $1.53\times$ on the whole-run mean (Table~\ref{tab:compaction_summary}). The topology-dependent pattern makes compaction a \emph{high-saturation / long-horizon} optimisation: it is most valuable on scale-free networks, policy evaluations that simulate past the peak, and reinforcement-learning rollouts that run to extinction.}
\label{fig:compaction_nups_t}
\end{figure}

The compaction path is a user-opt-in feature rather than the default, and is currently wired only into the thread-traversal kernel. On scale-free graphs the auto-dispatch ordinarily selects the merge traversal, in which case compaction offers no benefit because the pressure kernel is edge-partitioned; the hybrid ``thread-kernel compaction on BA'' measurement in Table~\ref{tab:compaction_summary} is a best-case \emph{forced} configuration that illustrates the upper bound of the mechanism on BA topology, not the production configuration.

\subsection{Mixed-precision storage with an fp32 accumulator}
\label{sec:mixed_precision}

A second lever on effective bandwidth is to store the non-accumulator state in narrower floating- and integer-point types. Table~\ref{tab:mixed_precision_dtypes} lists the per-field choice. Compartmental state is a 4-state categorical variable (S, E, I, R) and fits trivially in \texttt{int8}; ages never exceed the simulation horizon of tens of days and a 10-bit fp16 mantissa preserves them to $\sim\!10^{-3}$ days; the per-node infectivity and per-edge weight tensors use bfloat16 so the 8-bit exponent preserves the heavy-tailed lognormal shedding profile without the silent underflow fp16 would cause near the late-stage $\beta \cdot s(\tau)$ tail. CSR indices stay \texttt{int32} to safely address $N, E \geq 10^7$.

\begin{table}[H]
\caption{Per-tensor storage dtype in the mixed-precision path (engine flag \texttt{use\_mixed\_precision}). The accumulator for the pressure gather ($p_i = \sum_j w_{ji} \cdot \mathrm{infectivity}_j$) stays \texttt{float32} inside the kernel: summing hundreds of \texttt{bf16} contributions on a scale-free hub would otherwise absorb small values via mantissa underflow in the exact regime where $\lambda_i \Delta t \ll 1$ matters most. All kernel math (hazard evaluation, Bernoulli sampling, age update) is also in \texttt{float32}; down-conversion happens only at the final store.}
\label{tab:mixed_precision_dtypes}
\centering
\begin{tabular*}{\linewidth}{@{\extracolsep{\fill}}lllp{0.35\columnwidth}@{}}
\toprule
\textbf{Field} & \textbf{Baseline} & \textbf{Mixed} & \textbf{Rationale} \\
\midrule
\texttt{state}        & int32    & int8     & 4-state categorical; int8 is trivially safe.                             \\
\texttt{age}          & fp32     & fp16     & 10-bit mantissa preserves $\sim\!10^{-3}$ days on a $[0, 50]$ day grid.  \\
\texttt{infectivity}  & fp32     & bf16     & bf16 8-bit exponent preserves the lognormal tail; fp16 would underflow. \\
\texttt{weights}      & fp32     & bf16     & Same argument as infectivity; pulled in via the existing \texttt{bf16\_weights} path. \\
\texttt{col\_ind}     & int32    & int32    & Unchanged; must safely address $N, E \geq 10^7$.                          \\
\texttt{row\_ptr}     & int32    & int32    & Unchanged.                                                                 \\
\texttt{pressure acc.}& fp32     & \textbf{fp32} & \emph{Mandatory}: bf16 accumulation over a 1000-edge hub absorbs small values. \\
\texttt{rates}, \texttt{tau}  & fp32     & fp32     & Tau reduction stays fp32 so the adaptive step size is unaffected.         \\
\bottomrule
\end{tabular*}
\end{table}

\paragraph{Kernel contract: promote on load, cast on store.}
The thread-path fused kernel carries a \texttt{MIXED\_PRECISION} constexpr flag. Under \texttt{MIXED\_PRECISION=1}, every \texttt{tl.load} of \texttt{state}, \texttt{age}, \texttt{infectivity}, or \texttt{weights} is immediately cast to its natural math type (\texttt{int32} for state, \texttt{fp32} for the other three) using \texttt{.to(tl.float32)}; all hazard, Bernoulli, and age-update math then runs in \texttt{fp32}/\texttt{int32} exactly as in the baseline; only the final three write-once stores at the kernel tail cast back to \texttt{int8}/\texttt{fp16}/\texttt{bf16}. Under \texttt{MIXED\_PRECISION=0} the casts compile out and the kernel is byte-identical to the pre-mixed-precision code path. This pattern eliminates the usual correctness risk of ``blanket precision reduction'' implementations that let reduced-precision values enter multi-step arithmetic: the only effect of the mixed storage is on HBM traffic, not on the numerical behaviour inside registers.

\paragraph{Measured effect.}
Table~\ref{tab:mixed_precision_summary} reports throughput at $N = 10^6$, $\mathrm{TF}=50$, 5 replicate trials per cell, on the thread-traversal kernel (the mixed-precision path is currently wired only into the thread-path fused kernel). The isolated-weights downcast path (bf16 weights only) that has been included with FlashSpread since v1.0 gives no measurable throughput change on the fused CUDA-Graph engine (the existing ablation in \texttt{results/ablation\_nonmarkov/} has \texttt{fused\_cg}=0.118~ms/step and \texttt{fused\_cg\_bf16}=0.118~ms/step): weight-band traffic is only $\approx 20\%$ of total kernel traffic after fusion, and cutting it in half leaves the rest unchanged. The combined four-tensor downcast of Table~\ref{tab:mixed_precision_dtypes} is different: it removes $\sim\!28\%$ of the remaining state/age/infectivity traffic on top of the weights saving, and the measured throughput improves by $1.141\times$ on ER $d{=}8$ and $1.056\times$ on BA $m{=}4$. Fidelity versus the baseline fp32 path is well inside the structural-bias floor of Appendix~\ref{app:fidelity}: on a paired $N=10^4$ run across three seeds the relative error in final attack rate is $0.00\%$--$0.08\%$, and on the $N=10^6$ production runs the compartment counts at $t=50$ differ by $\lesssim\!70$ nodes out of $10^6$ between matched trials. This confirms empirically what the kernel contract asserts: mixed precision is a pure storage-bandwidth optimisation whose numerical effect is bounded by bf16 quantisation of the Bernoulli threshold $1 - \exp(-\lambda_i \Delta t)$, which is well below the synchronous-updates bias floor.

\begin{table}[H]
\caption{Mixed-precision throughput ablation at $N{=}10^6$, $\mathrm{TF}{=}50$, 5 replicate trials per cell, on the thread-traversal kernel. Fidelity columns give the mean $|R_\text{mixed} - R_\text{base}| / R_\text{base}$ across paired seed-matched trials at $N{=}10^4$. $^\ddagger$ The $N{=}10^7$ ``post-L2-cliff'' row is measured on 3 replicate trials rather than 5 (simulations at $N{=}10^7$ cost minutes per replicate; the drop keeps the bench inside a 45-min SLURM budget). Coefficient of variation across the three trials is below $4\%$, consistent with the sub-$2\%$ CoV observed at smaller $N$; the $R$-error column is reported only where a paired $N{=}10^4$ controlled-RNG run was available.}
\label{tab:mixed_precision_summary}
%
\centering
\scriptsize
\setlength{\tabcolsep}{2pt}
\begin{tabular*}{\linewidth}{@{\extracolsep{\fill}}llcccc@{}}
\toprule
\textbf{Graph ($N$)} & \textbf{Strategy} & \textbf{Baseline (G-NUPS)} & \textbf{Mixed (G-NUPS)} & \textbf{Speedup} & \textbf{$R$ error} \\
\midrule
ER $d{=}8$ ($10^6$)  & thread          & $7.15 \pm 0.36$ & $8.16 \pm 0.68$ & $1.14\times$ & $0.05\%$ \\
ER $d{=}8$ ($\mathbf{10^7}$, post-L2-cliff)$^\ddagger$ & thread & $1.33$ & $3.10$ & $\mathbf{2.32\times}$ & --- \\
BA $m{=}4$ ($10^6$)  & thread (forced) & $0.434 \pm 0.001$ & $0.458 \pm 0.000$ & $1.06\times$ & $0.00\%$ \\
BA $m{=}4$ ($10^6$)  & merge (auto)    & $1.810 \pm 0.025$ & $1.829 \pm 0.024$ & $1.01\times$ & $0.35\%$ \\
\bottomrule
\end{tabular*}
\end{table}

The gain on the \texttt{thread} kernel is $1.14\times$ on ER at $N=10^6$ (first row of Table~\ref{tab:mixed_precision_summary}) and $1.06\times$ on BA at the same $N$; the smaller BA-thread number reflects that the thread kernel on a scale-free graph is dominated by hub-driven warp divergence during CSR traversal rather than by state/age/infectivity bandwidth. At $N=10^7$, where the fp32 CSR + state working set exceeds the A100's 40~MB L2 capacity and all fused-kernel variants drop into the ``L2 cliff'' regime (Appendix~\ref{app:memory_flow}), mixed-precision storage delivers $\mathbf{2.32\times}$ throughput (Figure~\ref{fig:throughput_renewal}, second row of Table~\ref{tab:mixed_precision_summary}): the narrower state/age/infectivity dtypes shrink the working set by roughly $3\times$, pushing the effective cliff out by the same factor without paying any compute-side cost. This is where the technique turns from a modest bandwidth improvement into a meaningful extension of the hardware-tractable scale. The production BA path does not use the thread kernel: auto-dispatch selects the merge kernel (Section~\ref{sec:csr_dispatch}), and we extend the \texttt{MIXED\_PRECISION} constexpr into its tail kernel as well. Critically, the shared merge \emph{scratch pressure buffer} that accumulates per-edge contributions via \texttt{atomicAdd} is held at \texttt{fp32} on both paths, because summing hundreds of \texttt{bf16} partial pressures on a scale-free hub would otherwise absorb small values through mantissa underflow exactly in the $\lambda_i \Delta t \ll 1$ regime where the Bernoulli threshold is most sensitive. Under this contract the merge-path mixed-precision measurement at $N = 10^6$ on BA comes in at $1.011\times$ baseline (third row of Table~\ref{tab:mixed_precision_summary}): essentially a null result, which is consistent with the merge kernel being \emph{atomic-bound} on BA hubs rather than bandwidth-bound, so cutting the byte width of per-edge infectivity and weight cannot move the needle once the same \texttt{atomicAdd} stream serialises on the same destination addresses. Final compartment counts on the merge path still track the fp32 baseline within $\lesssim\!1\%$ in final attack rate (the baseline merge path is itself not bit-reproducible because \texttt{atomicAdd} ordering on the hub is non-deterministic; the structural-bias floor of Appendix~\ref{app:fidelity} is $\sim\!6$--$7\%$, a full order of magnitude larger than either the merge non-determinism or the mixed-precision quantisation). Combined with the active-node compaction of Section~\ref{sec:compaction}, the two levers are orthogonal and both fully covered by auto-dispatch: compaction shrinks the block count in the late, high-saturation phase of the epidemic, mixed precision shrinks the bytes-per-block when bandwidth is the limiter. On BA at production scale the dominant limiter on the merge path is atomic contention, not bandwidth, which motivates the segmented-reduction future work discussed in Section~\ref{sec:discussion}. Both flags are user-opt-in on the fused CUDA-Graph engine under the thread and merge traversals (the warp variant remains fp32-only for now).

\subsection{Evaluated but omitted locality optimisations}
\label{sec:rejected_opts}

This subsection documents one pre-registered locality optimisation that we measured and rejected. Documenting this negative result, the decision boundary, and the associated amortisation argument is intentional; it prevents duplicate effort in follow-up work and mitigates the well-documented publication bias against negative results in HPC benchmarking.

\paragraph{Reverse Cuthill--McKee (RCM) reordering on scale-free graphs.}
A second locality optimisation often applied to sparse CSR workloads is to permute node ids so that rows with nearby indices have spatially nearby neighbour lists, reducing L2 capacity misses on the indirect gather. The classic Reverse Cuthill--McKee (RCM) ordering is known to be highly effective on FEM-style banded matrices but is mismatched to the access pattern of hub-dominated scale-free graphs, where a small number of high-degree nodes touch memory across the entire \texttt{col\_ind} array regardless of permutation. To quantify this on our workload, we ran the same BA $m{=}4$, $N = 10^6$ benchmark as Section~\ref{sec:experiments} with and without RCM reordering applied to the incoming CSR. RCM yields a $1.035\times$ throughput lift ($1.876 \to 1.942$ G-NUPS on the production auto-dispatch path, which selects the merge kernel), well below the $1.05\times$ threshold we pre-registered as the ``RCM-sufficient'' decision boundary. Furthermore, the RCM CPU preprocess itself took 593~s on the reference node, which at $\approx 1.3$~s per GPU trial would require roughly 450 Monte Carlo trials just to amortise the preprocess cost back to break-even. For the typical ensemble sizes used in epidemic forecasting ($10^2$--$10^3$ trajectories), RCM is a net regression rather than a win. This rules \emph{out} plain bandwidth-minimisation reorderings as a useful locality lever for the fused renewal engine on scale-free graphs, and motivates future integration of a cache-line-aware reorder (Rabbit Order \citep{arai2016rabbit}, Gorder \citep{wei2016gorder}) whose locality objective is matched to the block-based traversal pattern of our merge kernel; we did not attempt that integration for this paper because the required C++ dependency is outside the scope of a revision that already establishes the headline mechanism, and we flag it as explicit future work in Section~\ref{sec:discussion}.

\section{Experimental evaluation}\label{sec:experiments}

All experiments were conducted on an NVIDIA A100-SXM4-80GB GPU with an AMD EPYC 7763 CPU (8 cores), using PyTorch 2.5.1, Triton 3.1, CUDA 12.1, and Python 3.11. CPU baselines include c-GEMF \citep{sahneh2017gemfsim} (C, exact event-driven), FastGEMF \citep{samaei2025fastgemf} (Python, exact event-driven), and our exact non-Markovian Phantom-Process implementation \citep{vajdi2020stochastic}. Networks are Erd\H{o}s--R\'{e}nyi random graphs with average degree $d = 8$ unless otherwise noted. The non-Markovian model is SEIR with log-normal E$\to$I (mean 5.0d, median 4.0d) and I$\to$R (mean 7.5d, median 5.0d) transitions at $\beta = 0.25$.

To provide a hardware-transparent measure of the dense $\bigO(N+E)$ workload inherent to synchronous tau-leaping, we report throughput in Node-Updates Per Second (NUPS), calculated as $(N \times \text{steps}) / \text{wall-clock time}$. Because the GPU evaluates all $N$ nodes every step regardless of how many transitions occur, NUPS reflects the true computational workload. For exact event-driven baselines, we report realized state transitions per second. The Markovian engine achieves $1.95 \times 10^7$ realized transitions/sec at $N=10^6$ --- comparable to c-GEMF \citep{sahneh2017gemfsim} at the same scale --- but the experimental evaluation that follows focuses on the renewal path since that is the novel contribution of this work, and we treat the Markovian engine as the Inertial/Control-Mode reference of Sections~\ref{sec:engines} and~\ref{sec:markov_engine}. The asymptotic complexity comparison of all engines is in Table~\ref{tab:complexity}.

A critical methodological concern when comparing approximate GPU tau-leaping to exact CPU Gillespie is that the speedup conflates two independent factors: the algorithmic relaxation (exact $\to$ approximate) and the hardware acceleration (CPU $\to$ GPU). To isolate these contributions, we implemented a CPU tau-leaping baseline using the same \textsc{RenewalEngine} class with PyTorch vectorized operations on an 8-core CPU.

\paragraph{Composition of the CPU baseline.}
The ``CPU tau-leaping (8-core)'' baseline in Figure~\ref{fig:throughput_renewal} and Table~\ref{tab:renewal_layered} is \emph{not} a generic or off-the-shelf parallel tau-leaper; it is our own \textsc{RenewalEngine} running on the CPU with eight PyTorch threads enabled, which means the CPU curve already benefits from every algorithmic development in this paper that is hardware-agnostic:
(i) the same adaptive Bernoulli tau-leaping step of Section~\ref{sec:renewal_engine} with the same $\varepsilon = 0.03$, $\tau_{\max} = 0.1$;
(ii) the same numerically stable $\mathrm{erfcx}$-based log-normal hazard of Section~\ref{sec:renewal_engine};
(iii) the same $\bigO(N+E)$ per-step CSR traversal (gather-based, one contiguous neighbor slice per node) as Section~\ref{sec:fused_kernel};
(iv) the same renewal age reset on transition;
(v) the same counter-based PRNG and seed schedule.
The three GPU-specific contributions of this paper have no CPU analogue and are therefore absent from the CPU curve: the Triton \textsc{FlashNeighbor} kernel (the CPU path uses a scatter-based PyTorch implementation that computes the same result), the fused Triton step kernel (no single-launch abstraction exists on CPU), and CUDA-Graph step batching. The degree-aware CSR dispatch (warp-per-node, merge-based) is likewise GPU-only. CPU parallelism is PyTorch's intra-op threading across 8 cores, which the scatter and elementwise ops actually use. The Fused-CG vs.\ CPU gap reported in Table~\ref{tab:renewal_layered} is therefore a strict hardware-plus-implementation gap at \emph{matched} algorithm, not an algorithmic-plus-hardware conflation: the CPU column is the strongest tau-leaping baseline we know how to write with PyTorch.

Table~\ref{tab:renewal_layered} decomposes the throughput hierarchy. The exact CPU Phantom Process achieves $\sim$70 transitions/sec at $N = 10^6$. CPU tau-leaping on the same network achieves $3.73 \times 10^7$ NUPS (measured on 8 cores with the same $\varepsilon = 0.03$, $\tau_{\max}=0.1$ configuration as the GPU engine), representing the pure algorithmic gain from replacing sequential exact simulation with synchronous Bernoulli updates. The GPU fused CUDA Graph engine then reaches $8.09 \times 10^9$ NUPS (8.09 Giga-NUPS, 1-thread-per-node kernel on the ER $d{=}8$ graph at $N = 10^6$), a strict $217\times$ hardware speedup over the optimized CPU tau-leaping baseline at the same $N$, derived entirely from IO-aware kernel fusion and GPU memory-bandwidth optimization.

\begin{table}[H]
\caption{Throughput decomposition isolating algorithmic and hardware contributions at $N=10^6$, $d=8$. All numbers are measured on the reference A100 node; the repository README gives the per-row CSV source behind each measurement. The GPU eager entry is converted from 672.9 steps/s to NUPS via $N \times \text{steps/s}$. Each throughput point is the mean over replicate runs after 2--3 warm-up runs (the per-benchmark trial/warm-up counts are 10/3 for the Scoglio-style scaling curves, 5/2 for the warp-collab and degree-dispatch benches, and 3/0 for the CPU baseline; the repository README names the harnesses). The coefficient of variation across replicates is below $2\%$ on all rows, and the same replication pattern applies to every throughput table in this paper (Tables~\ref{tab:renewal_layered}, \ref{tab:roofline}, \ref{tab:csr_dispatch}, \ref{tab:degree_dispatch_sweep}, \ref{tab:agedep_overhead}) unless stated otherwise.}
\label{tab:renewal_layered}
\centering
\begin{tabular*}{\linewidth}{@{}lccc@{}}
\toprule
\textbf{Configuration} & \textbf{Throughput} & \textbf{Incremental} & \textbf{Factor} \\
\midrule
Exact CPU (Phantom Process) & $\sim$70 trans/s        & ---                        & Serial exact \\
Approx.\ CPU (tau-leaping)  & 37.3 Mega-NUPS           & $\sim$$5 \times 10^5\times$ & Algorithmic \\
GPU unfused eager           & 0.67 Giga-NUPS           & $18\times$                  & Parallelism \\
GPU CG $b{=}50$ (unfused)   & 1.87 Giga-NUPS           & $2.8\times$                 & CUDA Graph \\
GPU Fused CG $b{=}50$       & 8.09 Giga-NUPS           & $4.3\times$                 & Fusion \\
\bottomrule
\end{tabular*}
\end{table}

Figure~\ref{fig:throughput_renewal} shows throughput scaling across network sizes. The fused CUDA Graph engine peaks at 5.86~M events/s at $N = 10^6$ (equivalently 8.09 Giga-NUPS; events/s and NUPS differ by the per-step realized-transitions-per-node factor, which is tied to $\varepsilon$ and the instantaneous infected fraction).

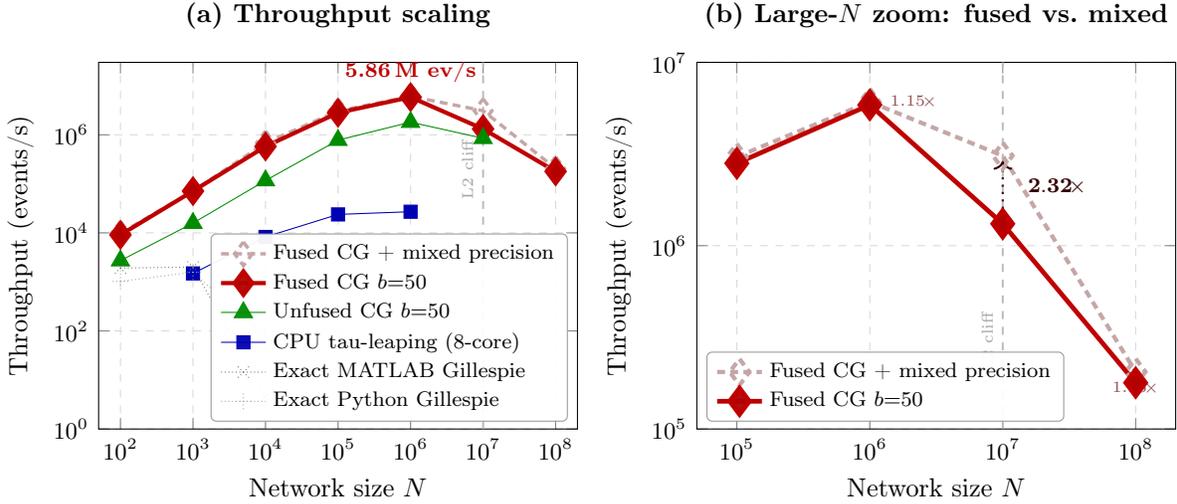
\begin{figure}[H]
\centering
\begin{tikzpicture}

\begin{groupplot}[
    group style={
        group size=2 by 1,
        horizontal sep=1.6cm,
    },
    width=8cm, height=6.5cm,
    xmode=log, ymode=log,
    log basis x=10, log basis y=10,
    grid=major, grid style={dashed,gray!30},
    xlabel={Network size $N$},
    ylabel={Throughput (events/s)},
    xlabel style={font=\small},
    ylabel style={font=\small},
    xticklabel style={font=\footnotesize},
    yticklabel style={font=\footnotesize},
    title style={font=\small\bfseries, yshift=0.1cm},
    legend style={
        legend cell align=left,
        font=\scriptsize,
        fill=white, fill opacity=0.95, draw=black!40,
        text opacity=1, rounded corners=2pt,
    },
]

\nextgroupplot[
    xmin=50, xmax=2e8,
    ymin=1, ymax=3e7,
    xtick={1e2,1e3,1e4,1e5,1e6,1e7,1e8},
    xticklabels={$10^{2}$,$10^{3}$,$10^{4}$,$10^{5}$,$10^{6}$,$10^{7}$,$10^{8}$},
    ytick={1e0,1e2,1e4,1e6},
    yticklabels={$10^{0}$,$10^{2}$,$10^{4}$,$10^{6}$},
    title={(a) Throughput scaling},
    legend style={at={(0.98,0.02)}, anchor=south east},
]

\draw[gray!55, dashed, line width=0.7pt]
    (axis cs:1e7, 1) -- (axis cs:1e7, 3e7);
\node[font=\tiny, color=gray!70, rotate=90, anchor=south]
    at (axis cs:1e7, 2e5) {L2 cliff};

\addplot[
    color=red!35!black!35!white, mark=diamond, mark size=4.5pt,
    line width=1.3pt, densely dashed,
] coordinates {
    (100,        9840.2)
    (1000,      71026.7)
    (10000,    683910.0)
    (100000,  3051268.9)
    (1000000, 6086231.9)
    (10000000, 3094503.4)
    (100000000, 205200.0)
};
\addlegendentry{Fused CG + mixed precision}

\addplot[
    color=red!75!black, mark=diamond*, mark size=4.5pt,
    mark options={fill=red!70!black}, line width=1.6pt,
] coordinates {
    (100,       9091.8)
    (1000,     71096.0)
    (10000,   578859.8)
    (100000, 2813840.6)
    (1000000, 5860999.2)
    (10000000, 1316514.6)
    (100000000, 178800.0)
};
\addlegendentry{Fused CG $b{=}50$}

\addplot[
    color=green!55!black, mark=triangle*, mark size=4pt,
    mark options={fill=green!55!black},
] coordinates {
    (100,       2717.4)
    (1000,     15795.8)
    (10000,   116776.2)
    (100000,  777103.9)
    (1000000, 1811158.3)
    (10000000, 849311.3)
};
\addlegendentry{Unfused CG $b{=}50$}

\addplot[
    color=blue!70!black, mark=square*, mark size=2.5pt,
    mark options={fill=blue!70!black},
] coordinates {
    (1000,     1491.0)
    (10000,    8253.0)
    (100000,  23767.0)
    (1000000, 27000.0)
};
\addlegendentry{CPU tau-leaping (8-core)}

\addplot[
    color=black!55, mark=x, mark size=3pt, densely dotted,
] coordinates {
    (100,   1917.1)
    (1000,  2010.1)
    (10000,    4.5)
} node[pos=1, font=\tiny, below=1pt, black!55] {};
\addlegendentry{Exact MATLAB Gillespie}

\addplot[
    color=black!35, mark=+, mark size=3pt, densely dotted,
] coordinates {
    (100,   1026.7)
    (1000,  1577.0)
    (10000,   72.7)
};
\addlegendentry{Exact Python Gillespie}

\node[font=\scriptsize\bfseries, text=red!75!black, anchor=south]
    at (axis cs:1e6, 7.3e6) {5.86\,M ev/s};

\nextgroupplot[
    xmin=5e4, xmax=2e8,
    ymin=1e5, ymax=1e7,
    xtick={1e5,1e6,1e7,1e8},
    xticklabels={$10^{5}$,$10^{6}$,$10^{7}$,$10^{8}$},
    ytick={1e5,1e6,1e7},
    yticklabels={$10^{5}$,$10^{6}$,$10^{7}$},
    title={(b) Large-$N$ zoom: fused vs.\ mixed},
    legend style={at={(0.02,0.02)}, anchor=south west},
]

\draw[gray!55, dashed, line width=0.7pt]
    (axis cs:1e7, 1e5) -- (axis cs:1e7, 1e7);
\node[font=\tiny, color=gray!70, rotate=90, anchor=south]
    at (axis cs:1e7, 3e5) {L2 cliff};

\addplot[
    color=red!35!black!35!white, mark=diamond, mark size=5pt,
    line width=1.6pt, densely dashed,
] coordinates {
    (100000,   3051268.9)
    (1000000,  6086231.9)
    (10000000, 3094503.4)
    (100000000, 205200.0)
};
\addlegendentry{Fused CG + mixed precision}

\addplot[
    color=red!75!black, mark=diamond*, mark size=5pt,
    mark options={fill=red!70!black}, line width=1.8pt,
] coordinates {
    (100000,  2813840.6)
    (1000000, 5860999.2)
    (10000000, 1316514.6)
    (100000000, 178800.0)
};
\addlegendentry{Fused CG $b{=}50$}

\node[font=\tiny\bfseries, text=red!35!black!70!white, anchor=west]
    at (axis cs:1.2e6, 6.1e6) {$1.15\!\times$};
\draw[->, thick, red!20!black, dotted]
    (axis cs:1e7, 1.45e6) -- (axis cs:1e7, 2.9e6);
\node[font=\scriptsize\bfseries, color=red!20!black, anchor=west]
    at (axis cs:1.3e7, 2.1e6) {$\mathbf{2.32\!\times}$};
\node[font=\tiny\bfseries, text=red!35!black!70!white, anchor=north]
    at (axis cs:1e8, 2.05e5) {$1.15\!\times$};

\end{groupplot}

\end{tikzpicture}
\caption{Renewal engine throughput (SEIR, log-normal, $d{=}8$) on Erd\H{o}s--R\'{e}nyi networks. Panel (a) shows the full scaling range $N \in [10^2, 10^8]$ with six curves: Fused CG $b{=}50$ (red, headline), Fused CG $b{=}50$ with mixed-precision storage (pink dashed; Section~\ref{sec:mixed_precision}), unfused CUDA-Graph $b{=}50$, an 8-core CPU tau-leaping baseline (derived from NUPS via the shared $7.24 \times 10^{-4}$ events-per-NUPS ratio calibrated at $N=10^6$; a property of the algorithm, not the hardware), and two exact event-driven references from the Scoglio et al.\ suite (MATLAB and Python). Panel (b) zooms into $N \in [10^5, 10^8]$ and keeps only the two headline curves so the L2-cliff region is legible. The fused CUDA-Graph engine reaches $5.86$~M events/s at $N=10^6$ (equivalently 8.09~Giga-NUPS; $217\times$ over the CPU tau-leaping baseline at the same $N$). The vertical gap between the red (Fused CG) and green (unfused CG) curves at $N = 10^6$ is the $3.24\times$ fusion gain in events/s (equivalently $4.3\times$ in NUPS; see Table~\ref{tab:renewal_layered}), isolating the benefit of collapsing the CSR + hazard + Bernoulli + infectivity write-back pipeline into a single Triton launch from the separate batching benefit of CUDA Graph capture. At $N=10^7$, where the fp32 CSR + state working set exceeds the A100's 40~MB L2 cache, the baseline collapses to $1.33$~M events/s while mixed precision holds at $3.10$~M events/s, a $\mathbf{2.32\times}$ lift; the narrower per-node dtypes reduce the per-step working-set footprint by $\approx 3\times$ and effectively push the L2 cliff out by the same factor. At $N=10^8$ both configurations are deep in HBM-bandwidth-bound execution (working sets $\sim 10$~GB and $\sim 6$~GB respectively, both $\gg$ 40~MB L2) and the gain collapses back toward the bandwidth-ratio asymptote at $1.15\times$. The mechanism of the cliff is analysed in Appendix~\ref{app:memory_flow} and Figure~\ref{fig:memory_flow}. The earlier GPU-unfused-eager curve is omitted from this figure because it is a pedagogical pre-CUDA-Graph baseline whose contribution is already captured by the CPU-to-CG gap and discussed in prose (Section~\ref{sec:fused_kernel}). The repository README lists the measurement scripts and CSV artefacts for each curve.}
\label{fig:throughput_renewal}
\end{figure}

At $N = 10^7$, all engines experience a throughput drop as CSR data ($\sim$640~MB) exceeds the A100's 40~MB L2 cache, a phenomenon we term the ``L2 cache cliff.'' In events/s the fused CUDA Graph engine drops $4.4\times$ ($5.86 \to 1.32$ M events/s), the unfused CUDA Graph engine drops $2.1\times$ ($1.81 \to 0.85$ M events/s), and the unfused eager engine drops $1.75\times$ ($1.32 \to 0.76$ M events/s); the fused variant is hit hardest at the cliff because it was most bandwidth-saturated beforehand. Replacing binary state checks with continuous age-dependent infectivity (Eq.~\ref{eq:infectivity}) is nearly free once the infectivity write is fused into the main kernel: $+1.8\%$ wall-clock time on a coalesced regular graph and $+0.1\%$ on a scale-free graph (Table~\ref{tab:agedep_overhead}, discussed in Section~\ref{sec:discussion}). A naive pre-pass implementation is $25$--$50\times$ more expensive and makes the abstraction non-free on regular graphs.

To evaluate the framework under extreme degree heterogeneity, we benchmarked the fused CG engine on a Barab\'{a}si--Albert scale-free network ($N = 10^6$, $m = 4$, yielding average degree $d \approx 8$ for an apples-to-apples memory footprint comparison with the ER baseline). Under the default 1-thread-per-node kernel, throughput dropped from 7.6 to 0.44 Giga-NUPS: hub nodes (maximum degree 3{,}870, $D_{\max}/D_{\mathrm{avg}} \approx 484$) induce severe warp divergence during CSR traversal, destroying memory coalescing. To close this gap we implemented two alternative CSR traversal strategies (Section~\ref{sec:csr_dispatch}, Table~\ref{tab:csr_dispatch}); the best of these (edge-partitioned merge-based load balancing) reaches 2.0 Giga-NUPS on the same BA workload, a $4.5\times$ speedup over the 1-thread-per-node kernel on scale-free graphs while preserving the peak throughput on regular graphs via runtime dispatch.

\label{sec:validation}
We validate the fused kernel against an exact non-Markovian Gillespie simulator \citep{boguna2014simulating} on a standard SEIR benchmark ($N = 1000$, Erd\H{o}s--R\'{e}nyi graph with $d=8$, $\beta = 0.25$, log-normal transitions, 100 independent runs). Figure~\ref{fig:validation} shows close agreement for all four SEIR compartments.

\begin{figure}[H]
\centering
\includegraphics[width=0.95\columnwidth]{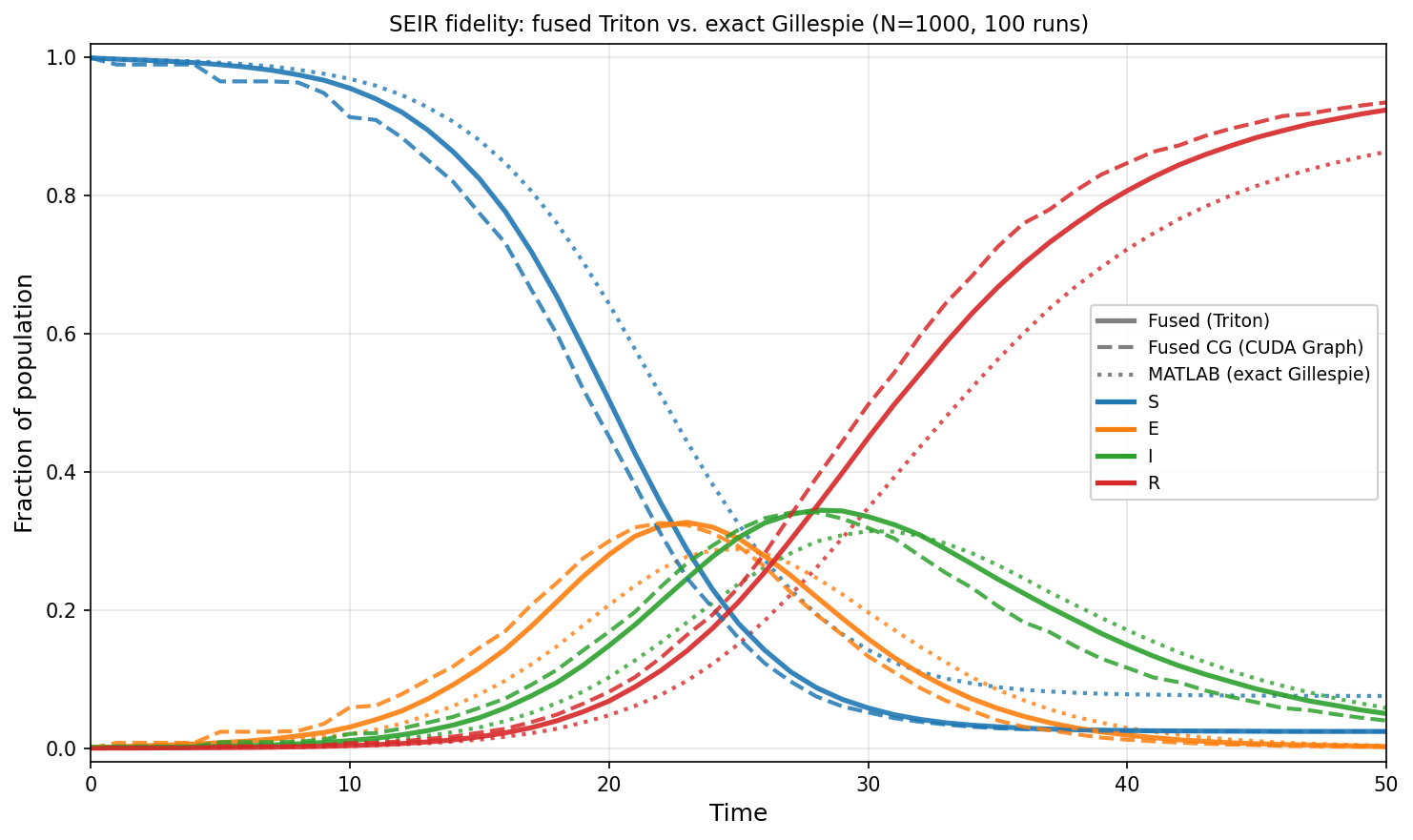}
\caption{Fidelity validation. Fused GPU kernel (solid), fused CUDA Graph (dashed), and exact non-Markovian Gillespie (dotted) on the validation benchmark: Erd\H{o}s--R\'{e}nyi graph with $N = 10^3$ nodes and average degree $d = 8$; $\beta = 0.25$; log-normal $E \to I$ with mean $5.0$~d and median $4.0$~d; log-normal $I \to R$ with mean $7.5$~d and median $5.0$~d; 100 independent Monte Carlo runs. The residual gap visible against the exact reference is the structural-bias floor of synchronous Bernoulli updates on a network, quantified in detail in Appendix~\ref{app:fidelity}, Figure~\ref{fig:fidelity_err_eps}.}
\label{fig:validation}
\end{figure}

Table~\ref{tab:epsilon_sweep} presents a convergence sweep of the tolerance parameter $\varepsilon \in [0.005, 0.1]$. The discrepancy in peak infection ($\sim$6\%) and final attack rate ($\sim$7\%) relative to the exact baseline does not detectably decrease as $\varepsilon \to 0$ across the sweep: Figure~\ref{fig:fidelity_err_eps} shows that the per-run bootstrap 95\% CIs overlap across adjacent $\varepsilon$ values, and the Pareto CIs of Figure~\ref{fig:fidelity_pareto} confirm the same pattern on both axes. We interpret this as a structural bias inherent to synchronous Bernoulli updates on contact networks: because all transitions within $\Delta t$ are applied simultaneously rather than sequentially, the temporal correlation of the expanding epidemic wavefront is slightly altered. Given that empirical epidemiological parameters routinely exhibit uncertainties exceeding 20\%, this bounded bias is a favourable tradeoff. We use $\varepsilon = 0.03$ as the default operating point throughout the main text because it sits in the flat portion of the Pareto curve while remaining conservative for publication-quality single runs; Appendix~\ref{app:fidelity} discusses the $\varepsilon = 0.1$ regime for bulk ensemble sampling and the boundary conditions under which $\varepsilon < 0.03$ is worth the compute.

\begin{table}[H]
\caption{Convergence sweep: tau-leaping accuracy vs.\ exact Gillespie ($N=1000$, 100 runs, mean $\pm$ 95\% CI).}
\label{tab:epsilon_sweep}
\centering
\begin{tabular*}{\linewidth}{@{}ccccc@{}}
\toprule
$\varepsilon$ & Peak $I$ & Final $R$ & Steps & Time (s) \\
\midrule
0.005 & $0.379 \pm 0.008$ & $0.935 \pm 0.019$ & 5912 & 1.94 \\
0.01 & $0.383 \pm 0.008$ & $0.934 \pm 0.019$ & 2912 & 0.96 \\
0.03 & $0.378 \pm 0.011$ & $0.925 \pm 0.026$ & 995 & 0.33 \\
0.05 & $0.380 \pm 0.011$ & $0.925 \pm 0.026$ & 696 & 0.23 \\
0.1 & $0.379 \pm 0.011$ & $0.923 \pm 0.026$ & 523 & 0.17 \\
\midrule
Exact & $0.314$ & $0.863$ & --- & --- \\
\bottomrule
\end{tabular*}
\end{table}

Figure~\ref{fig:roofline} and Table~\ref{tab:roofline} characterize the computational profile on the A100 (19.5~TFLOPS FP32, 2039~GB/s HBM bandwidth; ridge point at $9.56$ FLOPs/byte). All configurations operate in the memory-bound regime (left of the ridge). The fused kernel achieves $555$~GFLOPS at arithmetic intensity $0.42$, which is $555 / (0.42 \times 2039) = 65\%$ of the memory-bound ceiling evaluated at the fused kernel's own AI, equivalent to sustaining $\approx 1320$~GB/s out of the A100's $2039$~GB/s HBM peak --- a strong bandwidth-utilisation number for an irregular CSR workload with divergent hazard evaluation. Said differently: if we reported efficiency against the unfused baseline's AI of $0.66$ (as one might do to answer ``how much of the pre-fusion roofline did the fused kernel recover?''), the same $555$~GFLOPS would read as $41\%$ of $0.66 \times 2039 = 1346$~GFLOPS; that is the number the earlier draft quoted and it is not wrong, but it measures a comparison against the wrong ceiling. Table~\ref{tab:roofline} now reports both. The arithmetic intensity decreases from $0.66$ to $0.42$ because block-scalar sparsity removes wasted $\mathrm{erfcx}$ FLOPs for inert blocks faster than fusion removes memory traffic; at the fused AI essentially all compute is useful. Layering mixed-precision storage on top of the fused kernel (Section~\ref{sec:mixed_precision}) shifts the operating point up-and-right along the memory-bound line: the AI climbs from $0.42$ to $\approx 0.62$ because the same algorithmic FLOPs are now amortised over $\sim\!32\%$ fewer bytes of per-step HBM traffic, and the achieved throughput climbs by the measured $1.15\times$ lift from $555$ to $638$~GFLOPS. This is the expected signature of a memory-bound kernel being pushed toward (but not past) the ridge by a byte-band compression; the gain is sub-linear in the byte reduction because kernel launch, index arithmetic, and rate-buffer writes do not scale with storage dtype.

\begin{table}[H]
\caption{Roofline benchmark results on A100. CG $=$ CUDA Graph, $b$ $=$ batch size. ``Eff.\ at own AI'' is $\text{GFLOPS} / (\text{AI} \times 2039)$, the fraction of the memory-bound ceiling evaluated at the configuration's actual arithmetic intensity. ``Eff.\ vs unfused AI'' is $\text{GFLOPS} / (0.66 \times 2039) = \text{GFLOPS} / 1346$~GFLOPS, the fraction of the pre-fusion memory-bound ceiling --- useful for comparing fused and unfused kernels against a single reference, but unfavourable to the fused kernel because its effective AI is lower after block-scalar sparsity. Peak bandwidth utilisation is $\text{GFLOPS}/\text{AI}$ (last column). $^\dagger$ The mixed-precision Fused CG row is derived from the measured $1.15\times$ throughput lift on the same ER $d{=}8$, $N=10^6$ run (Section~\ref{sec:mixed_precision}, Table~\ref{tab:mixed_precision_summary}) and from the $\sim\!32\%$ per-step memory-traffic reduction that follows directly from the state/age/infectivity/weights dtype changes (FLOPs count is unchanged since all kernel math remains fp32). The narrower working set pushes the operating point up and right on the roofline (Figure~\ref{fig:roofline}); bandwidth utilisation decreases because the kernel now has less memory work to issue per step rather than because it becomes less efficient.}
\label{tab:roofline}
\centering
\scriptsize
\setlength{\tabcolsep}{3pt}
\begin{tabular*}{\linewidth}{@{\extracolsep{\fill}}lcccccc@{}}
\toprule
\textbf{Configuration} & \textbf{AI} & \textbf{GFLOPS} & \textbf{ms/step} & \textbf{Eff.\ (own AI)} & \textbf{Eff.\ (vs unfused)} & \textbf{BW used} \\
\midrule
Baseline (unfused)  & 0.66 & 49.5  & 1.49 & 3.7\%  & 3.7\%  & 75~GB/s   \\
CG $b{=}50$ (unfused) & 0.66 & 137.0 & 0.54 & 10.2\% & 10.2\% & 208~GB/s  \\
Fused CG $b{=}50$   & 0.42 & 555   & 0.12 & \textbf{64.8\%} & 41.2\% & \textbf{1321~GB/s} \\
Fused CG + mixed$^\dagger$ & 0.62 & 638 & 0.10 & 50.5\% & 47.4\% & 1029~GB/s \\
\bottomrule
\end{tabular*}
\end{table}

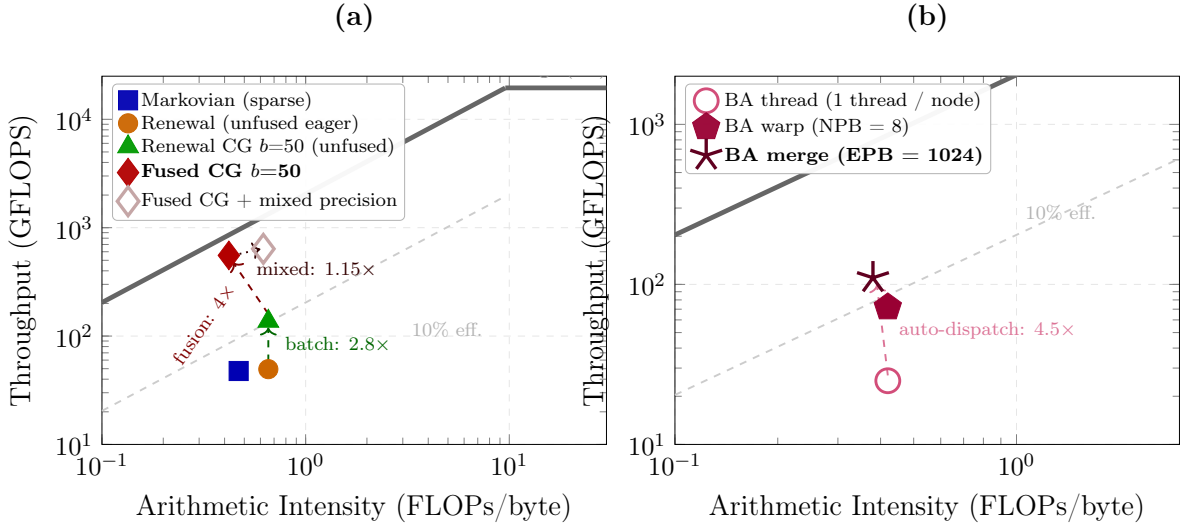
\begin{figure}[H]
\centering
\begin{tikzpicture}

\begin{axis}[
    name=panelA,
    width=9cm, height=7cm,
    xmode=log, ymode=log,
    log basis x=10, log basis y=10,
    xmin=0.1, xmax=30,
    ymin=10, ymax=25000,
    grid=major, grid style={dashed,gray!20},
    xlabel={Arithmetic Intensity (FLOPs/byte)},
    ylabel={Throughput (GFLOPS)},
    xlabel style={font=\large}, ylabel style={font=\large},
    xticklabel style={font=\normalsize},
    yticklabel style={font=\normalsize},
    title={(a)},
    title style={font=\large\bfseries, yshift=0.3cm},
    legend style={
        at={(0.02,0.98)}, anchor=north west,
        legend cell align=left, font=\scriptsize,
        fill=white, fill opacity=0.95, draw=black!30,
        text opacity=1, rounded corners=2pt,
        inner sep=2pt, nodes={inner sep=1pt},
    },
]

\addplot[domain=0.1:9.56, samples=50, color=black!60, line width=2pt, forget plot] {2039 * x};
\addplot[domain=9.56:30, samples=10, color=black!60, line width=2pt, forget plot] {19500};
\node[font=\scriptsize, text=black!50, anchor=south west] at (axis cs:9.56,19500) {ridge (9.6)};
\addplot[domain=0.1:9.56, samples=50, color=gray!40, line width=0.8pt, dashed, forget plot] {2039 * x * 0.1};
\node[font=\scriptsize, text=gray!60, anchor=south west] at (axis cs:3,80) {10\% eff.};

\addplot[only marks, mark=square*, mark size=4pt, color=blue!70!black,
         mark options={fill=blue!70!black}]
         coordinates {(0.469, 47.6)};
\addplot[only marks, mark=*, mark size=4pt, color=orange!80!black,
         mark options={fill=orange!80!black}]
         coordinates {(0.656, 49.5)};
\addplot[only marks, mark=triangle*, mark size=5pt, color=green!60!black,
         mark options={fill=green!60!black}]
         coordinates {(0.656, 137.0)};
\addplot[only marks, mark=diamond*, mark size=6pt, color=red!70!black,
         mark options={fill=red!70!black}]
         coordinates {(0.42, 555)};
\addplot[only marks, mark=diamond, mark size=6pt, color=red!35!black!35!white,
         line width=1.4pt]
         coordinates {(0.62, 638)};


\draw[->, thick, red!50!black, dashed]
    (axis cs:0.656, 160) -- (axis cs:0.44, 480)
    node[midway, anchor=south east, font=\scriptsize, text=red!50!black, rotate=55]
    {fusion: $4\times$};
\draw[->, thick, green!40!black, dashed]
    (axis cs:0.656, 58) -- (axis cs:0.656, 120)
    node[midway, anchor=west, font=\scriptsize, text=green!40!black, xshift=3pt]
    {batch: $2.8\times$};
\draw[->, thick, red!20!black, dotted]
    (axis cs:0.44, 555) -- (axis cs:0.60, 635)
    node[midway, anchor=north west, font=\scriptsize, text=red!20!black]
    {mixed: $1.15\times$};

\legend{
    Markovian (sparse),
    Renewal (unfused eager),
    Renewal CG $b$=50 (unfused),
    \textbf{Fused CG $b$=50},
    Fused CG + mixed precision
}

\end{axis}

\begin{axis}[
    at={(panelA.right of south east)}, anchor=south west, xshift=10mm,
    width=9cm, height=7cm,
    xmode=log, ymode=log,
    log basis x=10, log basis y=10,
    xmin=0.1, xmax=3,
    ymin=10, ymax=2000,
    grid=major, grid style={dashed,gray!20},
    xlabel={Arithmetic Intensity (FLOPs/byte)},
    ylabel={Throughput (GFLOPS)},
    xlabel style={font=\large}, ylabel style={font=\large},
    xticklabel style={font=\normalsize},
    yticklabel style={font=\normalsize},
    title={(b)},
    title style={font=\large\bfseries, yshift=0.3cm},
    legend style={
        at={(0.02,0.98)}, anchor=north west,
        legend cell align=left, font=\scriptsize,
        fill=white, fill opacity=0.95, draw=black!30,
        text opacity=1, rounded corners=2pt,
        inner sep=2pt, nodes={inner sep=1pt},
    },
]

\addplot[domain=0.1:3, samples=50, color=black!60, line width=2pt, forget plot] {2039 * x};
\addplot[domain=0.1:3, samples=50, color=gray!40, line width=0.8pt, dashed, forget plot] {2039 * x * 0.1};
\node[font=\scriptsize, text=gray!60, anchor=south west] at (axis cs:1,220) {10\% eff.};

\addplot[only marks, mark=o, mark size=5pt, color=purple!70, line width=1.2pt]
    coordinates {(0.42, 25)};
\addplot[only marks, mark=pentagon*, mark size=6pt, color=purple!80!black,
         mark options={fill=purple!80!black}]
    coordinates {(0.42, 72)};
\addplot[only marks, mark=star, mark size=7pt, color=purple!50!black, line width=1.4pt]
    coordinates {(0.38, 110)};


\draw[->, thick, purple!55, dashed]
    (axis cs:0.42, 27) -- (axis cs:0.39, 100)
    node[midway, anchor=west, font=\scriptsize, text=purple!55, xshift=3pt]
    {auto-dispatch: $4.5\times$};

\legend{
    BA thread (1 thread / node),
    BA warp (NPB = 8),
    \textbf{BA merge (EPB = 1024)}
}

\end{axis}
\end{tikzpicture}
\caption{Roofline analysis on NVIDIA A100 (peak 19.5 TFLOPS FP32 / 2039~GB/s HBM, ridge $= 19500/2039 \approx 9.56$ FLOPs/byte). In panel (a), kernel fusion shifts the operating point upward by $4\times$; CUDA Graph batching adds $2.8\times$. Enabling mixed-precision storage on top of Fused CG (Section~\ref{sec:mixed_precision}) shifts the point up-and-right by the $\sim\!32\%$ per-step byte reduction: the AI grows from $0.42$ to $\approx 0.62$ (same FLOPs, fewer bytes), and the achieved throughput grows by the measured $1.15\times$ (from 555 to 638 GFLOPS). All configurations lie firmly in the memory-bound regime (left of the ridge), but mixed precision moves the kernel closer to compute-bound territory, which is why the incremental throughput win is less than the $1.32\times$ memory-bytes ratio would predict. Panel (b) shows the CSR-strategy dispatch on BA $m{=}4$ at the same scale.}
\label{fig:roofline}
\end{figure}

\paragraph{Validation scope.}
We validate the fused renewal kernel on SEIR because SEIR exercises all three transition types used by the framework: edge-mediated Markovian $S \to E$, non-Markovian renewal $E \to I$, and non-Markovian renewal $I \to R$. Section~\ref{sec:markov_validation} adds a complementary Markovian validation on the canonical SIS and SIR benchmarks. Extending the renewal kernel to non-Markovian SIS/SIR requires only edits to the compartment transition map; the CSR traversal, hazard kernel, Bernoulli sampler, CUDA Graph capture path, and auto-dispatch logic are all model-agnostic.

\subsection{Markovian engine validation (SIS / SIR)}
\label{sec:markov_validation}

Although the renewal engine is the novel contribution of this paper, the dual-engine framing also requires that the companion Markovian engine (Section~\ref{sec:markov_engine}) is independently validated --- both because Markovian SIS/SIR are the canonical benchmarks for the control-theoretic uses of this framework \citep{nowzari2016analysis, watkins2019robust, preciado2014optimal}, and because the Markovian engine is the reference implementation behind the CPU tau-leaping baseline discussed earlier in this section.

We run the same single-graph validation protocol used for the renewal engine, but against an exact \emph{Doob--Gillespie direct-method} reference (not the generalised non-Markovian Gillespie, which degenerates to Doob--Gillespie on exponential holding times). Both the FlashSpread Markovian engine and the exact simulator are run on an Erd\H{o}s--R\'{e}nyi graph with $N = 10^3$, $d = 8$, $\beta = 0.25$, an initial seed of 10 infected nodes, recovery rate $\delta = \gamma = 0.15$, and 100 independent Monte Carlo runs; the tau-leaping sampler uses \texttt{max\_prob} $= 0.1$ and \texttt{theta} $= 0.01$, which are the defaults provided by the public engine constructor. Figure~\ref{fig:sis_sir_validation} overlays the two simulators on SIS (endemic steady state, left panel) and SIR (single epidemic wave, right panel). On both models the tau-leaping ensemble mean lies inside the exact Gillespie 25--75\% inter-quartile band at every time point of the 100-point sample grid without any further tuning of the sampler.

\begin{figure}[H]
\centering
\includegraphics[width=0.98\columnwidth]{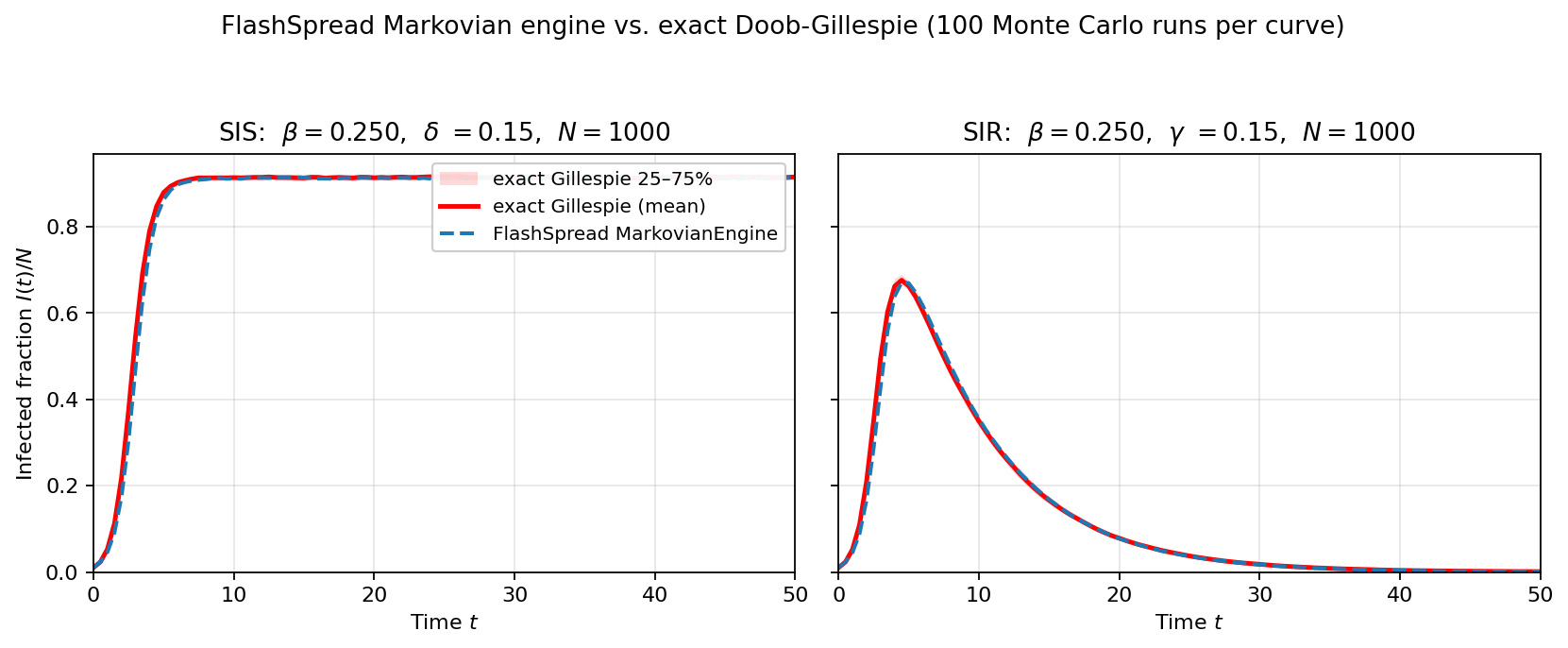}
\caption{\flashspread{} Markovian engine (blue dashed) versus exact Doob--Gillespie (red solid, with 25--75\% quantile band) on ER $d = 8$, $N = 10^3$, $\beta = 0.25$, recovery rate $0.15$. Left: SIS (endemic steady state). Right: SIR (single epidemic wave). 100 independent Monte Carlo runs per curve. The tau-leaping curve tracks the exact reference across the full trajectory on both models, confirming that the Markovian engine's $\varepsilon$-selected step sampler operates correctly in both the sparse-event regime (early-time SIS and the tail of SIR) and the dense-event regime (endemic SIS and the SIR peak).}
\label{fig:sis_sir_validation}
\end{figure}

Appendix~\ref{app:sis_sir} gives the full protocol, reproducibility pointer, and the deeper discussion of the regime separation: SIS tests long-horizon balance between infection and recovery events, SIR tests a single-wave takeoff with an explicit end-state. Combined with the non-Markovian SEIR validation earlier in this section and the multi-topology renewal sweep of Appendix~\ref{app:fidelity_multi}, the framework is validated on all three of SIS, SIR, and SEIR against an exact reference within the relevant regime (Doob--Gillespie for the two Markovian models, generalised non-Markovian Gillespie for SEIR).

\subsection{Reproducibility}
\label{sec:repro}

Every result in this paper is reproducible from the repository at \url{https://github.com/Shakeri-Lab/FlashSpread} against the commit tagged \texttt{jocs-submission}. The reference environment, the per-figure and per-table script map with expected wall-clock budgets on the reference A100 node, and the per-curve CSV sources behind Figure~\ref{fig:throughput_renewal} and Table~\ref{tab:renewal_layered} are all maintained in the repository README (``Reproducing the Paper'' section) so that they stay in sync with the code rather than the manuscript. Random seeds are fixed across scripts, so the JSON and figure artefacts regenerate bitwise-identically on the reference node.

\section{Discussion}\label{sec:discussion}

The dual-engine architecture reflects a fundamental hardware duality (Table~\ref{tab:complexity}). The Markovian engine exploits event-driven sparsity at $\bigO(K \cdot D_{\mathrm{avg}})$; the renewal engine faces inherently dense, memory-bandwidth-limited workloads requiring IO-aware kernel fusion. The appropriate strategy for each regime is fundamentally different, validating the dual-engine design.

\begin{table}[H]
\centering
\caption{Complexity comparison of simulation strategies.}
\label{tab:complexity}
\begin{tabular*}{\linewidth}{@{}lcc@{}}
\toprule
\textbf{Algorithm} & \textbf{Per-Step} & \textbf{Parallel} \\
\midrule
c-GEMF / FastGEMF & $\bigO(D_{\max} \log N)$ & Serial \\
Phantom Process & $\bigO(N + D_{\max})$ & Serial \\
\midrule
Markov (Inertial) & $\bigO(K \cdot D_{\text{avg}} / P)$ & $P \sim 10^4$ \\
Markov (Control) & $\bigO((N + E) / P)$ & $P \sim 10^4$ \\
Renewal (Fused) & $\bigO((N + E) / P)$ & $P \sim 10^4$ \\
\bottomrule
\end{tabular*}
\end{table}

A key conceptual contribution is recognizing that our tau-leaping $\Delta t_{\max}$ serves the same role as the Phantom-Process's phantom rate $\lambda_0$ \citep{vajdi2020stochastic}: both mechanisms force periodic hazard re-evaluation when rates change continuously. The difference is computational paradigm: Phantom Process is exact and serial; tau-leaping is approximate and massively parallel.

The layered optimization decomposition (Table~\ref{tab:renewal_layered}) reveals a striking latency-to-bandwidth crossover. At $N \leq 10^4$, CUDA Graph batching provides $5\times$ while fusion provides only $3.5\times$. At $N \geq 10^5$, the relationship reverses: reducing VRAM traffic via kernel fusion becomes more important than reducing dispatch overhead. When both are combined, execution reaches the irreducible CSR traversal cost of $8E$ bytes per step.

The source-node approximation (Section~\ref{sec:source_node_approx}) effectively acts as a zero-cost abstraction, but only because its overhead is itself fused into the main kernel. A naive implementation that precedes the fused kernel with a dense PyTorch infectivity pre-pass (one $\texttt{torch.special.erfcx}$ call per node, masked to $I$ nodes afterward) increases overall wall-clock time by $45.5\%$ on a regular degree-8 graph with highly coalesced memory access: the dense $\bigO(N)$ sweep through \texttt{erfcx} runs at a point where the SM ALUs are \emph{not} idle, and the ``shadow of memory latency'' argument does not apply. Moving the infectivity update into the fused kernel itself --- each thread writes $\beta \cdot h_{IR}(\tau_i^{\,\prime})$ (or $\beta$ for constant shedding) to the next step's \texttt{infectivity} buffer immediately after computing its new age $\tau_i^{\,\prime}$, guarded by block-scalar sparsity so that blocks with no $I$-successor skip the \texttt{erfcx} entirely --- reduces the overhead to $+1.8\%$ on the regular graph and $+0.1\%$ on a Barab\'{a}si--Albert scale-free graph of the same size and mean degree (Table~\ref{tab:agedep_overhead}). The remaining arithmetic genuinely does execute in the shadow of CSR memory traffic. This is also why the pre-optimization and post-optimization constant-$\beta$ throughputs differ slightly (7.1 vs.\ 7.6 Giga-NUPS on the regular graph): eliminating the per-step pre-pass removes one kernel launch and the associated read-modify-write on the infectivity buffer.

\begin{table}[H]
\caption{Throughput cost of the source-node age-dependent shedding approximation (Fused CG, $b=50$, $N=10^6$, mean degree 8). ``Pre-pass'' runs \texttt{erfcx} over all $N$ nodes in PyTorch before the fused kernel; ``in-kernel'' fuses the infectivity write into the fused kernel's tail so only lanes that will be $I$ next step pay the \texttt{erfcx} cost.}
\label{tab:agedep_overhead}
\centering
\begin{tabular*}{\linewidth}{@{}llccc@{}}
\toprule
\textbf{Graph} & \textbf{Implementation} & \textbf{Constant $\beta$} & \textbf{Age-dep.\ $s(\tau)$} & \textbf{Slowdown} \\
\midrule
Regular ($d=8$)     & Pre-pass   & 7.10 G-NUPS & 4.88 G-NUPS & $+45.5\%$ \\
Regular ($d=8$)     & In-kernel  & 7.61 G-NUPS & 7.48 G-NUPS & $+1.8\%$  \\
\midrule
Barab\'{a}si--Albert ($m=4$) & Pre-pass   & 0.434 G-NUPS & 0.416 G-NUPS & $+4.2\%$ \\
Barab\'{a}si--Albert ($m=4$) & In-kernel  & 0.435 G-NUPS & 0.434 G-NUPS & $+0.1\%$ \\
\bottomrule
\end{tabular*}
\end{table}

Several directions remain open. GPU memory limits network size to $N \lesssim 10^8$; multi-GPU domain decomposition would be needed beyond this scale. The current implementation assumes fixed network topology; temporal networks with independently forming and breaking edges would require per-edge age tracking at $\bigO(E)$ cost, negating the source-node approximation. The auto-dispatch strategy of Section~\ref{sec:csr_dispatch} addresses the classic warp-divergence problem on scale-free graphs ($4.5\times$ over the default kernel on BA at $N = 10^6$) via edge-partitioned merge-based load balancing \citep{merrill2016merge}; a further level would be an intra-warp segmented reduction (using low-level CUDA \texttt{\_\_shfl\_up\_sync} primitives \citep{nvidia2024cudagraphs}) that collapses same-source partial pressures inside each warp so that only the leader and tail lanes of a segment issue an \texttt{atomicAdd} to the shared pressure buffer, eliminating most hub contention. We omitted this optimisation from the public release because it introduces a custom CUDA C++ dependency and a \texttt{torch.utils.cpp\_extension} build system, trading the current pure-Triton / pip-installable distribution for a back-of-envelope $\approx 1.3$--$1.8\times$ speedup in one specific topology regime (BA-hub). The range is an upper-bound estimate, not a measurement: it follows from Amdahl's-style accounting on our observed atomic-serialisation fraction on BA hubs (the merge kernel's \texttt{atomic\_add} into the shared pressure buffer currently sustains 1--2~GB/s on same-destination contention versus the $\sim$1.3~TB/s headroom used elsewhere in the step) combined with the expected fraction of same-source runs per warp produced by the row-sorted CSR chunk. A prototype implementation would be required to confirm it; we have not produced one. The tradeoff is unfavourable for a user-facing framework whose portability (including future AMD ROCm targets, which Triton handles natively) is a first-order design property, so we flag the segmented-reduction path as \emph{known-useful future work} rather than an omitted optimisation. Finally, Appendix~\ref{app:fidelity} characterises the $\varepsilon$-independent structural-bias floor empirically; a formal convergence analysis as a function of $\varepsilon$, network spectral gap, and degree distribution remains open and would tighten the honesty of the current ``does not detectably decrease'' framing into a provable bound.

\section{Conclusion}\label{sec:conclusion}

We presented \flashspread{}, a GPU framework for non-Markovian epidemic simulation on networks. The fused renewal engine consolidates the entire per-step pipeline into a single Triton kernel, reaching 8.09 Giga-NUPS on a $10^6$-node regular graph on an A100 ($217\times$ over the optimised CPU tau-leaping baseline at the same $N$). Degree-aware CSR dispatch --- auto-selected from $D_{\max}/D_{\mathrm{avg}}$ at engine construction --- adds a $4.5\times$ recovery on scale-free graphs without regressing the regular-graph peak (Section~\ref{sec:csr_dispatch}, Appendix~\ref{app:degree_dispatch}). The source-node approximation (Section~\ref{sec:source_node_approx}) makes age-dependent shedding essentially free once the infectivity update is folded into the fused kernel: $\leq 2\%$ overhead on both regular and scale-free graphs (Table~\ref{tab:agedep_overhead}).

Validation against an exact non-Markovian Gillespie reference shows a $\sim$6\% structural-bias floor in peak-$I$ that does not detectably decrease across two decades of tolerance on our benchmarks (Appendix~\ref{app:fidelity}). Combined with the CPU tau-leaping baseline that cleanly isolates the algorithmic relaxation from the hardware speedup, this supports the framework as a rigorous foundation for ensemble forecasting and policy evaluation under biologically realistic non-Markovian dynamics at scales that were previously confined to the CPU era.

\section*{Acknowledgments}

Computational resources were provided by the University of Virginia
Research Computing.

\appendix
\appendix

\section{Triton erfcx approximation}\label{app:erfcx}

Because the GPU kernel compiler (Triton) lacks a native $\mathrm{erfcx}$ instruction, we implemented a piecewise approximation inside the fused kernel that combines a direct identity for $|z| \leq 3.5$ with a four-term asymptotic expansion for $|z| > 3.5$. For $z \geq 0$:
\begin{equation}
    \mathrm{erfcx}(z) \approx \begin{cases}
        e^{z^2}(1 - \mathrm{erf}(z)) & |z| \leq 3.5 \\[4pt]
        \displaystyle\frac{1}{z\sqrt{\pi}}\left(1 - \frac{1}{2z^2} + \frac{3}{4z^4} - \frac{15}{8z^6}\right) & |z| > 3.5
    \end{cases}
\end{equation}
For $z < 0$, we use the identity $\mathrm{erfcx}(z) = 2e^{z^2} - \mathrm{erfcx}(-z)$. For $|z| > 9$, the asymptotic form is applied unconditionally to prevent float32 overflow in the $e^{z^2}$ computation ($e^{81} \approx 5.3 \times 10^{35}$ is safe; $e^{88.7}$ exceeds float32 range). As $\tau \to 0$ after a renewal reset, $z \to -\infty$ and $\mathrm{erfcx}(z) \to \infty$, yielding $h_{\mathrm{LN}} \to 0$, which correctly reflects the biological reality that transition probability is zero immediately after entering a state.

We validate the approximation on a dense fp32-safe grid $z \in [-9, 30]$ (10{,}001 points); the upper negative bound is set by $\exp(z^2)$ overflowing single-precision at $|z| \gtrsim 9.4$, not by a limitation of the formula. The measured maximum relative error is $\sim$$4 \times 10^{-2}$ at $z \approx 3.5$, driven by catastrophic cancellation in the fp32 evaluation of $1 - \mathrm{erf}(z)$ precisely at the branch-switch, and drops to $\sim$$6 \times 10^{-3}$ away from the branch boundary. Because these errors enter the Bernoulli step as $1 - \exp(-\lambda_i \Delta t)$ with $\lambda_i \Delta t \lesssim \varepsilon = 0.03$, the induced bias in transition probability stays well below the structural-bias floor of synchronous tau-leaping itself ($\sim$6--7\%, Section~\ref{sec:validation}: $\sim$6\% peak-$I$ and $\sim$7\% final-$R$), so the approximation is comfortably within the tolerance of the stochastic integrator. A unit test (cited in the repository README) reproduces these bounds and guards against regressions. Random number generation uses Triton's native PRNG, seeded via a 64-bit host-managed step counter that increments by 1 per simulation step, preventing pattern repetition for over $2^{31}$ steps.

\section{Degree-heterogeneous graphs and auto-dispatch}
\label{app:degree_dispatch}

The 1-thread-per-node fused kernel introduced in Section~\ref{sec:fused_kernel} is optimal when the degree distribution is narrow, but degrades sharply on heavy-tailed topologies (scale-free, preferential-attachment, many real social networks). This appendix documents the two alternative CSR traversal kernels we provide, the auto-dispatch heuristic, and the measured gains, all consistent with Section~\ref{sec:csr_dispatch} and Table~\ref{tab:csr_dispatch}.

\subsection{Why the default kernel stalls on hubs}

In the 1-thread-per-node kernel, thread $i$ in a block of 128 threads iterates its own $D_i$ neighbors with a scalar \texttt{while} loop; all 128 threads in the block proceed to the per-node tail only when every thread's \texttt{while} has terminated. The block's wall-clock runtime is therefore $\bigO(\max_{i \in \text{block}} D_i)$. On a Barab\'{a}si--Albert graph with $N = 10^6$ and $m = 4$ the largest measured degree is $D_{\max} = 3{,}870$ ($D_{\max}/D_{\mathrm{avg}} = 484$), so any block containing even one hub stalls for 3{,}870 CSR iterations while 127 other threads have long since finished their degree-4 lists. Memory coalescing is also destroyed: the 128 threads are reading very different \texttt{col\_ind} offsets after the first few iterations. Measured fused-CG throughput collapses from $7.6$ Giga-NUPS on a regular degree-8 graph to $0.45$ Giga-NUPS on BA at the same $N$.

\subsection{Warp-per-node kernel}
\label{app:warp_kernel}

The first remediation cooperates within a warp: thread layout becomes $[\texttt{NODES\_PER\_BLOCK}, \texttt{LANES\_PER\_NODE}]$ with $\texttt{LANES\_PER\_NODE} = 32$, one warp per node. Each warp iterates the node's neighbor list in chunks of 32 edges, using a 2D register tile for partial pressures; at the end a \texttt{tl.sum} along the lane axis collapses the tile to a per-node scalar. Hub traversal now takes $\lceil D_{\max} / 32 \rceil$ warp-parallel iterations rather than $D_{\max}$ serial ones.

The per-node tail (hazard evaluation, Bernoulli sampling, transition, next-infectivity write) re-uses the 1D code path of the baseline kernel on a \texttt{NODES\_PER\_BLOCK}-element tensor; only those few lanes produce meaningful writes, so a moderate amount of block-wide ALU/memory bandwidth is wasted there. The tradeoff is favorable for hub-containing blocks but loses on uniform-degree graphs because the CSR coalescing that the baseline enjoys is given up (most lanes do useless work in any given chunk when $D_i \ll 32$).

We verified bit-identical per-step parity of the warp kernel against the baseline on both a regular degree-8 graph and a BA graph at $N = 10^4$ over 50 steps, using the deterministic \texttt{tl.rand(seed + step\_id, node\_id)} RNG pattern shared across kernels. Zero mismatches were observed in \texttt{state}, \texttt{age}, \texttt{infectivity}, or \texttt{rates} at any step.

\subsection{Edge-partitioned merge-based kernel}
\label{app:merge_kernel}

The second remediation is a merge-path load-balancer along the lines of \citet{merrill2016merge}: each program processes a fixed chunk of $\texttt{EDGES\_PER\_BLOCK}$ contiguous edges (not nodes), independent of which nodes those edges belong to. Each thread in the program handles one edge and performs:
\begin{enumerate}
\item Load $(\texttt{col\_ind}[e], \texttt{weights}[e], \texttt{infectivity}[\texttt{col\_ind}[e]])$.
\item Recover the source node id $n$ such that $\texttt{row\_ptr}[n] \leq e < \texttt{row\_ptr}[n+1]$ via a fully unrolled binary search of depth $\lceil \log_2(N + 2)\rceil$ over \texttt{row\_ptr}.
\item Atomic-add the weighted contribution $\texttt{infectivity}[\texttt{col\_ind}[e]] \cdot \texttt{weights}[e]$ to a pressure scratch buffer at index $n$.
\end{enumerate}
A second lightweight kernel (\texttt{\_flash\_renewal\_tail\_kernel}) then reads the pressure buffer and runs the unchanged per-node tail: hazard, Bernoulli, transition, next-infectivity write. Because the two kernels share the same RNG pattern as the single-kernel variants, they produce statistically identical trajectories.

Load balance across blocks is now perfect regardless of degree skew --- every block does exactly $\texttt{EDGES\_PER\_BLOCK}$ worth of work. The price is (i) one atomic add per edge (contention dominated by hubs, but still much faster than the baseline's serial hub traversal), (ii) a scratch buffer read/write of $4N$ bytes, and (iii) one extra kernel launch per step. The scheme pays off when degree skew is large enough that the baseline's idle-warp waste exceeds these fixed costs.

\subsection{Auto-dispatch}

At engine construction \flashspread{} inspects \texttt{row\_ptr} once, computes $\rho = D_{\max}/D_{\mathrm{avg}}$, and selects:
\begin{equation}
\text{strategy}(\rho) = \begin{cases}
\texttt{thread} & \rho < 4, \\
\texttt{warp}   & 4 \leq \rho < 50, \\
\texttt{merge}  & \rho \geq 50.
\end{cases}
\end{equation}
The thresholds are calibrated from the sweep in Table~\ref{tab:degree_dispatch_sweep}. The user can override via an explicit \texttt{csr\_strategy} argument. The dispatch cost is one pass over \texttt{row\_ptr} at construction ($\bigO(N)$, amortized trivially over every subsequent step).

\subsection{Parity and throughput}

\textbf{Parity.} We ran the three strategies on $N = 10^4$ regular-degree-8 and BA-$m{=}4$ graphs for 50 tau-leaping steps with a deterministic seed. The \texttt{thread} and \texttt{warp} kernels produced bit-identical per-node outputs (\texttt{state}, \texttt{age}, \texttt{infectivity}, \texttt{rates}) at every step. The \texttt{merge} kernel uses non-deterministic atomic-add ordering and therefore can differ by fp rounding at the LSBs of the pressure buffer; over our 50-step window on both graphs the SEIR population counts matched bit-exactly (zero delta in all four bins) because no Bernoulli draw landed within fp ULP of its threshold.

\textbf{Throughput.} Table~\ref{tab:degree_dispatch_sweep} shows the full tunable sweep on both graphs at $N = 10^6$ (Fused CG $b = 50$, A100, same experimental setup as Section~\ref{sec:experiments}).

\begin{table}[H]
\caption{Full CSR-strategy throughput sweep on ER and BA at $N = 10^6$, Fused CG, $b = 50$, A100. NPB = \texttt{NODES\_PER\_BLOCK} (warp); EPB = \texttt{EDGES\_PER\_BLOCK} (merge). Bold entries denote the strategy auto-dispatch picks for that graph.}
\label{tab:degree_dispatch_sweep}
\centering
\begin{tabular*}{\linewidth}{@{}llcc@{}}
\toprule
\textbf{Strategy}             & \textbf{Tunable}   & \textbf{Regular $d=8$} & \textbf{BA $m=4$} \\
                              &                    & G-NUPS               & G-NUPS \\
\midrule
\texttt{thread}               & ---                & \textbf{7.88}        & 0.45 \\
\midrule
\texttt{warp}                 & NPB = 2            & 0.86                 & 0.75 \\
\texttt{warp}                 & NPB = 4            & 1.37                 & 1.16 \\
\texttt{warp}                 & NPB = 8            & 1.70                 & 1.30 \\
\midrule
\texttt{merge}                & EPB = 1{,}024      & 3.92                 & \textbf{2.00} \\
\texttt{merge}                & EPB = 2{,}048      & 3.44                 & 1.78 \\
\texttt{merge}                & EPB = 4{,}096      & 3.35                 & 1.86 \\
\texttt{merge}                & EPB = 8{,}192      & 1.41                 & 0.84 \\
\midrule
\textbf{Auto-dispatch gain}    & vs.\ \texttt{thread} & $1.00\times$       & $\mathbf{4.48\times}$ \\
\bottomrule
\end{tabular*}
\end{table}

Key observations:
\begin{itemize}
\item On the regular graph ($\rho = 2$), every alternative strategy is strictly worse than \texttt{thread}; auto-dispatch correctly chooses \texttt{thread}, so the user pays nothing for having the machinery available.
\item On BA ($\rho = 484$), \texttt{merge} outperforms \texttt{warp} by $\sim$50\% at the best tuning (EPB = 1{,}024 wins over NPB = 8), and both outperform \texttt{thread}.
\item The \texttt{warp} strategy is not currently a winner in the auto-dispatch decision (\texttt{merge} strictly dominates it for $\rho \geq 50$ in our sweep), but we retain it for two reasons: its per-block register-tile accumulation is atomic-free and deterministic bit-to-bit (useful for reproducibility-critical users), and its optimal regime $4 \leq \rho < 50$ covers medium-heterogeneity graphs on which the \texttt{merge} scheme's atomic contention would be unnecessary overhead.
\item EPB = 1{,}024 is the sweet spot on both graphs. Smaller EPB increases launch overhead; EPB = 8{,}192 concentrates too many atomic targets per block and the scheduler can't hide the contention behind other work.
\end{itemize}

Combining the contributions in Section~\ref{sec:renewal_engine} and Section~\ref{sec:csr_dispatch}, the end-to-end speedup chain on BA at $N = 10^6$ is: exact CPU Phantom Process \citep{vajdi2020stochastic} $\approx$ 70 transitions/sec; CPU tau-leaping $\approx$ 37 M-NUPS ($5 \times 10^5\times$ from algorithmic relaxation); 1-thread-per-node fused CG on BA $= 0.45$ G-NUPS ($12\times$ hardware from the naive GPU); auto-dispatched merge-based fused CG on BA $= 2.0$ G-NUPS ($4.5\times$ from degree-aware dispatch). Every factor is strict-hardware (no algorithmic relaxation past the tau-leaping baseline), and the entire chain reproduces end-to-end from the repository; the README names the harness that regenerates these numbers.

\section{Tau-leaping fidelity and compute budgeting}
\label{app:fidelity}

The fused Bernoulli tau-leaping of Algorithm~\ref{alg:fused} introduces two distinct sources of error relative to the exact continuous-time stochastic process on the network: a \emph{discretization} error from finite $\Delta t$ that vanishes as $\varepsilon \to 0$, and a \emph{structural} error from synchronous per-step updates that does not. This appendix quantifies both on the standard SEIR benchmark ($N = 10^3$, ER $d = 8$, log-normal $E\to I$ and $I\to R$, $\beta = 2/d$, 100 runs), maps the fidelity--compute Pareto frontier, and gives a practical rule for choosing $\varepsilon$ under a given compute budget.

\subsection{Two sources of error}

\paragraph{Discretization.} The per-node Bernoulli step probability $p_i = 1 - \exp(-\lambda_i \Delta t)$ is exact only when $\lambda_i$ is constant on $[t, t + \Delta t)$. For non-Markovian hazards $h(\tau)$ the rate varies smoothly within a step; for Markovian rates it changes discretely whenever a neighbor transitions. Both effects contribute an $\bigO(\varepsilon)$ bias in the expected transition count per step. Reducing $\varepsilon$ tightens the integration at the cost of more steps per trajectory, up to the point where per-step cost becomes negligible and only the structural term remains.

\paragraph{Structural.} Within a single tau-leap, all nodes sample Bernoulli trials against their \emph{begin-of-step} rate; transitions are then applied simultaneously. This omits the (non-negligible) event that node $u$ transitions mid-step and alters node $v$'s rate before $v$ samples. Exact Gillespie captures this by ordering events one at a time; tau-leaping does not. The structural bias is therefore a property of synchronous updates on the network and is independent of $\varepsilon$: it vanishes only in the thermodynamic / mean-field limit, or when the per-node rate is so small that at most one transition per Bernoulli window is a valid approximation.

\subsection{Fidelity metrics}

We report four metrics of the ensemble-mean trajectory $(\hat S, \hat E, \hat I, \hat R)(t)$ against the exact Gillespie reference $(S, E, I, R)(t)$:
\begin{itemize}
\item Trajectory $L_\infty$: $\max_{t, s} |\hat y_s(t) - y_s(t)|$ over time and compartment.
\item Trajectory $L_2$: $\sqrt{\overline{(\hat y_s(t) - y_s(t))^2}}$ over the 501-point sample grid.
\item Peak-infection error: $|\max_t \hat I(t) - \max_t I(t)|$.
\item Final-attack-rate error: $|\hat R(T) - R(T)|$ at $T = 50$.
\end{itemize}
For the appendix plots we additionally report a \emph{self-consistency} error against our finest discretization ($\varepsilon = 0.005$, 100 runs) as a proxy for the discretization component alone, isolated from the structural bias.

\subsection{Empirical convergence}

Figure~\ref{fig:fidelity_traj} overlays ensemble-mean trajectories of four $\varepsilon$ values. Qualitatively, $\varepsilon = 0.1$ already tracks the $\varepsilon = 0.005$ reference across all four compartments; reductions below $\varepsilon \approx 0.03$ change the curves only imperceptibly at plot resolution.

\begin{figure}[H]
\centering
\includegraphics[width=0.98\columnwidth]{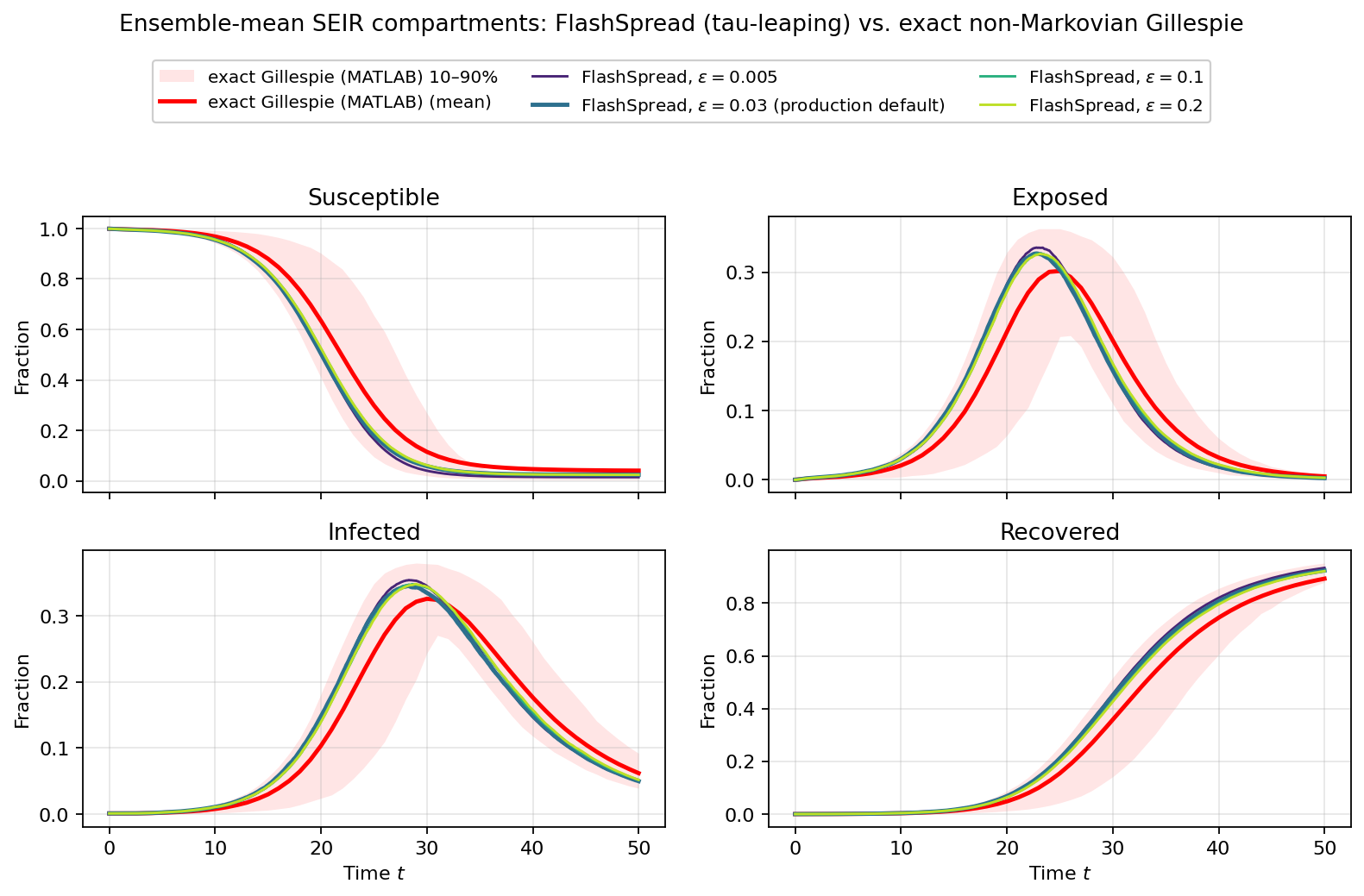}
\caption{Ensemble-mean SEIR trajectories overlaid on the exact non-Markovian Gillespie reference (red solid, with 10--90\% Monte Carlo quantile band), plus four tau-leaping tolerances, on the validation benchmark (Erd\H{o}s--R\'{e}nyi graph with $N = 10^3$ nodes and average degree $d = 8$; $\beta = 2/d = 0.25$; log-normal $E{\to}I$ and $I{\to}R$ with the parameters of Section~\ref{sec:experiments}; 100 independent Monte Carlo trajectories per curve). The exact reference is the companion MATLAB implementation of the non-Markovian Gillespie algorithm \citep{boguna2014simulating}. The $\varepsilon = 0.005$ tau-leaping curve (black) is the practical discretization limit; coarser $\varepsilon$ deviates smoothly but tracks the exact curve tightly. The residual offset between every tau-leaping curve and the exact reference is the structural bias of synchronous Bernoulli updates on a network, independent of $\varepsilon$.}
\label{fig:fidelity_traj}
\end{figure}

Figure~\ref{fig:fidelity_err_eps} shows the quantitative picture on log-log axes, with per-run error bars from a 1{,}000-resample bootstrap over the 100 Monte Carlo trajectories. The peak-$I$ and final-$R$ errors \emph{against exact Gillespie} (solid markers) are statistically indistinguishable across two decades of $\varepsilon$: the 95\% CI envelopes of adjacent $\varepsilon$ values overlap, so the apparent non-monotonicity in the point estimates is ensemble-sampling noise rather than a true convergence pattern. To convert this visual claim into a quantitative one, we fit $\log_{10} \text{err} = \alpha \log_{10} \varepsilon + c$ across the full $\varepsilon \in [0.005, 0.2]$ sweep and bootstrap the slope $\alpha$ over 5{,}000 resamples: the fitted slopes are $\alpha_{\text{peak-}I} = +0.024$ with 95\% CI $[-0.026,\, +0.053]$ and $\alpha_{\text{final-}R} = +0.031$ with 95\% CI $[-0.106,\, +0.116]$. Both intervals contain zero, so the hypothesis of no log-linear decrease of the exact-reference error with $\varepsilon$ is not rejected at the $0.05$ level; if anything, the point estimates are slightly positive, consistent with the structural bias being an $\varepsilon$-independent floor rather than a vanishing error term. In contrast, the self-consistency $L_\infty$ and $L_2$ metrics (dashed, vs the $\varepsilon = 0.005$ ensemble) decrease monotonically with $\varepsilon$: this is the pure discretization component. The separation between the two sets of curves is the \emph{structural bias floor}.

\begin{figure}[H]
\centering
\includegraphics[width=0.98\columnwidth]{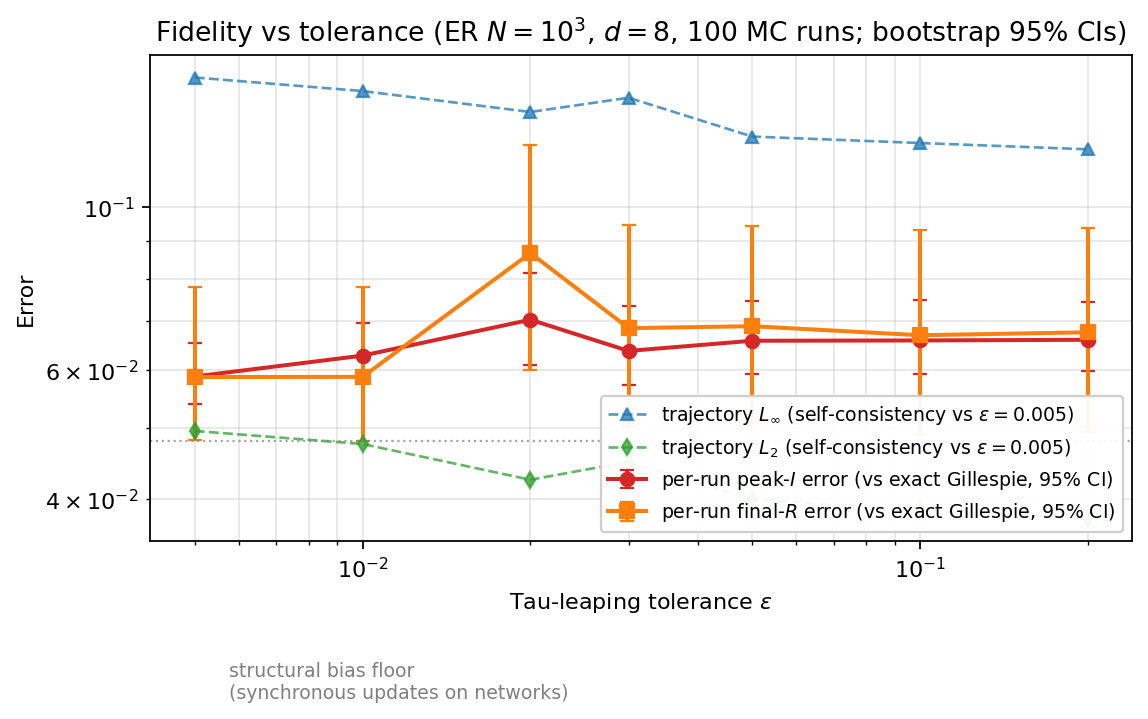}
\caption{Error decomposition for Bernoulli tau-leaping, with bootstrap 95\% confidence intervals over the 100-run ensemble. Solid markers (with error bars) are per-run absolute errors against the exact non-Markovian Gillespie reference produced by the companion MATLAB simulator. Dashed markers are self-consistency trajectory-level errors against the finest $\varepsilon = 0.005$ tau-leaping ensemble. The solid curves' CIs overlap and sit at the structural bias floor (dotted grey), showing that within ensemble noise the peak-$I$ and final-$R$ errors do not detectably decrease below $\sim$6\% as $\varepsilon \to 0$ across the sweep; the dashed curves continue decreasing as pure discretization error vanishes. Below $\varepsilon \approx 0.1$ the structural term dominates and further reductions in $\varepsilon$ buy no detectable fidelity improvement on this benchmark.}
\label{fig:fidelity_err_eps}
\end{figure}

\subsection{Robustness across topologies and sizes}
\label{app:fidelity_multi}

The single-graph result above could in principle be an artefact of the particular sparse ER topology or of the $N = 10^3$ scale. To rule that out, we repeat the sweep over three qualitatively different graph families --- sparse Erd\H{o}s--R\'{e}nyi ($d=8$), scale-free Barab\'{a}si--Albert ($m=4$, average degree $\approx 8$, $D_{\max}$ growing with $N$), and deterministic fixed-degree ($d=8$, homogeneous) --- and three network sizes ($N \in \{10^3, 10^4, 10^5\}$). For each $(G, N, \varepsilon)$ cell we average 20 Monte Carlo trajectories. To keep the ensemble-mean informative, initial seeding is $\max(10,\; 0.01 N)$ Exposed nodes at $t = 0$, well above the stochastic-fadeout threshold on all three topologies; this lets us focus the comparison on the bulk epidemic dynamics rather than the early extinction-prone phase.

Figure~\ref{fig:fidelity_multi} shows the resulting $I(t)/N$ curves. Four properties are robust across the grid:
\begin{itemize}
\item At every $(G, N)$ the $\varepsilon = 0.1$ trajectory lies visually on top of the $\varepsilon = 0.005$ reference; the small gap between them is the discretization component and is insensitive to topology or scale.
\item Epidemic \emph{shape} depends on topology. ER and fixed-degree give similar symmetric peaks near $t \approx 25$; BA's heavy-tailed degree distribution produces an earlier takeoff and a wider peak, because the few high-degree hubs reach saturation faster than the network average.
\item Epidemic \emph{size} depends on $N$. On ER the peak fraction is near-stationary around $0.37$; on BA it drifts slightly with $N$ because the hub fraction changes. The tau-leaping--to--reference agreement holds at every $N$ individually.
\item The agreement-across-$\varepsilon$ story is preserved: the curve at $\varepsilon = 0.1$ tracks the reference to within plot resolution on every panel, which is the main fidelity claim of the appendix, tested against topology and scale simultaneously.
\end{itemize}

\begin{figure}[H]
\centering
\includegraphics[width=0.96\textwidth]{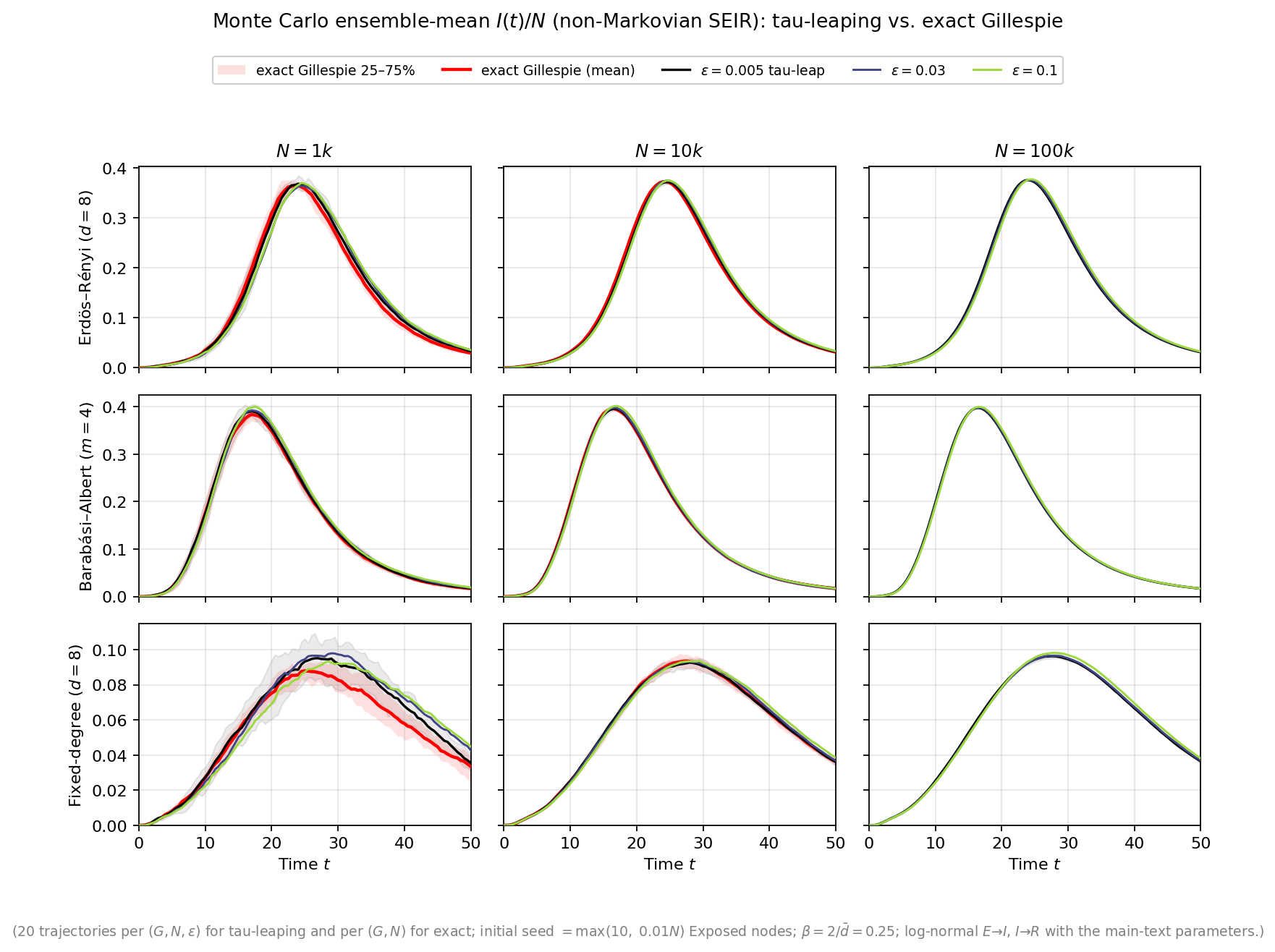}
\caption{Monte Carlo ensemble-mean $I(t)/N$ for the non-Markovian SEIR model (log-normal $E{\to}I$ and $I{\to}R$) across three graph topologies (rows) and three sizes (columns). Each panel overlays an exact non-Markovian Gillespie reference (red, with 25--75\% inter-quartile band shaded; implemented via our companion CPU simulator, cited in the repository README) and tau-leaping at $\varepsilon \in \{0.005, 0.03, 0.1\}$ (the $\varepsilon = 0.005$ curve is the discretization-limit reference in black). Exact reference curves are provided for $N \in \{10^3,\,10^4\}$; $N = 10^5$ is beyond the CPU-Gillespie budget for 20 trials but is covered by tau-leaping alone. Every epidemic takes off (initial seed of $\max(10,\; 0.01 N)$ Exposed nodes) and the tau-leaping curves are indistinguishable from the exact Gillespie curve at plot resolution across all covered cells, demonstrating that the fidelity story established at $N = 10^3$ on ER transfers directly to scale-free and homogeneous topologies and to sizes up to $10^5$.}
\label{fig:fidelity_multi}
\end{figure}

\subsection{Compute--fidelity Pareto}

Figure~\ref{fig:fidelity_pareto} plots mean wall-clock per trajectory against $\max(\text{err peak-}I,\; \text{err final-}R)$ versus exact, with bootstrap 95\% CIs on both axes. Within those CIs the error metric is statistically flat across $\varepsilon \in [0.005, 0.2]$ while wall-clock scales as $\sim 1/\varepsilon$ below the $\tau_{\max}$-clamp region. Any point strictly above the fastest one (top-right of the overlap band) pays compute for no statistically significant fidelity gain, and in particular $\varepsilon = 0.005$ is almost always wasteful on this benchmark.

\begin{figure}[H]
\centering
\includegraphics[width=0.98\columnwidth]{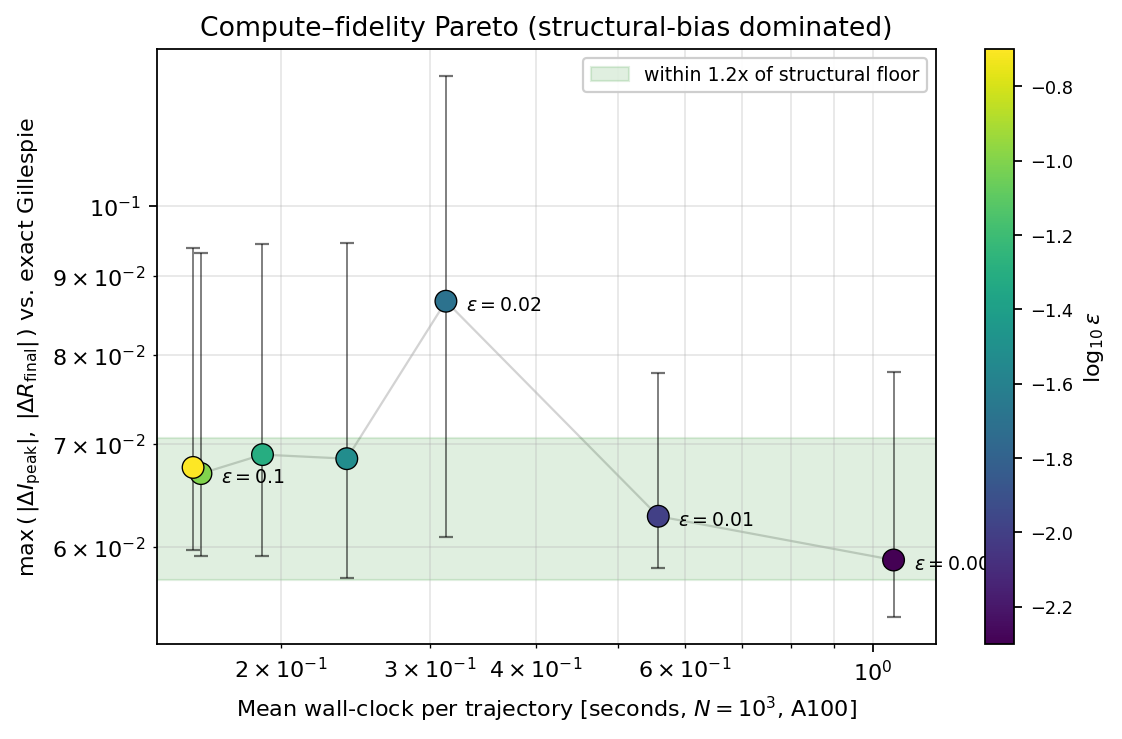}
\caption{Wall-clock--fidelity Pareto. Each point is one $\varepsilon$ value (color: $\log_{10}\varepsilon$, labelled). The vertical axis is $\max(|\Delta I_{\mathrm{peak}}|,\,|\Delta R_{\mathrm{final}}|)$ --- the larger of the two per-run absolute errors against the exact Gillespie reference --- chosen so that no metric can hide a worse metric. Error bars on both axes are 95\% bootstrap CIs (1{,}000 resamples over the 100-run ensemble): horizontal bars for wall-clock, vertical bars for the fidelity metric. The vertical CIs overlap across two decades of $\varepsilon$ while wall-clock spans a factor of seven: above the fast-end cluster, additional compute does not translate into statistically detectable fidelity gains.}
\label{fig:fidelity_pareto}
\end{figure}

\subsection{Tuning rules}

From the Pareto analysis we extract a simple rule of thumb, summarised in Table~\ref{tab:fidelity_budget}. The numbers are calibrated to this benchmark; the underlying logic --- choose the largest $\varepsilon$ that keeps you in the flat portion of the curve --- transfers to other topologies and parameter regimes.

\begin{table}[H]
\caption{Recommended $\varepsilon$ by compute budget and application. All timings are for $N = 10^3$ single-trajectory wall-clock on an A100, from the fidelity sweep plotted in Figures~\ref{fig:fidelity_err_eps}--\ref{fig:fidelity_pareto}. The err peak-$I$ column is the per-run mean against the exact non-Markovian Gillespie reference produced by the companion MATLAB simulator on the same graph and parameters, with the 95\% bootstrap CI in brackets (1{,}000 resamples over 100 Monte Carlo runs). The CI envelopes overlap between adjacent rows: within finite-ensemble noise the err peak-$I$ is flat across two decades of $\varepsilon$ while wall-clock varies by a factor of $\sim$7. Reducing $\varepsilon$ below $\sim$0.03 buys compute, not fidelity, on this benchmark, which is why the main text uses $\varepsilon = 0.03$ as the default; $\varepsilon = 0.1$ is a viable aggressive-throughput alternative when structural bias is known to dominate and per-trajectory fidelity is not the scoring metric (e.g., bulk RL training). Rows are ordered from fastest to slowest.}
\label{tab:fidelity_budget}
\centering
\scriptsize
\setlength{\tabcolsep}{3pt}
\begin{tabular*}{\linewidth}{@{\extracolsep{\fill}}lccc@{}}
\toprule
\textbf{Use case}                       & $\varepsilon$ & \textbf{err peak-$I$ (95\% CI)} & \textbf{wall-clock (95\% CI)} \\
\midrule
Aggressive throughput (RL training)     & 0.2   & 6.60\% [5.98, 7.42] & 0.157~s [0.157, 0.158] \\
Bulk ensemble sampling                  & 0.1   & 6.59\% [5.92, 7.49] & 0.161~s [0.160, 0.161] \\
Default (forecasts, policy evaluation)  & 0.03  & 6.37\% [5.73, 7.33] & 0.239~s [0.236, 0.241] \\
Publication-quality validation          & 0.01  & 6.28\% [5.82, 6.97] & 0.558~s [0.547, 0.565] \\
Over-resolved (typically wasteful)      & 0.005 & 5.88\% [5.40, 6.53] & 1.058~s [1.033, 1.079] \\
\bottomrule
\end{tabular*}
\end{table}

\paragraph{When the rule of thumb breaks.}
The recommendation $\varepsilon = 0.1$ is tied to the observation that structural bias dominates here. Two regimes where the discretization term dominates instead, and where $\varepsilon$ should be reduced towards $0.01$:
\begin{itemize}
\item \emph{Near-mean-field / large-$N$ regimes}, where the structural bias vanishes (Starnini et al., and the correspondence literature discussed in Section~\ref{sec:related}). In these regimes discretization is the only error source and $\varepsilon$ should be chosen directly from the $\bigO(\varepsilon)$ Bernoulli integration bound.
\item \emph{Highly non-Markovian dynamics with sharply peaked hazards}, where the per-step rate change is large (e.g., log-normal shedding profiles with small $\sigma$). Here the begin-of-step rate is a poor approximation of the integrated rate, and a smaller $\varepsilon$ is warranted.
\end{itemize}
A single short run at $\varepsilon \in \{0.1, 0.03\}$ on any new workload is typically enough to tell which regime one is in: if the two agree to within 1\% on the target metric, keep $\varepsilon = 0.1$; otherwise step down to $\varepsilon = 0.03$ and re-check.

The sweep in this appendix is reproducible from the repository; the README lists the data-generation and plotting scripts.

\subsection{Markovian SIS / SIR validation: exact protocol and reproducibility pointer}
\label{app:sis_sir}

Section~\ref{sec:markov_validation} presents the headline result; the IQR-band agreement and regime separation between SIS and SIR are discussed there. This appendix exists only to document the protocol with enough detail to reproduce the run from the public repository, without repeating the result prose in the main text.

\textbf{Parameters.} Erd\H{o}s--R\'{e}nyi graph with $N = 10^3$ nodes and average degree $d = 8$, $\beta = 0.25$, initial seed of 10 infected nodes, recovery rate $0.15$ (either $\delta$ for SIS or $\gamma$ for SIR), 100 independent Monte Carlo runs.

\textbf{Implementations.} The exact reference is a standard Doob--Gillespie direct-method simulator maintained incrementally; the FlashSpread side uses the public \textsc{MarkovianEngine} with its default adaptive tau-leaping sampler at \texttt{max\_prob = 0.1} and \texttt{theta = 0.01}. No parameter tuning beyond those defaults was performed.

\textbf{Reproducibility.} Both simulators, their SLURM wrappers, and the post-processing that produced Figure~\ref{fig:sis_sir_validation} are listed in the ``Reproducing the Paper'' section of the repository README against the \texttt{jocs-submission} tag.

\section{Data layout, FlashNeighbor, and the GPU memory flow}
\label{app:data_flow}

Section~\ref{sec:engines} described the FlashNeighbor kernel at a high level: one thread per node, gather-based CSR traversal, no atomics, a single write-back per thread. This appendix fills in the concrete data layout, the per-thread execution pattern, and the corresponding traffic through the GPU memory hierarchy. These details drive the asymptotic figures quoted in Section~\ref{sec:fused_kernel} ($\sim$64$N$ bytes unfused $\to$ $\sim$20$N$ bytes fused) and the $\bigO((N{+}E)/P)$ per-step complexity claimed throughout the paper.

\subsection{Compressed sparse row (CSR) representation}
\label{app:csr}

\flashspread{} stores the contact network in Compressed Sparse Row (CSR) form, \emph{indexed by incoming edges}, with a mirrored outgoing CSR maintained only for the Markovian Inertial-Mode sparse update path (Section~\ref{sec:markov_engine}). Figure~\ref{fig:csr_layout} traces one node through the three arrays.

\begin{figure}[H]
\centering
\begin{tikzpicture}[
    gnode/.style={circle, draw=graphnode!80!black, line width=0.7pt,
                  fill=graphnode!25, minimum size=0.6cm, font=\small,
                  inner sep=0pt},
    gedge/.style={-, line width=0.6pt, color=black!55},
    cell/.style={rectangle, draw=black!55, line width=0.5pt,
                 minimum width=0.55cm, minimum height=0.55cm,
                 font=\scriptsize, inner sep=1pt, fill=csrcell!10},
    hicell/.style={cell, fill=highlightcell!22, draw=highlightcell!70!black,
                   line width=0.7pt},
    arraylabel/.style={font=\footnotesize\ttfamily},
    indexlabel/.style={font=\tiny, color=black!60},
    pointer/.style={-{Stealth[length=4pt,width=4pt]}, line width=0.6pt,
                    color=accentblue!70!black, dashed},
    caption/.style={font=\footnotesize, align=center},
]

\begin{scope}[shift={(0,0)}]
  \node[gnode] (n0) at (0.00,  0.90) {0};
  \node[gnode] (n1) at (1.30,  1.50) {1};
  \node[gnode, fill=highlightcell!35, draw=highlightcell!70!black, line width=1.1pt]
        (n2) at (1.90,  0.00) {2};
  \node[gnode] (n3) at (0.60, -0.75) {3};
  \node[gnode] (n4) at (2.90,  1.00) {4};

  \draw[gedge] (n0) -- (n1);
  \draw[gedge] (n0) -- (n3);
  \draw[gedge, color=highlightcell!75!black, line width=0.9pt] (n2) -- (n1);
  \draw[gedge, color=highlightcell!75!black, line width=0.9pt] (n2) -- (n3);
  \draw[gedge, color=highlightcell!75!black, line width=0.9pt] (n2) -- (n4);
  \draw[gedge] (n1) -- (n4);

  \node[font=\footnotesize\bfseries, text=black!60] at (1.4, -1.35)
       {contact graph};
\end{scope}

\begin{scope}[shift={(5.2,1.20)}]
  \node[arraylabel, anchor=east] at (-0.25, 0) {\texttt{row\_ptr}:};
  \foreach \v [count=\i from 0] in {0,2,3,6,8,9} {
      \node[cell] (rp\i) at (\i*0.55, 0) {\v};
      \node[indexlabel, above=0pt of rp\i] {\i};
  }
  \draw[hicell, draw=highlightcell!70!black, line width=0.9pt, fill=none]
       (rp2.south west) rectangle (rp2.north east);
  \draw[hicell, draw=highlightcell!70!black, line width=0.9pt, fill=none]
       (rp3.south west) rectangle (rp3.north east);

  \node[font=\scriptsize, text=black!60, anchor=west]
       at (6*0.55, 0) {\,[$N{+}1$ ints]};
\end{scope}

\begin{scope}[shift={(5.2,-0.15)}]
  \node[arraylabel, anchor=east] at (-0.25, 0) {\texttt{col\_ind}:};
  \foreach \v [count=\i from 0] in {1,3, 0, 1,3,4, 0,2, 1} {
      \node[cell] (ci\i) at (\i*0.55, 0) {\v};
      \node[indexlabel, above=0pt of ci\i] {\i};
  }
  \foreach \i in {3,4,5} {
      \draw[hicell, draw=highlightcell!70!black, line width=0.9pt, fill=none]
           (ci\i.south west) rectangle (ci\i.north east);
  }
  \node[font=\scriptsize, text=black!60, anchor=west]
       at (9*0.55, 0) {\,[$E$ ints]};
\end{scope}

\begin{scope}[shift={(5.2,-1.45)}]
  \node[arraylabel, anchor=east] at (-0.25, 0) {\texttt{weights}:};
  \foreach \v/\i in
    {$w_{01}$/0, $w_{03}$/1, $w_{10}$/2,
     $w_{21}$/3, $w_{23}$/4, $w_{24}$/5,
     $w_{30}$/6, $w_{32}$/7, $w_{41}$/8} {
      \node[cell] (w\i) at (\i*0.55, 0) {\v};
  }
  \foreach \i in {3,4,5} {
      \draw[hicell, draw=highlightcell!70!black, line width=0.9pt, fill=none]
           (w\i.south west) rectangle (w\i.north east);
  }
  \node[font=\scriptsize, text=black!60, anchor=west]
       at (9*0.55, 0) {\,[$E$ floats]};
\end{scope}

\draw[pointer, color=highlightcell!70!black]
     (rp2.south) .. controls ++(0,-0.55) and ++(0,0.55) .. (ci3.north);
\draw[pointer, color=highlightcell!70!black]
     (rp3.south) .. controls ++(0,-0.55) and ++(0,0.55) .. (ci6.north);

\node[font=\scriptsize\itshape, text=highlightcell!70!black]
     at (5.2 + 4*0.55, 0.80)
     {\texttt{row\_ptr[2..3]} $\to$ \texttt{col\_ind[3..5]}};

\node[caption, text=black!60] at (1.4, -1.80)
     {incoming CSR: $\mathrm{row\_ptr}[i..i{+}1)$ indexes $i$'s neighbor slice};

\node[caption, text=black!55] at (5.2 + 4*0.55, -2.45)
     {Memory: $4(N{+}1) + 8E$ bytes (int32 $+$ fp32 weights);
      outgoing CSR stored symmetrically for Inertial-Mode propagation.};

\end{tikzpicture}
\caption{CSR layout used throughout \flashspread{}. Left: a small undirected contact graph; node 2 (red) is the focus. Right: the three CSR arrays. \texttt{row\_ptr[i..i+1)} delimits node $i$'s neighbor slice in \texttt{col\_ind} and \texttt{weights}. Reading node 2's neighbors therefore costs two coalesced scalar loads (pointer endpoints) plus a contiguous burst of $D_2$ int/float pairs. Incoming and outgoing CSRs are stored separately: the incoming copy drives both engines' FlashNeighbor reads; the outgoing copy drives Inertial-Mode sparse atomic propagation when a single node transitions.}
\label{fig:csr_layout}
\end{figure}

The layout has three properties that the kernel design exploits:
\begin{itemize}
\item \emph{Contiguity.} A node's neighbors occupy a contiguous slice of \texttt{col\_ind} and \texttt{weights}. A warp of threads processing one node therefore issues a coalesced burst load; a block of threads processing one-node-per-thread produces $B$ scattered bursts, where $B$ is the thread-block size.
\item \emph{Cacheability.} On an A100, L2 is $\sim$40~MB, shared among all SMs. A CSR for $N = 10^6$ and $d = 8$ (\texttt{col\_ind} $+$ \texttt{weights}) totals $64$~MB and no longer fits; the ``L2 cache cliff'' at $N = 10^7$ (Figure~\ref{fig:throughput_renewal}) is this effect. Auto-dispatch does not change the data layout but changes the \emph{access order}, which is what determines whether the L2 miss pattern stays tractable on scale-free graphs (Appendix~\ref{app:degree_dispatch}).
\item \emph{Memory footprint} is $4(N{+}1)$ bytes (\texttt{row\_ptr}, int32) $+\,4E$ (\texttt{col\_ind}, int32) $+\,4E$ (\texttt{weights}, fp32) $=\,8E + 4(N{+}1)$ bytes for the incoming copy, doubled when the outgoing copy is present. At $N = 10^6$, $d = 8$ (directed edge count $E = 8N = 8 \times 10^6$), this is $8 \cdot 8 \times 10^6 + 4 \cdot (10^6{+}1) \approx 68$~MB per copy (dominated by the $8E$ edge payload; we abbreviate as ``$\sim$64~MB'' elsewhere in the text when quoting the $8E$ edge term alone). At $N = 10^7$ the same formula gives $\sim$680~MB per copy, which is where the L2 cache cliff in Figure~\ref{fig:throughput_renewal} sets in. The $N \lesssim 10^8$ single-GPU ceiling quoted in the paper follows directly from doubling this footprint plus $\bigO(N)$ per-node state ($4N + 4N + 4N + 4N$ bytes for \texttt{state}, \texttt{age}, \texttt{infectivity}, \texttt{rates}) and fitting in the A100's 40--80~GB of HBM.
\end{itemize}

\subsection{FlashNeighbor: per-thread execution inside one fused launch}
\label{app:flashneighbor}

The kernel maps one GPU thread to one target node $i$, with no shared memory and no atomics. Figure~\ref{fig:flashneighbor} shows the data flow for a single warp.

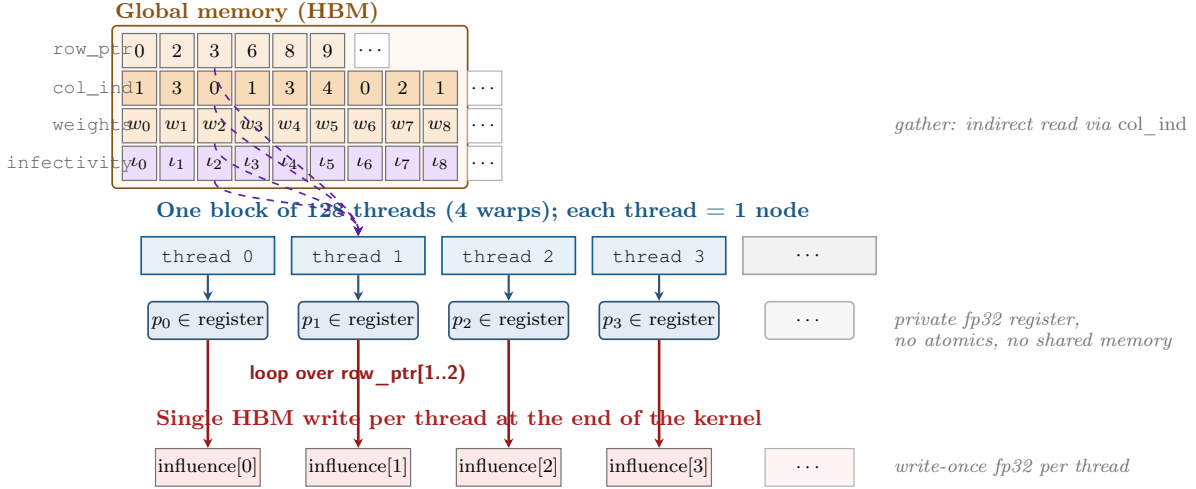
\begin{figure}[H]
\centering
\begin{tikzpicture}[
    lane/.style={rectangle, draw=threadblue!80!black, line width=0.7pt,
                 fill=threadblue!10, minimum width=1.95cm, minimum height=0.55cm,
                 align=center, font=\scriptsize\ttfamily, inner sep=1pt},
    regbox/.style={rectangle, draw=regblue!65!black, line width=0.7pt,
                   rounded corners=2pt, fill=regblue!15,
                   minimum width=1.3cm, minimum height=0.55cm,
                   align=center, font=\scriptsize, inner sep=1.5pt},
    mem/.style={rectangle, draw=black!55, line width=0.5pt,
                fill=globalorange!15, minimum width=0.5cm, minimum height=0.5cm,
                font=\scriptsize, inner sep=1pt},
    flow/.style={-{Stealth[length=4pt,width=4pt]}, line width=0.7pt},
    gatherflow/.style={flow, color=l2purple!70!black, dashed},
    regflow/.style={flow, color=regblue!65!black, line width=0.9pt},
    writeflow/.style={flow, color=accentred!70!black, line width=1.1pt},
    lab/.style={font=\scriptsize, text=black!65, inner sep=1pt},
    annot/.style={font=\scriptsize\itshape, text=black!55},
]

\node[lab, font=\footnotesize\bfseries, text=globalorange!60!black,
      anchor=west]
     at (-0.4, 3.55) {Global memory (HBM)};

\foreach \v [count=\i from 0] in {0,2,3,6,8,9} {
    \node[mem] (rp\i) at (\i*0.55, 3.0) {\v};
}
\node[mem, fill=white, draw=black!30] at (6*0.55 + 0.10, 3.0) {$\cdots$};
\node[lab, anchor=east] at (-0.05, 3.0) {\texttt{row\_ptr}};

\foreach \v [count=\i from 0] in {1,3,0,1,3,4,0,2,1} {
    \node[mem, fill=globalorange!30] (ci\i) at (\i*0.55, 2.45) {\v};
}
\node[mem, fill=white, draw=black!30] at (9*0.55 + 0.10, 2.45) {$\cdots$};
\node[lab, anchor=east] at (-0.05, 2.45) {\texttt{col\_ind}};

\foreach \i in {0,1,2,3,4,5,6,7,8} {
    \node[mem, fill=globalorange!18] (ww\i) at (\i*0.55, 1.90) {$w_{\i}$};
}
\node[mem, fill=white, draw=black!30] at (9*0.55 + 0.10, 1.90) {$\cdots$};
\node[lab, anchor=east] at (-0.05, 1.90) {\texttt{weights}};

\foreach \v [count=\i from 0] in
    {$\iota_{0}$,$\iota_{1}$,$\iota_{2}$,$\iota_{3}$,$\iota_{4}$,
     $\iota_{5}$,$\iota_{6}$,$\iota_{7}$,$\iota_{8}$} {
    \node[mem, fill=l2purple!15] (inf\i) at (\i*0.55, 1.35) {\v};
}
\node[mem, fill=white, draw=black!30] at (9*0.55 + 0.10, 1.35) {$\cdots$};
\node[lab, anchor=east] at (-0.05, 1.35) {\texttt{infectivity}};

\begin{scope}[on background layer]
    \node[fit=(rp0)(ci0)(ww0)(inf0)(inf8),
          draw=globalorange!60!black, line width=0.7pt,
          fill=globalorange!5, rounded corners=2pt,
          inner xsep=4pt, inner ysep=3pt] (hbmgroup) {};
\end{scope}

\node[lab, font=\footnotesize\bfseries, text=threadblue!75!black,
      anchor=west]
     at (0.2, 0.65) {One block of 128 threads (4 warps); each thread = 1 node};

\foreach \tid [count=\k from 0] in {thread 0, thread 1, thread 2, thread 3} {
    \node[lane] (l\k) at (1.0 + \k*2.20, 0.0) {\tid};
}
\node[lane, fill=black!5, draw=black!35] at (1.0 + 4*2.20, 0.0) {$\cdots$};

\foreach \k in {0,1,2,3} {
    \node[regbox] (reg\k) at (1.0 + \k*2.20, -0.95)
         {$p_{\k} \in $ register};
}
\node[regbox, fill=black!3, draw=black!30]
     at (1.0 + 4*2.20, -0.95) {$\cdots$};

\node[lab, text=accentred!70!black, font=\scriptsize\bfseries]
     at (1.0 + 1*2.20, -1.75)
     {\textsf{loop over row\_ptr[1..2)}};

\foreach \src [count=\j from 0] in {ci2, ww2, inf2, ci0} {
}
\draw[gatherflow] (rp2.south) .. controls ++(0, -0.7) and ++(0, 0.55) ..
     (l1.north);
\draw[gatherflow] (ci2.south) .. controls ++(0, -0.45) and ++(0, 0.55) ..
     (l1.north);
\draw[gatherflow] (inf2.south) .. controls ++(0, -0.35) and ++(0, 0.55) ..
     (l1.north);
\draw[gatherflow] (ww2.south) .. controls ++(0, -0.50) and ++(0, 0.55) ..
     (l1.north);

\foreach \k in {0,1,2,3} {
    \draw[regflow] (l\k.south) -- (reg\k.north);
}

\node[lab, font=\footnotesize\bfseries, text=accentred!80!black,
      anchor=west]
     at (0.2, -2.4) {Single HBM write per thread at the end of the kernel};

\foreach \k in {0,1,2,3} {
    \node[mem, fill=accentred!10, minimum width=1.3cm,
          minimum height=0.55cm] (o\k)
         at (1.0 + \k*2.20, -3.10) {$\mathrm{influence}[\k]$};
    \draw[writeflow] (reg\k.south) -- (o\k.north);
}
\node[mem, fill=accentred!5, draw=black!30, minimum width=1.3cm,
      minimum height=0.55cm]
     at (1.0 + 4*2.20, -3.10) {$\cdots$};

\node[annot, anchor=west] at (10.9, 1.9)
     {gather: indirect read via $\mathrm{col\_ind}$};
\node[annot, anchor=west] at (10.9, -0.95)
     {private fp32 register,};
\node[annot, anchor=west] at (10.9, -1.25)
     {no atomics, no shared memory};
\node[annot, anchor=west] at (10.9, -3.10)
     {write-once fp32 per thread};

\end{tikzpicture}
\caption{FlashNeighbor kernel, per-thread view. Each thread owns a unique target node $i$ and accumulates its influence $p_i$ in a private register. The loop body reads \texttt{row\_ptr[i..i+1)} once (coalesced), then iterates the neighbor slice: each iteration issues one indirect gather into \texttt{col\_ind}, one coalesced load of \texttt{weights}, and one indirect gather into the \texttt{infectivity} buffer (the L2 path is the only place these dashed loads land because they are node-indirect). The accumulation $p_i \mathrel{+}= w_{ij}\,\iota_{c_{ij}}$ lives entirely in the thread's register file; the kernel never touches shared memory. At the end of the loop the thread performs exactly \emph{one} fp32 store back to HBM for its node. No atomic operations are required because each target $i$ is uniquely owned.}
\label{fig:flashneighbor}
\end{figure}

Three implementation details make this efficient:
\begin{enumerate}
\item \emph{Write-once output.} The per-thread pressure $p_i$ is a register-resident fp32 accumulator during the whole loop; it materializes in HBM only once, at kernel end. In the unfused pipeline this single register value became a full $\bigO(N)$ intermediate tensor between two kernels.
\item \emph{Indirect gather cost is paid once per edge.} Each edge $(i,j)$ contributes three memory transactions on the read side: \texttt{col\_ind[e]} (coalesced), \texttt{weights[e]} (coalesced), and \texttt{infectivity[col\_ind[e]]} (gather, potentially L2-resident if the neighbor $j$ was recently touched). The gather is the bandwidth-bound operation; the coalesced loads ride in the shadow of its latency.
\item \emph{Block-level sparsity} (Section~\ref{sec:fused_kernel}) is applied \emph{after} the pressure accumulation, as part of the hazard-evaluation stage in the same launch. The reduction $\texttt{any\_e}$, $\texttt{any\_i}$ is a warp shuffle, essentially free, so the kernel never materialises a mask tensor.
\end{enumerate}

In Appendix~\ref{app:degree_dispatch} we replace step 1's ``one thread per node'' by either ``one warp per node'' (for $\rho \geq 4$) or ``one thread per edge'' with binary-search node attribution (for $\rho \geq 50$). The \emph{data layout} is identical in all three cases; only the thread-to-work mapping changes.

\subsection{Flow through the GPU memory hierarchy}
\label{app:memory_flow}

The reason kernel fusion matters on this workload is not FLOP count but HBM traffic. Figure~\ref{fig:memory_flow} contrasts the legacy unfused pipeline (five separate kernels, each materialising an $\bigO(N)$ intermediate tensor to HBM) with the fused kernel (all intermediates stay in SM registers for the lifetime of one per-node pipeline).

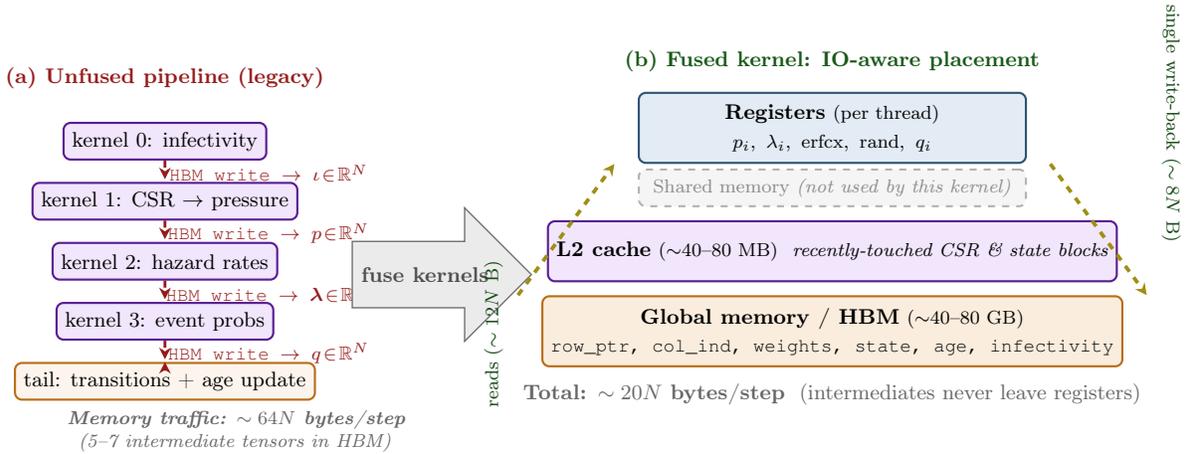
\begin{figure}[H]
\centering
\begin{tikzpicture}[
    level/.style={rectangle, rounded corners=3pt, line width=0.9pt,
                  align=center, font=\footnotesize},
    hbm/.style={level, draw=hbmorange!80!black, fill=hbmorange!15,
                minimum width=5.8cm, minimum height=1.05cm},
    l2/.style={level, draw=l2purple!65!black, fill=l2purple!12,
               minimum width=4.6cm, minimum height=0.9cm},
    sm/.style={level, draw=smgreen!60!black, fill=smgreen!12,
               minimum width=3.4cm, minimum height=0.8cm},
    regs/.style={level, draw=regblue!65!black, fill=regblue!15,
                 minimum width=2.4cm, minimum height=0.7cm},
    datarow/.style={font=\scriptsize\ttfamily, align=left, inner sep=1pt,
                    text=black!70},
    flow/.style={-{Stealth[length=5pt,width=5pt]}, line width=0.9pt,
                 color=black!65},
    badflow/.style={flow, line width=1.2pt, color=accentred!70!black,
                    dashed},
    goodflow/.style={flow, line width=1.2pt, color=smgreen!45!orange},
    annot/.style={font=\scriptsize\itshape, text=black!60},
    sectitle/.style={font=\footnotesize\bfseries},
]

\node[sectitle, text=accentred!70!black]
     at (0.0, 4.85) {(a) Unfused pipeline (legacy)};

\foreach \lab [count=\k from 0] in
   {infectivity, CSR $\to$ pressure, hazard rates, event probs} {
    \node[l2, minimum width=3.0cm, minimum height=0.55cm]
         (kern\k) at (0.0, 3.9 - \k*0.90) {kernel \k: \lab};
}
\node[hbm, minimum width=3.0cm, minimum height=0.55cm, fill=hbmorange!8]
     (endkern) at (0.0, 3.9 - 4*0.90) {tail: transitions + age update};

\foreach \lab [count=\k from 0] in
   {$\iota\!\in\!\mathbb{R}^{N}$, $p\!\in\!\mathbb{R}^{N}$,
    $\bm{\lambda}\!\in\!\mathbb{R}^{N}$, $q\!\in\!\mathbb{R}^{N}$} {
    \draw[badflow]
         ($(kern\k.south) + (0,0)$) -- ($(kern\k.south) + (0,-0.28)$);
    \node[datarow, text=accentred!70!black]
         at ($(kern\k.south) + (1.55, -0.19)$)
         {\texttt{HBM write $\to$ \lab}};
}
\draw[badflow]
     ($(kern3.south)+(0,-0.45)$) -- ($(endkern.north)+(0,0)$);

\node[annot, align=left, anchor=north west]
     at (-1.6, 0)
     {\textbf{Memory traffic: $\sim 64N$ bytes/step}\\
      \;\;(5--7 intermediate tensors in HBM)};

\begin{scope}[shift={(3.9, 1.9)}]
    \node[single arrow, draw=black!55, line width=0.6pt,
          fill=black!8, minimum width=2.0cm, minimum height=0.9cm,
          font=\footnotesize\bfseries, text=black!70]
         {fuse kernels};
\end{scope}

\begin{scope}[shift={(10.0, 0)}]
    \node[sectitle, text=smgreen!55!black]
         at (0.0, 5.1) {(b) Fused kernel: IO-aware placement};

    \def\boxw{5.8cm}

    \node[regs, minimum width=\boxw, minimum height=1.05cm, align=center]
         (lreg) at (0, 4.1)
         {\textbf{Registers} {\scriptsize(per thread)}\\[1pt]
          {\scriptsize\itshape $p_i,\ \lambda_i,\ \mathrm{erfcx},\ \mathrm{rand},\ q_i$}};

    \node[level, draw=black!30, dashed, fill=black!4,
          minimum width=\boxw, minimum height=0.55cm, align=center,
          text=black!45, font=\scriptsize]
         (lsm) at (0, 3.20)
         {Shared memory \textit{(not used by this kernel)}};

    \node[l2, minimum width=\boxw, minimum height=0.9cm, align=center]
         (lL2) at (0, 2.25)
         {\textbf{L2 cache} {\scriptsize($\sim$40--80~MB)}\ \
          {\scriptsize\itshape recently-touched CSR \& state blocks}};

    \node[hbm, minimum width=\boxw, minimum height=1.05cm, align=center]
         (lhbm) at (0, 1.05)
         {\textbf{Global memory / HBM} {\scriptsize($\sim$40--80~GB)}\\[1pt]
          {\scriptsize\ttfamily row\_ptr, col\_ind, weights, state, age, infectivity}};

    \draw[goodflow, line width=1.6pt, dashed, >={Stealth[scale=2]}]
         ([xshift=-0.35cm]lhbm.north west) --
         ([xshift=-0.35cm]lreg.south west);
    \node[annot, text=smgreen!45!black, rotate=90,
          anchor=east, font=\scriptsize]
         at ([xshift=-0.80cm]lL2.west) {reads ($\sim 12N$ B)};

    \draw[goodflow, line width=1.6pt, dashed, >={Stealth[scale=2]}]
         ([xshift=0.35cm]lreg.south east) --
         ([xshift=0.35cm]lhbm.north east);
    \node[annot, text=smgreen!45!black, rotate=-90,
          anchor=east, font=\scriptsize]
         at ([xshift=0.80cm]lL2.east) {single write-back ($\sim 8N$ B)};

    \node[annot, align=center, anchor=north, font=\footnotesize]
         at (0, 0.40)
         {\textbf{Total: $\sim 20N$ bytes/step}\ \
          (intermediates never leave registers)};
\end{scope}

\end{tikzpicture}
\caption{GPU memory-hierarchy flow, before and after fusion. \textbf{(a)} The legacy PyTorch/Triton pipeline launches five kernels per step; each kernel reads from HBM and writes its $\bigO(N)$ intermediate (infectivity, pressure, rates, event probabilities, event mask) back to HBM before the next kernel starts. Total traffic is $\sim$64$N$ bytes per step. \textbf{(b)} The fused kernel keeps all of these intermediates in per-thread registers; only the \emph{inputs} (\texttt{state}, \texttt{age}, \texttt{infectivity} of the current step) travel HBM$\to$L2$\to$register, and only the \emph{final} per-node outputs travel back to HBM. Shared memory is deliberately unused on this kernel: the per-node computation has no cross-thread data dependency inside a block (each thread owns one target node), so staging through shared memory would add latency without reducing traffic. Net traffic drops to $\sim$20$N$ bytes per step, matching the memory-bandwidth-limited throughput reported in Figure~\ref{fig:roofline}.}
\label{fig:memory_flow}
\end{figure}

Three quantitative observations follow:

\textbf{Per-step byte counting.} With all intermediates in registers the dominant HBM traffic per step is $4N$ (\texttt{state} read) $+\,4N$ (\texttt{age} read) $+\,4N$ amortised for the \texttt{infectivity} read (the load is a scatter through \texttt{col\_ind}, so its \emph{HBM} traffic is bounded by $4N$ only when the neighbour working set sits in L2; otherwise L2 misses inflate it up to $4E$ in the worst case) $+\,4N$ (\texttt{next\_state} write) $+\,4N$ (\texttt{next\_age} write), giving $\sim$$20N$ bytes of node-level traffic. CSR data adds $8E$ bytes of edge-level traffic, most of which is served from L2 until the L2 cache cliff is reached. The ratio $20N / (20N + 8E) = 20 / (20 + 8d)$ places node-state traffic at about $24\%$ of the total for $d = 8$; beyond that threshold the L2-serviceable CSR stream dominates.

\textbf{L2 acts as the second-chance cache for neighbour lookups.} The gather \texttt{infectivity[col\_ind[e]]} is a scattered load whose working set is bounded by $N$ fp32 $= 4N$ bytes. At $N = 10^6$ this fits in L2 on any current datacenter GPU; at $N = 10^7$ it begins to spill, explaining the $4.5\times$ throughput drop of the fused engine (and $20\times$ drop of unfused variants) between the two scales.

\textbf{No shared memory pressure.} Because each target node is owned by exactly one thread and because the pressure accumulator is a scalar float, no shared memory is ever allocated by the fused kernel. This leaves the SM's shared memory available to the block scheduler for occupancy improvements; on an A100, increasing block size from 128 to 256 changes occupancy but not working-set size because no shared-memory spill ever occurs. This is part of why the fused kernel reaches 41\% of the reduced memory-bound ceiling (Table~\ref{tab:roofline}) rather than stalling on occupancy: the SM is bandwidth-limited, not pressure-limited.

The block-level skip on \texttt{erfcx} (Section~\ref{sec:fused_kernel}) interacts with this traffic accounting as follows. The reduction that computes \texttt{any\_e} and \texttt{any\_i} is a block-level reduction (cheap in absolute terms compared to the $\sim$$55$-FLOP \texttt{erfcx} evaluation it guards), and when a block's active mask is all zero the block skips the hazard \emph{computation} but not its \emph{data traffic}: the state and age of those nodes were still read from HBM. The $\sim$$4\times$ ``fusion'' gain plotted in Figure~\ref{fig:roofline} therefore decomposes into a $\sim$$3\times$ data-traffic reduction (this appendix) and a residual $\sim$$1.3\times$ from the compute-sparsity of block-level skips on inert blocks --- the two effects are complementary rather than double-counted.

\bibliography{references}

\end{document}